\documentclass[aps,10pt,onecolumn,showpacs,citeautoscript,amsmath,amssymb,floatfix,superscriptaddress,pra]{revtex4-2}

\usepackage{eurosym}
\usepackage{amsfonts}
\usepackage{braket}
\usepackage[utf8]{inputenc}
\usepackage[T1]{fontenc}

\usepackage[most]{tcolorbox}

\usepackage[normalem]{ulem}

\makeatletter

\DeclareFontFamily{U}{MnSymbolF}{}

\DeclareFontShape{U}{MnSymbolF}{m}{n}{
    <-6>  MnSymbolF5
   <6-7>  MnSymbolF6
   <7-8>  MnSymbolF7
   <8-9>  MnSymbolF8
   <9-10> MnSymbolF9
  <10-12> MnSymbolF10
  <12->   MnSymbolF12}{}

\DeclareFontShape{U}{MnSymbolF}{b}{n}{
    <-6>  MnSymbolF-Bold5
   <6-7>  MnSymbolF-Bold6
   <7-8>  MnSymbolF-Bold7
   <8-9>  MnSymbolF-Bold8
   <9-10> MnSymbolF-Bold9
  <10-12> MnSymbolF-Bold10
  <12->   MnSymbolF-Bold12}{}

\DeclareSymbolFont{MnSymbolFonly}{U}{MnSymbolF}{m}{n}
\SetSymbolFont{MnSymbolFonly}{bold}{U}{MnSymbolF}{b}{n}

\DeclareMathSymbol{\tsumint}{\mathop}{MnSymbolFonly}{"6E}
\DeclareMathSymbol{\dsumint}{\mathop}{MnSymbolFonly}{"6F}

\def\sumint{\DOTSI\tsumint\ilimits@}

\makeatother

\newcommand{\aidisclosure}{%
   {\footnotesize(\emph{AI disclosure}: Illustration created with the assistance of ChatGPT 5.5.)}
}

\usepackage[colorlinks=true]{hyperref}
\hypersetup{%
  colorlinks = true,
  linkcolor  = red,
  citecolor = red,
  urlcolor = red
}

\usepackage{cleveref}
    \crefname{figure}{Figure}{figures}
\usepackage{graphicx}

\usepackage[dvipsnames]{xcolor}
\usepackage{soul}
\usepackage{amsthm}
\usepackage{cancel}
\usepackage{comment}
\usepackage{dsfont}
\usepackage{tikz}
\usetikzlibrary{calc}
\usetikzlibrary{positioning}

\usepackage{enumitem}

\newtheorem{theorem}{Theorem}
\newtheorem{lemma}[theorem]{Lemma}

\newtheorem{corollary}[theorem]{Corollary}

\newtheorem{definition}[theorem]{Definition}
\newtheorem{example}[theorem]{Example}
\def\be{\begin{equation}}
\def\ee{\end{equation}}
\def\ba{\begin{eqnarray}}
\def\ea{\end{eqnarray}}

\def\Nl{{\mathchoice
{\setbox0=\hbox{$\displaystyle\rm N$}\hbox{\hbox to0pt
{\kern0.4\wd0\vrule height0.9\ht0\hss}\box0}}
{\setbox0=\hbox{$\textstyle\rm N$}\hbox{\hbox to0pt
{\kern0.4\wd0\vrule height0.9\ht0\hss}\box0}}
{\setbox0=\hbox{$\scriptstyle\rm N$}\hbox{\hbox to0pt
{\kern0.4\wd0\vrule height0.9\ht0\hss}\box0}}
{\setbox0=\hbox{$\scriptscriptstyle\rm N$}\hbox{\hbox to0pt
{\kern0.4\wd0\vrule height0.9\ht0\hss}\box0}}}}
\def\Zl{{\mathchoice
{\setbox0=\hbox{$\displaystyle\rm Z$}\hbox{\hbox to0pt
{\kern0.4\wd0\vrule height0.9\ht0\hss}\box0}}
{\setbox0=\hbox{$\textstyle\rm Z$}\hbox{\hbox to0pt
{\kern0.4\wd0\vrule height0.9\ht0\hss}\box0}}
{\setbox0=\hbox{$\scriptstyle\rm Z$}\hbox{\hbox to0pt
{\kern0.4\wd0\vrule height0.9\ht0\hss}\box0}}
{\setbox0=\hbox{$\scriptscriptstyle\rm Z$}\hbox{\hbox to0pt
{\kern0.4\wd0\vrule height0.9\ht0\hss}\box0}}}}
\def\Ql{{\mathchoice
{\setbox0=\hbox{$\displaystyle\rm Q$}\hbox{\hbox to0pt
{\kern0.4\wd0\vrule height0.9\ht0\hss}\box0}}
{\setbox0=\hbox{$\textstyle\rm Q$}\hbox{\hbox to0pt
{\kern0.4\wd0\vrule height0.9\ht0\hss}\box0}}
{\setbox0=\hbox{$\scriptstyle\rm Q$}\hbox{\hbox to0pt
{\kern0.4\wd0\vrule height0.9\ht0\hss}\box0}}
{\setbox0=\hbox{$\scriptscriptstyle\rm Q$}\hbox{\hbox to0pt
{\kern0.4\wd0\vrule height0.9\ht0\hss}\box0}}}}
\def\Rl{{\mathchoice
{\setbox0=\hbox{$\displaystyle\rm R$}\hbox{\hbox to0pt
{\kern0.4\wd0\vrule height0.9\ht0\hss}\box0}}
{\setbox0=\hbox{$\textstyle\rm R$}\hbox{\hbox to0pt
{\kern0.4\wd0\vrule height0.9\ht0\hss}\box0}}
{\setbox0=\hbox{$\scriptstyle\rm R$}\hbox{\hbox to0pt
{\kern0.4\wd0\vrule height0.9\ht0\hss}\box0}}
{\setbox0=\hbox{$\scriptscriptstyle\rm R$}\hbox{\hbox to0pt
{\kern0.4\wd0\vrule height0.9\ht0\hss}\box0}}}}
\def\Cl{{\mathchoice
{\setbox0=\hbox{$\displaystyle\rm C$}\hbox{\hbox to0pt
{\kern0.4\wd0\vrule height0.9\ht0\hss}\box0}}
{\setbox0=\hbox{$\textstyle\rm C$}\hbox{\hbox to0pt
{\kern0.4\wd0\vrule height0.9\ht0\hss}\box0}}
{\setbox0=\hbox{$\scriptstyle\rm C$}\hbox{\hbox to0pt
{\kern0.4\wd0\vrule height0.9\ht0\hss}\box0}}
{\setbox0=\hbox{$\scriptscriptstyle\rm C$}\hbox{\hbox to0pt
{\kern0.4\wd0\vrule height0.9\ht0\hss}\box0}}}}
\def\Hl{{\mathchoice
{\setbox0=\hbox{$\displaystyle\rm H$}\hbox{\hbox to0pt
{\kern0.4\wd0\vrule height0.9\ht0\hss}\box0}}
{\setbox0=\hbox{$\textstyle\rm H$}\hbox{\hbox to0pt
{\kern0.4\wd0\vrule height0.9\ht0\hss}\box0}}
{\setbox0=\hbox{$\scriptstyle\rm H$}\hbox{\hbox to0pt
{\kern0.4\wd0\vrule height0.9\ht0\hss}\box0}}
{\setbox0=\hbox{$\scriptscriptstyle\rm H$}\hbox{\hbox to0pt
{\kern0.4\wd0\vrule height0.9\ht0\hss}\box0}}}}
\def\Ol{{\mathchoice
{\setbox0=\hbox{$\displaystyle\rm O$}\hbox{\hbox to0pt
{\kern0.4\wd0\vrule height0.9\ht0\hss}\box0}}
{\setbox0=\hbox{$\textstyle\rm O$}\hbox{\hbox to0pt
{\kern0.4\wd0\vrule height0.9\ht0\hss}\box0}}
{\setbox0=\hbox{$\scriptstyle\rm O$}\hbox{\hbox to0pt
{\kern0.4\wd0\vrule height0.9\ht0\hss}\box0}}
{\setbox0=\hbox{$\scriptscriptstyle\rm O$}\hbox{\hbox to0pt
{\kern0.4\wd0\vrule height0.9\ht0\hss}\box0}}}}
\newcommand{\cb}{\mathcal B}

\newcommand{\cale}{\mathcal E}

\newcommand{\cg}{\mathcal G}
\newcommand{\ch}{\mathcal H}
\newcommand{\ci}{\mathcal I}

\newcommand{\cl}{\mathcal L}

\newcommand{\cs}{\mathcal S}

\newcommand{\eqa}{\begin{eqnarray}}
\newcommand{\neqa}{\end{eqnarray}}

\definecolor{myblue}{rgb}{0.2,0.2,0.8}

\usepackage{bm}
\usepackage{bbm}

\newcommand{\ketbra}[2] {
	| #1 \rangle \! \langle #2 |}

\def\dim{{\rm dim}}

\def\phys{{\rm phys}}
\def\kin{{\rm kin}}

\definecolor{darkgreen}{rgb}{0.0, 0.5, 0.13}

\begin{document}

\date{July 27, 2026}

\title{Quantum reference frames beyond subsystems: a reconstruction and generalization of the perspective-neutral framework}

\author{Stefan L. Ludescher}
\email{stefanludescher@msn.com}
\affiliation{Institute for Quantum Optics and Quantum Information,
Austrian Academy of Sciences, Boltzmanngasse 3, A-1090 Vienna, Austria}
\affiliation{Vienna Center for Quantum Science and Technology (VCQ),
Faculty of Physics, University of Vienna, Vienna, Austria}

\author{Manuel Mekonnen}
\email{Manuel.Mekonnen@oeaw.ac.at}
\thanks{\newline SLL and MM contributed equally to this work.}
\affiliation{Institute for Quantum Optics and Quantum Information,
Austrian Academy of Sciences, Boltzmanngasse 3, A-1090 Vienna, Austria}
\affiliation{Vienna Center for Quantum Science and Technology (VCQ),
Faculty of Physics, University of Vienna, Vienna, Austria}

\author{Thomas D. Galley}
\affiliation{Institute for Quantum Optics and Quantum Information,
Austrian Academy of Sciences, Boltzmanngasse 3, A-1090 Vienna, Austria}
\affiliation{Vienna Center for Quantum Science and Technology (VCQ),
Faculty of Physics, University of Vienna, Vienna, Austria}
\affiliation{Université Paris-Saclay, Laboratoire des M{\'e}thodes Formelles, 91190 Gif-sur-Yvette, France}

\author{Markus P. M\"uller}
\affiliation{Institute for Quantum Optics and Quantum Information,
Austrian Academy of Sciences, Boltzmanngasse 3, A-1090 Vienna, Austria}
\affiliation{Vienna Center for Quantum Science and Technology (VCQ),
Faculty of Physics, University of Vienna, Vienna, Austria}
\affiliation{Perimeter Institute for Theoretical Physics,
31 Caroline Street North, Waterloo, Ontario N2L 2Y5, Canada}

\begin{abstract}
We generalize the notion of quantum reference frames (QRFs) to cases where the frame does not necessarily correspond to a tensor factor subsystem, but to a covariant quantum instrument. This unlocks a variety of physical applications: ``frames of labeling'' for indistinguishable particles, suggesting explanations for the symmetrization postulate and the absence of parastatistics, and yielding a transparent description of the entanglement of bosons and fermions; and relational clocks reproducing the Schr\"odinger equation exactly even when all subsystems are interacting or when there are frequency superselection sectors. Our work generalizes the perspective-neutral approach to QRFs pioneered by H\"ohn and co-authors, which we reconstruct from a simple operational scenario. We give a resource-theoretic grounding of this framework, and show how the notion of \textit{completely covariant operations} explains the relevance of the charge-zero sector and the pure-state transformation behavior across perspectives. This also suggests operational clarifications of some aspects of constraint quantization, e.g.\ of the meaning of constraint equations such as $C|\psi\rangle=0$. Some of our results, such as our generalization of relationalization maps to instruments, apply more broadly to other QRF frameworks too, and they contribute to bridging the gap between operational quantum information theory and the internal QRF research program.
\end{abstract}

\maketitle

\section{Introduction}

Given the manifold role of symmetries in physics, it is by now widely recognized that physical quantities acquire their meaning only relative to some frame of reference. But while the role of reference frames in classical physics is rather uncontroversial, the question of how to properly define or interpret their quantum counterparts is much less clear: should quantum reference frames (QRFs) be understood as quantum coordinate systems relative to which other systems are \emph{described}~\cite{AharonovKaufherr,Mazzucchi2001,Angelo2012,Angelo,Pereira2015,pienaar_relational_2016,Giacomini,Vanrietvelde2020,HametteGalley,CastroRuiz2025,Carette2025,devuyst2025relationperspectiveneutralalgebraiceffective,garmier2026perspectivesnonidealquantumreference}? Or should they be regarded as resources in information-processing tasks that enable other systems to be \emph{manipulated}~\cite{AharonovSusskind1967a,WWW,Loveridge2018,Bartlett_2006,poulin_dynamics_2007,boileau_quantum_2008,ahmadi_dynamics_2010,ahmadi_way_2013,popescu_quantum_2018,Skotiniotis2017macroscopic,enk_quantifying_2005,bartlett_quantum_2009,faist_continuous_2020,woods_continuous_2020,hayden_error_2021,yang_optimal_2022,GourSpekkens,GourMarvianSpekkens2009,MarvianSpekkens2014}? Since measuring a quantum system will typically disturb it, these two aspects of QRFs will in general lead to very different behaviors. Moreover, what exactly \textit{is} a QRF? The idea of classical ``rods and clocks'' is usually translated into the quantum domain by demanding that QRFs are ``things'', such as particles or fields, that correspond to subsystems, i.e.\ tensor factors or subalgebras. However, there are many ways -- even classically -- to define a reference frame that does not itself correspond to properties of a given physical object, e.g.\ by the apparent relative positions of distant quasars as in the International Celestial Reference Frame~\cite{Charlot2020ICRF3}, or by the center of mass in a collection of particles. So \emph{what is the most general meaningful way to define a QRF?}

In this work, we shed light on all three questions. We begin with the perspective-neutral (PN) framework~\cite{Vanrietvelde2020,Hamette2021}, which is a paradigmatic example of a ``descriptive'' QRF framework, and give an operational derivation of its basic structures which is resource-theoretic in spirit. This provides a partial unification of the aforementioned two aspects of QRFs. Doing so, we arrive at a generalization of this framework where QRFs need not correspond to subsystems, but to covariant instruments that are not in general tied to subalgebras. Moreover, our results help resolve a puzzle that the PN framework shares with Dirac quantization: how can constraint equations such as $C|\psi\rangle=0$ find a justification from quantum information theory (QIT) if QIT disregards the role of state vectors in favor of density matrices? While the restriction to \emph{covariant operations} in QIT is known to lead to the requirement that the density matrices must commute with the symmetry~\cite{MarvianThesis}, we show that the restriction to a compositional version termed \emph{completely covariant operations} leads to constraint equations on state vectors.\\

\textbf{Relation to previous work.} Historically, these two different roles played by QRFs have been studied separately, focusing either on the idea that (i) they are a resource for information processing tasks, or that (ii) they define  internal perspectives relative to which other systems are described. The absence of an external reference frame induces a symmetry on the quantum-mechanical system of interest. This restricts the set of physically allowed states, transformations and measurements and moreover induces multiple physically equivalent descriptions of these objects. The absence of an external frame thereby motivates the introduction of QRFs: internal degrees of freedom relative to which the other physical degrees of freedom are either (i) manipulated or (ii) described.
Broadly speaking, role (i) emphasizes the operational nature of QRFs, whilst role (ii) is primarily about their perspectival character.

The operational role of QRFs goes back to their introduction in the 1960s as a way of  bypassing superselection rules~\cite{AharonovSusskind1967a,WWW},
forbidding the preparation of states in superpositions of different charges. QRFs are used as a resource to prepare states in a superposition of relative charges (whether this constitutes an actual challenge to the fundamental nature of the superselection rule is still debated~\cite{Loveridge2018}).
QRFs have further been shown to be a resource in a number of information-theoretic tasks in the presence of a global symmetry: implementing quantum operations (including measurements)~\cite{Bartlett_2006,poulin_dynamics_2007,boileau_quantum_2008,ahmadi_dynamics_2010,ahmadi_way_2013,popescu_quantum_2018}, quantum state discrimination~\cite{Skotiniotis2017macroscopic}, communication tasks~\cite{enk_quantifying_2005,bartlett_quantum_2009}, and quantum error correction~\cite{faist_continuous_2020,woods_continuous_2020,hayden_error_2021,yang_optimal_2022}, amongst others. The resource-theoretic aspect of QRFs is formally captured in the resource theory of asymmetry~\cite{GourSpekkens,GourMarvianSpekkens2009,MarvianSpekkens2014}.

The perspectival aspect of QRFs was first described in the 1980's~\cite{AharonovKaufherr} via the introduction of explicit changes of QRF. Since then there have been many proposals within which the perspectival aspect of QRFs play an important role~\cite{Toller1997,Mazzucchi2001,Angelo2012,Angelo,Pereira2015,pienaar_relational_2016,Giacomini,Vanrietvelde2020,HametteGalley,CastroRuiz2025,Carette2025,devuyst2025relationperspectiveneutralalgebraiceffective,garmier2026perspectivesnonidealquantumreference}, and it has been shown to be relevant for a wide range of phenomena, such as the relativity of superposition, entanglement and subsystems~\cite{Giacomini,Ahmad2022,hametteObserverDependentEntropyDiagonal2026}, relational time and dynamics~\cite{Trinity,CastroRuiz2020,Hoehn2020,baumannTimeDelocalizationCausality2026}, and quantum equivalence principles~\cite{GiacominiBrukner2020}. One framework which captures this aspect of QRFs and significantly generalizes Ref.~\cite{AharonovKaufherr} as well as the ``purely perspectival'' approach of Refs.~\cite{Giacomini,HametteGalley}
is the perspective-neutral (PN) framework of H{\"o}hn and collaborators~\cite{Vanrietvelde2020,Hamette2021}.

The PN framework is based on constraint quantization where the presence of a gauge symmetry leads to a redundancy in the description of the physical state~\cite{Vanrietvelde2020}. Different choices of QRF correspond to different ways of removing the redundancy, which induce different descriptions of the physical state, mimicking the structure of reference frames in special relativity~\cite{Hamette2021}. This framework has been applied to a wide range of areas, such as relational dynamics~\cite{PhysRevD.104.066001}, subsystem structure~\cite{Ahmad2022}, gauge theories~\cite{araujoregado2025relationalentanglemententropiesquantum},  quantum gravity~\cite{chen2026quantumreferencefieldstransformations,aguilar-gutierrezRelationalPathIntegral2026} and quantum error correction~\cite{carrozza2025correspondencequantumerrorcorrecting,rothlin2026errorcorrectionlatticequantum,lacambra2026gausslawcodesvacuum}, amongst others. 

The theory of special relativity motivates a unification of both roles of QRFs: a reference frame is a set of rods and clocks, i.e.\  measurement devices which give operational meaning to the coordinates (role (i)) associated with the perspective of an observer (role (ii)). As such, a full quantum generalization of reference frames should capture both aspects (i) and (ii). When combining these aspects into a single unified treatment, two natural approaches suggest themselves: take an operational framework for quantum theory and then ``relationalize'' it, or take a perspectival framework for quantum theory and ``operationalize'' it.

A prominent ``relationalization'' of operational quantum theory is the  operational approach to QRFs by Loveridge and collaborators~\cite{Loveridge2017,Loveridge2018,Loveridge2020,glowacki2023quantumreferenceframesfinite,Carette2025,Jorquera_Riera_2025}. Within the quantum information approach~\cite{Bartlett2007} there also exist ``relationalizations'' of quantum theory using QRFs such as Ref.~\cite{Poulin_2006}. Explicit changes of QRF were introduced in Ref.~\cite{Palmer2014}, where perspectives of agents are constrained by the resources available to them.

In the present work, we adopt the opposite approach. We begin from a framework which naturally captures the concept of perspectives and relationalism, namely the PN framework, and we ``operationalize'' it. More specifically, we provide operational arguments of resource-theoretic nature singling out the physical Hilbert space, a generalization of QRFs from coherent state systems to covariant POVMs, and a further generalization of QRFs to covariant instruments, no longer requiring them to be associated to subsystems.

Some existing work has explored the connection between the PN framework and the information-theoretic approach in specific situations, such as Refs.~\cite{HoehnKrummMueller,KrummHoehnMueller}, where a characterization of QRF changes for finite Abelian groups was given in terms of a game between agents. Explicit models of measurements in the PN approach were also studied in Refs.~\cite{Yang2020switchingquantum,Hausmann2025measurementevents}.  However, to the best of our knowledge, the present work is the first proposal of a systematic operational characterization of the PN approach and its associated constraint equations on state vectors.

While our generalization from coherent state systems to arbitrary POVMs lifts the PN framework to the same level of generality as other approaches to QRFs, our further generalization to covariant instruments seems not to have been considered previously in the literature on any approach to QRFs. The only exception known to us is Ref.~\cite{mekonnen_invariance_2026}, where some of us have studied the special case of instrument QRFs for permutation symmetry to rule out parastatistics. As we show in Definition~\ref{DefRel} and the subsequent lemmas, this generalization is not restricted to the PN framework, but can be applied in other approaches to QRFs, too. In this work, we give several physically relevant applications of this generalization of QRFs beyond subsystems, including the description of indistinguishable particles and relational clocks in worlds where all subsystems are interacting. One motivation to study this generalization comes from the idea that realistic scenarios will typically involve symmetries that do not act independently on subsystems (in particular in quantum gravity~\cite{donnelly_observables_2016}), but another one comes from a lesson learnt from the previous QRF literature: if QRF symmetry implies that the very notion of a subsystem is frame-dependent~\cite{Ahmad2022}, then it seems questionable to tie the definition of a QRF to some subsystem in the first place.

\bigskip
\textbf{Our paper is organized as follows.} In Subsection~\ref{SubsecOpMot}, we provide an operational reconstruction of the PN framework. We introduce the notion of a \emph{completely covariant operation}, explore it mathematically, and demonstrate how it explains the appearance of constraint equations and the physical Hilbert space. We show how ``jumping into a perspective'' and QRF transformations appear naturally as \emph{catalytic completely covariant operations}, which can either be interpreted as isometries or as Kraus operators pertaining to a covariant quantum instrument that measures a QRF.

In Subsection~\ref{SubsecPNMathematical}, we give a thorough mathematical rederivation of the PN framework for the case of compact Lie symmetry groups. It generalizes the PN framework to cases where the QRF is described by a covariant quantum instrument, which will not in general correspond to any quantum subsystem. We generalize several results of the previous PN framework to ours, show that coherently controlled gauge transformations are special cases of QRF transformations, and provide a generalization of relational Dirac observables and the relationalization maps $\$$ and $\yen$ to the instrument case. This generalization holds beyond the PN framework. In Section~\ref{SecRelPrevious}, we explain in detail how the previous PN framework (where QRFs pertain to subsystems) can be obtained as a special case of ours. In Section~\ref{SecExamples}, we give several physically relevant example applications of our framework. We begin in Subsection~\ref{SubsecAbelian} in the simple scenario of finite position and momentum spaces, showing how our framework admits e.g.\ to jump from the perspective of the first particle to that of the center of mass, or into perspectives described by QRF measurements of higher rank, and that our QRFs typically end up entangled with the systems they describe.

In Subsection~\ref{SubsecBosonsFermions}, we turn to the description of bosons and fermions within our framework. We show that the (anti)symmetric subspace appears as the physical Hilbert space under permutation symmetry, showing that whatever motivates the latter also motivates the symmetrization postulate. We give examples of ``jumping into a labeling perspective'', and show how this is related to the problem of entanglement of indistinguishable particles. In Subsection~\ref{SubsecPW}, we show that our generalization applies in interesting ways to the problem of relational quantum clocks. First, we show that it can lead to Page-Wootters clocks that reproduce the Schr\"odinger equation exactly even when all subsystems are interacting; second, we show that a previous ad hoc generalization for clock-constraints with degenerate spectra leading to frequency superselection sectors falls naturally into our framework.

Finally, we conclude in Section~\ref{SecConclusions}, commenting also on the conceptual similarities between our notion of complete covariance and quantum information's notion of complete positivity.
\\

\textbf{Notation.} Physical systems, which are described by complex Hilbert spaces $\mathcal{H}_A,\mathcal{H}_B, \mathcal{H}_C$ etc., will be denoted by upper-case letter $A,B,C$ etc. The algebra of linear operators on a Hilbert space $\ch$ will be denoted $\mathcal{L}(\ch)$, which is the same as the algebra of bounded operators because all Hilbert spaces in this paper are finite-dimensional unless stated otherwise. The upper-case letters will not only denote the physical systems of interest, but also the corresponding algebras of operators, i.e.\ $A=\mathcal{L}(\ch_A)$. This means that all algebras that we consider, and that we denote by capital letters, are full matrix algebras (i.e.\ factors), unless stated otherwise. If we have a linear map $M:W\to X$, where $W$ and $X$ are linear spaces (e.g.\ Hilbert spaces), then we use the notation $M\upharpoonright V$ for the restriction of $M$ to $V$, if $V$ is a linear subspace of $W$. Furthermore, for conjugation with a linear map (in particular, a unitary), we use the superoperator notation $\mathcal{U}(\rho):=U\rho U^\dagger$ for the corresponding quantum operation. The letter $\chi$ will be used for a one-dimensional unitary representation of the compact Lie group $\mathcal{G}$, which is a map $\chi:\cg\to\mathbb{C}$ such that $|\chi(g)|=1$ for all $g\in\cg$ and $\chi(gh)=\chi(g)\chi(h)$. In other works, this is sometimes called a ``linear character''.

Finally, we use the abbreviations ``QRF'' for ``quantum reference frame'' and ``PN'' for ``perspective-neutral''.

\section{Generalization of the Perspective-Neutral Framework}

In this section, we first introduce an operational scenario which motivates not only the perspective-neutral framework of H\"ohn and co-authors~\cite{Vanrietvelde2020,Hamette2021}, but also our generalization where QRFs do not necessarily correspond to subsystems. Then we give a thorough mathematical definition and analysis of this framework, extending several of the results of~\cite{Hamette2021} and other works to the more general framework.

To keep the technicalities to a minimum, our operational motivation in Subsection~\ref{SubsecOpMot} assumes that the symmetry group $\mathcal{G}$ is finite. However, our subsequent mathematical development of the framework in Subsection~\ref{SubsecPNMathematical} applies to all compact Lie groups $\mathcal{G}$ (including finite groups). We will always assume that all Hilbert spaces are finite-dimensional. We believe that our results hold more generally, e.g.\ for unimodular locally compact groups $\mathcal{G}$ acting on separable Hilbert spaces as in~\cite{Hamette2021}, but leave the elaboration of this to future work.

\subsection{Operational motivation of the framework}
\label{SubsecOpMot}
Consider some quantum system $B$ on which some finite symmetry group $\mathcal{G}$ acts. Our operations on the system are \textit{symmetry-constrained}: lacking a frame of reference that would break the symmetry, we can only ever perform symmetric quantum operations. This situation has been extensively studied in the context of a so-called \textit{resource theory of asymmetry}~\cite{Bartlett2007,MarvianThesis}, and we will consider a particular version of this scenario, described below. For our discussion, it is not necessary to be familiar with the notion of resource theory, but it provides a useful conceptual background for understanding the context of our construction.

To be in line with established terminology, let us denote the Hilbert space of the quantum system $B$ by $\mathcal{H}_{\rm kin}$, the \textit{kinematical Hilbert space}. It carries a unitary representation $g\mapsto U_g^B$ of the symmetry group $\mathcal{G}$. We do not at this point commit ourselves to a particular interpretation of $\mathcal{G}$, e.g.\ to whether it should be interpreted as a gauge group or as a symmetry group according to the terminology of~\cite{Hamette2021}, or in some other way. Nonetheless, we will call the elements $U_g^B$ \textit{gauge transformations}, anticipating the interpretation that physical predictions will typically be invariant under the application of those transformations, elaborated in more detail below. The most straightforward implementation of a lack of reference frame that would break the $\mathcal{G}$-symmetry amounts to postulating the following:
\begin{tcolorbox}[
  colback=gray!10,
  colframe=black
]
\textbf{Postulate 1 (tentative).} Only covariant operations can be physically implemented on $B$.
\end{tcolorbox}

By ``operations on $B$'', we mean quantum operations that are either maps from $B$ to $B$, or from the trivial system to $B$, i.e.\ state preparations. At this point, we are not considering measurements. Informally, covariant operations are those that can be applied without any reference frame that breaks the $\mathcal{G}$-symmetry. Formally, a covariant quantum channel $T:X\to Y$ is a completely positive, trace-preserving map with $T\circ\mathcal{U}_g^X=\mathcal{U}_g^Y\circ T$, where $\mathcal{U}_g^X=U_g^X\bullet {U_g^X}^\dagger$ is the associated representation of $\mathcal{G}$ on $X$. If $X$ is the trivial system, corresponding to a one-dimensional Hilbert space, then the associated channel becomes a quantum state $\rho$, and covariance becomes invariance: $\rho=\mathcal{U}_g^Y(\rho)$ for all $g\in\mathcal{G}$. Hence, under our postulate above, we can only ever prepare quantum states that commute with all $U_g^B$, and that are hence block-diagonal in the charge sectors of the representation. Consequently, the accessible observables are those of the same block-diagonal form. These form a subalgebra of $B$, sometimes called the \textit{invariant subalgebra $\mathcal{A}_{\rm inv}$}~\cite{KrummHoehnMueller,HoehnKrummMueller,Doat2025}.

Hence, the operational restriction of the resource theory constrains our access to a subalgebra $\mathcal{A}_{\rm inv}$ of the total algebra of observables $B=\mathcal{L}(\mathcal{H}_{\rm kin})$. As discussed in several recent works~\cite{Doat2025,garmier_perspectives_2025}, the question of which observables are accessible is closely related to the question of which framework of quantum reference frames applies to a given scenario. In contrast to other recent proposals, however, we begin the construction of our scenario by postulating a specific symmetry claim, rather than an algebra of accessible observables.

While the use of covariant operations is standard in the context of the resource theory of asymmetry, and more generally in quantum information theory, we will see that the perspective-neutral approach and the definition of constraints such as $C|\psi\rangle=0$ follows effectively a more ``paranoid'' definition of free operations. As we will explain and motivate below, we will commit ourselves to the following postulate:

\begin{tcolorbox}[
  colback=gray!10,
  colframe=black
]
\textbf{Postulate 1 (final).} Only completely covariant operations can be physically implemented on $B$.
\end{tcolorbox}
Let us first define what this means mathematically, and then analyze this notion operationally, which will in turn give us a concrete motivation for this postulate. To define completely covariant quantum operations, we introduce the notion of an \textit{extension} of a quantum operation. If $T:X\to Y$ is any quantum operation, i.e.\ completely positive, trace-preserving linear map, then an \textit{extension} is a quantum operation $T':X\to YE$, with $E$ some other quantum system, such that $T(\rho)={\rm Tr}_E \left(\strut T'(\rho)\right)$. The well-known Stinespring dilation theorem~\cite{paulsen_completely_2003} guarantees that all quantum operations have isometric extensions. That is, for every quantum operation $T:X\to Y$, there is some quantum system $E$ and an isometry $V:\mathcal{H}_X\to\mathcal{H}_Y\otimes\mathcal{H}_E$ such that $T(\rho)={\rm Tr}_E(V\rho V^\dagger)$. Let us call any extension $V\bullet V^\dagger$, where $V$ is an isometry, a \textit{purification}~\cite{chiribella_theoretical_2009,chiribella_probabilistic_2010}.

There is a specialized version of the dilation theorem for covariant quantum operations~\cite[Theorem 3.3]{keyl_fundamentals_2002}: if $T$ is covariant, then there is a quantum system $E$ with some representation $U_g^E$ of $\mathcal{G}$ and an isometry $V:X\to YE$ such that $T(\rho)={\rm Tr}_E(V\rho V^\dagger)$, and $V U_g^X=(U_g^Y\otimes U_g^E) V$ -- that is, there exists some purification which is itself covariant.

Let us now turn to a stricter requirement. To this end, consider only quantum systems $E$ that carry the trivial representation of $\mathcal{G}$, i.e.\ $U_g^E=\mathds{1}$ for all $g\in\mathcal{G}$. We will call such quantum systems \textit{external} --- these are systems that are not acted on at all by the group $\mathcal{G}$. For example, if $X$ is the spin degree of freedom of an electron (a spin-$1/2$ system), then the group $\mathcal{G}={\rm SU}(2)$ will act nontrivially on it. If $E$ is now a photon (a spin-$1$ system), then this group will \textit{also} act on it non-trivially, and it is hence not external in our terminology. However, if $E$ is some error-corrected logical qubit in a quantum computer, then spatial rotations will typically act trivially on it, which means that $E$ is external. Intuitively, external quantum systems are ones of a completely different type, such that they are unaffected by the symmetries under consideration.

This allows us to define complete covariance:
\begin{definition}
A quantum operation $T:X\to Y$ is completely covariant if every extension $T':X\to YE$ to every external system $E$ is covariant.
\end{definition}
Completely covariant quantum operations are covariant, but the converse is not in general true. In Appendix~\ref{AppProofLemma1}, we show the following characterization:
\begin{lemma}
\label{LemCompleteCovariance}
Let $T:X\to Y$ be a quantum operation. Then the following statements are all equivalent:
\begin{itemize}
    \item[(i)] $T$ is completely covariant, i.e.\ all extensions of $T$ on external systems are covariant.
    \item[(ii)] All purifications of $T$ on external systems are covariant.
    \item[(iii)] There exists a purification of $T$ on an external system which is covariant.
    \item[(iv)] There exists a Kraus representation
    \[
       T(\rho)=\sum_j K_j \rho K_j^\dagger
    \]
    such that all $K_j$ are intertwiners up to a phase $\chi(g)$ given by a one-dimensional representation $\chi$ of $\mathcal{G}$, i.e.\ $K_j U_g^X=\chi(g)U_g^Y K_j$ for all $ g\in\mathcal{G}$ and all $j$.
    \item[(v)] All Kraus representations of $T$ are of the form given in (iv).
\end{itemize}
\end{lemma}
From item (iv), it follows immediately that the composition of two completely covariant operations is again completely covariant. Moreover, if $T$ is an isometry (e.g.\ unitary), then covariance is the same as complete covariance, and the inverse $T^{-1}$ of any covariant isometry $T$ is completely covariant too (see Lemma~\ref{app:CompInvLemma} in Appendix~\ref{AppProofLemma1} for details).

Why are we postulating that only completely covariant operations are allowed on our system $B$ of interest? This requirement will be well-motivated if we have a situation where we think that the symmetry is a strict ``law of nature'': local symmetries $U_g^B$ will preserve all physical predictions on $B$, \textit{including all its correlations with external systems}. That is, ``nothing in the world'' is allowed to change under such transformations --- unless we specifically consider other systems that also carry a representation of $\mathcal{G}$ and co-transform with $B$, i.e.\ that are not external in our sense.

In other words, Postulate 1 is enforced from the following assumptions:
\begin{tcolorbox}[
  colback=gray!10,
  colframe=black
]
\textbf{Motivating assumptions for Postulate 1.} Suppose that
\begin{itemize}
    \item every physically implementable operation on $B$ admits a purification on some external system $E$, and
    \item our ``covariance due to lack of reference frame'' assumption extends to the composite system $BE$.
\end{itemize}
Then Lemma~\ref{LemCompleteCovariance} implies complete covariance of all physically implementable operations on $B$.
\end{tcolorbox}

\begin{figure}[hbt]
\includegraphics[width=\columnwidth]{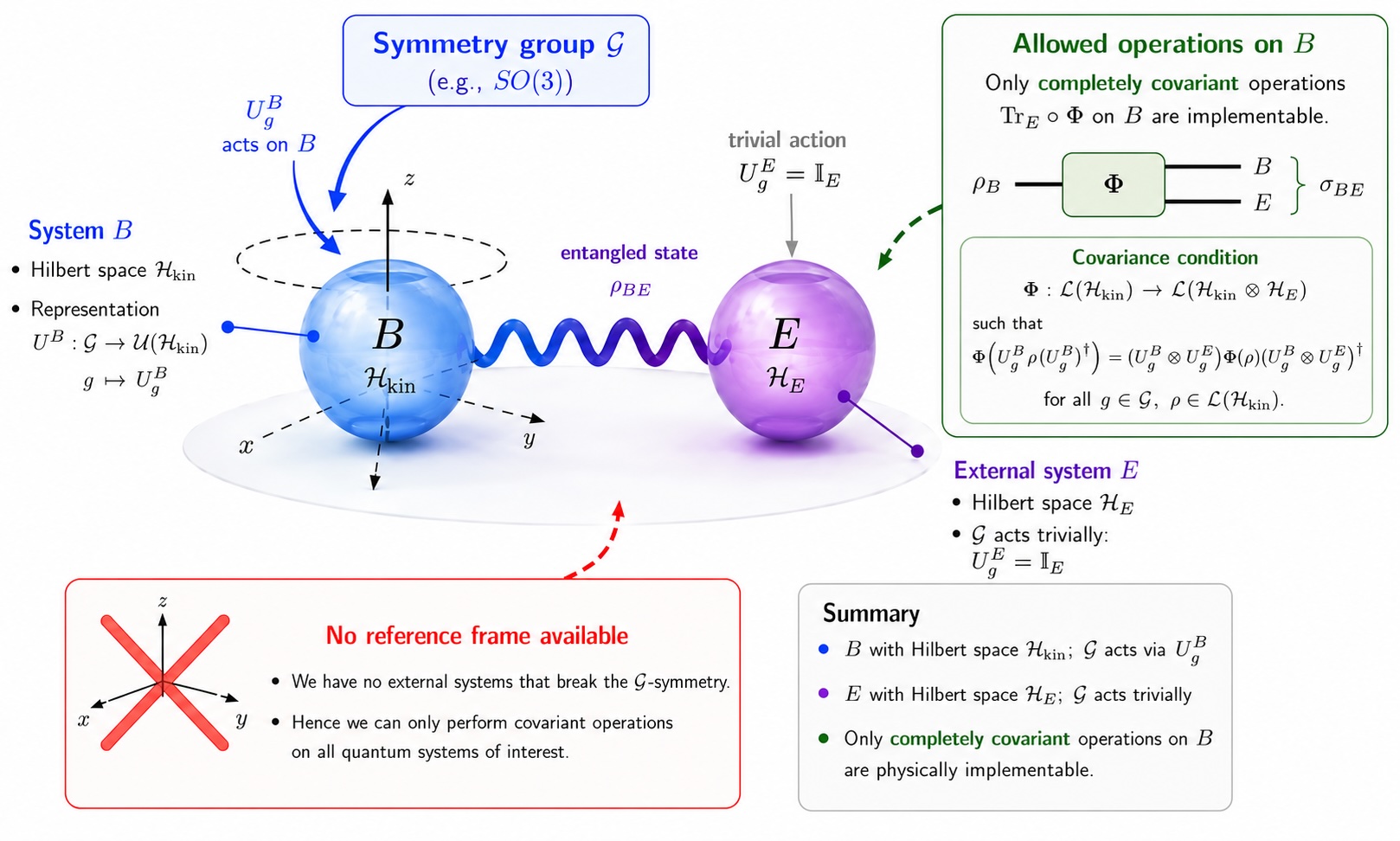}
\caption{First version of our scenario that motivates the PN framework. Due to the purifications of the operations on our system $B$ of interest that are potentially held by some external system $E$, carrying the trivial representation of the symmetry group, the global absence of a frame of reference implies complete covariance for all operations on $B$. \aidisclosure}
\label{fig:pn1}
\end{figure}
Our notion of complete covariance (as a stronger ``compositional'' version of covariance) is inspired by the notion of complete positivity (as a stronger ``compositional'' version of positivity) in quantum information theory. In Outlook and Conclusions (Section~\ref{SecConclusions}), we discuss this analogy in more detail.

Let us now consider the special case of preparation procedures on $B$, i.e.\ quantum operations from the trivial system to $B$. Such operations are quantum states $\rho$, and the notion of complete covariance becomes what we call \textit{complete invariance}: Since $\dim\, \mathcal{H}_A=1$ and $A$ carries the trivial representation, item (iv) of Lemma~\ref{LemCompleteCovariance} implies
\[
   \rho=\sum_j p_j |\psi_j\rangle\langle\psi_j|,\qquad \bar\chi(g)|\psi_j\rangle=U_g^B|\psi_j\rangle \mbox{ for all }g\in\mathcal{G}\mbox{ and all }j.
\]
This means that $\rho$ must have full support on a charge sector that corresponds to a one-dimensional representation $g\mapsto\bar\chi(g)$ of $\mathcal{G}$. This reproves a lemma of~\cite{mekonnen_invariance_2026}, where some of us have used this condition to rule out statistics beyond bosons and fermions for $\mathcal{G}$ the permutation group.
By redefining the representation as $\tilde U_g^B:=\chi(g)U_g^B$ (and omitting the tilde), we can hence restrict our attention to what has been called the \textit{physical Hilbert space}
\[
   \mathcal{H}_{\rm phys}=\{|\psi\rangle\in\mathcal{H}_{\rm kin}\,\,|\,\, U_g^B|\psi\rangle=|\psi\rangle\mbox{ for all }g\in\mathcal{G}\},
\]
and conclude that all physically preparable states have full support on this subspace.

In the special case of the translation group $\mathcal{G}=\mathbb{R}$ (ignoring its non-compactness for the moment), the condition $U_g^B|\psi\rangle=|\psi\rangle$ can be rewritten as $e^{-ia\hat P}|\psi\rangle=|\psi\rangle$ for all $a\in\mathbb{R}$, or as $\hat P|\psi\rangle=0$ for the momentum operator $\hat P$, leading to the interpretation of $\ch_{\rm phys}$ as the ``charge zero sector''. However, simply adapting the representation in a projectively equivalent way as indicated above, via $\tilde U_g^B:=e^{i g P_0} U_g^B$, redefines $\ch_{\rm phys}$ as the subspace of total momentum $P_0$, showing that the PN framework applies to arbitrary total momenta. This simple fact has also been observed in~\cite{hamette_quantum_2025} (see also~\cite{Doat2025}), and it generalizes to all compact Lie groups $\cg$ considered in this work. This does not mean that all values of the total charge would be \emph{physically} equivalent, but it means that they can all be treated \emph{mathematically} with the above definition of $\ch_{\rm phys}$ in the PN framework.

Our notion of complete covariance has been discussed in the literature in a very different context and under a different name: as \textit{strong symmetry}, $U_g\rho=e^{i\theta} \rho$ (as opposed to \textit{weak symmetry} corresponding to the standard notion of covariance; for density matrices, $U_g\rho U_g^\dagger=\rho$). In particular, this notion has appeared in discussions of so-called ``strong-to-weak spontaneous symmetry breaking'' in mixed states of quantum many body systems. In~\cite{lessa_strong--weak_2025}, the intuition behind this definition is described as follows: \textit{``Physically, strong symmetry arises when the system does not exchange charges with the environment [...]. These concepts are deeply rooted in equilibrium statistical mechanics: canonical ensembles have ``strong'' particle-number conservation symmetry, while grand canonical ensembles are only weakly symmetric.''} This intuition resembles a motivation expressed by H\"ohn and co-authors~\cite{HoehnKotechaMele}: that weakly symmetric states are, in a specific sense, \textit{``not fully external frame-independent''}. Our motivation above can be interpreted as a substantial clarification of this statement, saying via Lemma~\ref{LemCompleteCovariance} exactly in what sense frame-dependence would remain if we chose the weak rather than the strong formulation of symmetry.

Let us denote the quantum system that is associated to the physical Hilbert space by $A$; i.e., formally, $A=\mathcal{L}(\mathcal{H}_{\rm phys})$. We have now identified several systems of interest: the ``kinematical'' quantum system $B$, the ``physical'' quantum system $A$ (whose Hilbert space is a subspace of that of $B$), and a hypothetical external system $E$. Our symmetry constraint of Postulate 1 restricts our states and operations to be supported on the physical Hilbert space $\ch_{\rm phys}$. However, we would like to implement the main idea of the internal QRF research program: that it can be interesting for various reasons to redescribe the physical states in a way that formally breaks the symmetry. For example, in the case of translation symmetry of $N$ particles, we might be interested in picking a ``quantum coordinate system'' such that the first particle is located at the origin, and ask what the state relative to this choice of ``gauge fixing'' looks like. Informally, the prescription to do so can be described as follows:

\textbf{Colloquial prescription:} \emph{``Pick some internal structure, and postulate that it points in direction $g$. Work out what the quantum state or observable looks like relative to this convention.''}

Formally, this means that we have maps $M_g$ from the physical Hilbert space to the kinematical Hilbert space, mapping the invariant description of the state or observable to some other, asymmetric description. These maps leave the inner product invariant such that all expectation values are preserved when applied to both states and observables:
\[
   {\rm Tr}(\rho X)={\rm Tr}[(M_g\rho M_g^\dagger) (M_g O M_g^\dagger)]\qquad (\rho,O\in A).
\]
This is the case if and only if $M_g:\ch_{\rm phys}\to\ch_{\rm kin}$ is an isometry, i.e.\ $M_g^\dagger M_g=\mathds{1}_{\ch_{\rm phys}}$. But which isometries $M_g$ should be admissible, and how are the $M_g$ for the different $g\in\mathcal{G}$ related? Moreover, should we think of the $M_g$ as being defined only up to a global phase $M_g\mapsto e^{i\theta_g} M_g$, as the above seems to suggest?

To answer these questions, we have to define what it means that ``some internal structure of $B$ points in direction $g$''. For our scenario described above and depicted in Figure~\ref{fig:pn1}, this notion is undefined: suppose we claim that something on $B$ points in direction $g$, and we apply a gauge transformation $\mathcal{U}_h^B$. Then we should conclude that the structure now points in direction $hg$. However, we have set up our scenario such that gauge transformations do not change any physical predictions; in particular, they leave all physical states invariant. This means that we need to add a further element to our scenario: a reference frame $C$ that allows us to break the symmetry.

We will not think of $C$ as an actually physically realized quantum system. It will be a mere calculation tool, allowing us to determine how we should define a redescription of a physical state $|\psi\rangle\in\ch_{\rm phys}$ on a piece of paper. But we would not like the corresponding map $|\psi\rangle\mapsto M_g|\psi\rangle$ to be an \emph{arbitrary} conventional redefinition, but one that implements our colloquial prescription above. For this reason, we will keep the scenario above fully intact, including its restriction to completely covariant operations. Our minimal change will be to replace the system $B$ by a \textit{composite} system $BC$, where $C$ is now a reference frame that is allowed to start out in an \emph{asymmetric} state. That is, we keep our initial scenario, but supplement it with an additional reference frame $C$.

Thinking of the $M_g$ as mere redescriptions, i.e.\ isometries, suggests that these maps should be invertible via allowed operations in our extended scenario. Similarly, anticipating that ``QRF changes'' between these different perspectives should be physically reversible (unitary), as is indeed the case in the previous PN framework of Ref.~\cite{Hamette2021}, tells us that the reference frame $C$ should not degrade when it is used in an operation.

This leaves us essentially no choice but to define $C$ as an ideal reference frame, with Hilbert space $\mathcal{H}_C:=L^2(\mathcal{G})$. Since $\mathcal{G}$ is finite in this subsection, this is nothing but $\mathbb{C}^{|\mathcal{G}|}$, i.e.\ finite-dimensional and with orthonormal basis $\{|g\rangle\}_{g\in\mathcal{G}}$ such that $\langle g|h\rangle=\delta_{g,h}$. It carries the so-called left-regular representation $U_g^C|h\rangle=|gh\rangle$, and the total action of $\mathcal{G}$ on $BC$ is $U_g^B\otimes U_g^C$. Assuming that it starts in some eigenstate $|g_0\rangle$ allows us to define what it even means that ``some internal structure of $B$ points in direction $g_0$'' -- $C$'s presence gives us now a way to make this claim meaningful. And we assume that this exhausts the role played by $C$: it starts in state $|g_0\rangle$, thereby enabling symmetry-breaking operations on $B$, but it also \emph{ends} in state $|g_0\rangle$, i.e.\ does not take part in the desired operations in any other way.

Requiring that $C$ starts and ends in some fixed state $|g_0\rangle$ means that $C$ behaves like a \emph{catalyst}~\cite{jonathan_entanglement-assisted_1999,lipka-bartosik_catalysis_2024}, yielding a second motivation for the choice of the letter $C$. In the context of resource theories, catalysts are additional quantum systems that are not ``free'' (i.e., in the resource theory of asymmetry, these might be quantum systems in asymmetric states) and that take part in the operations, but that are in the end returned exactly in their initial states, similarly as catalysts in chemistry. Hence, none of their resource content is actually being ``used up''.

In summary, to formalize our colloquial prescription above, and to obtain a physically meaningful definition of ``an internal structure of $B$ pointing in direction $g$'' under our scenario's symmetry claims, we postulate the following:

\begin{figure}[hbt]
\includegraphics[width=\columnwidth]{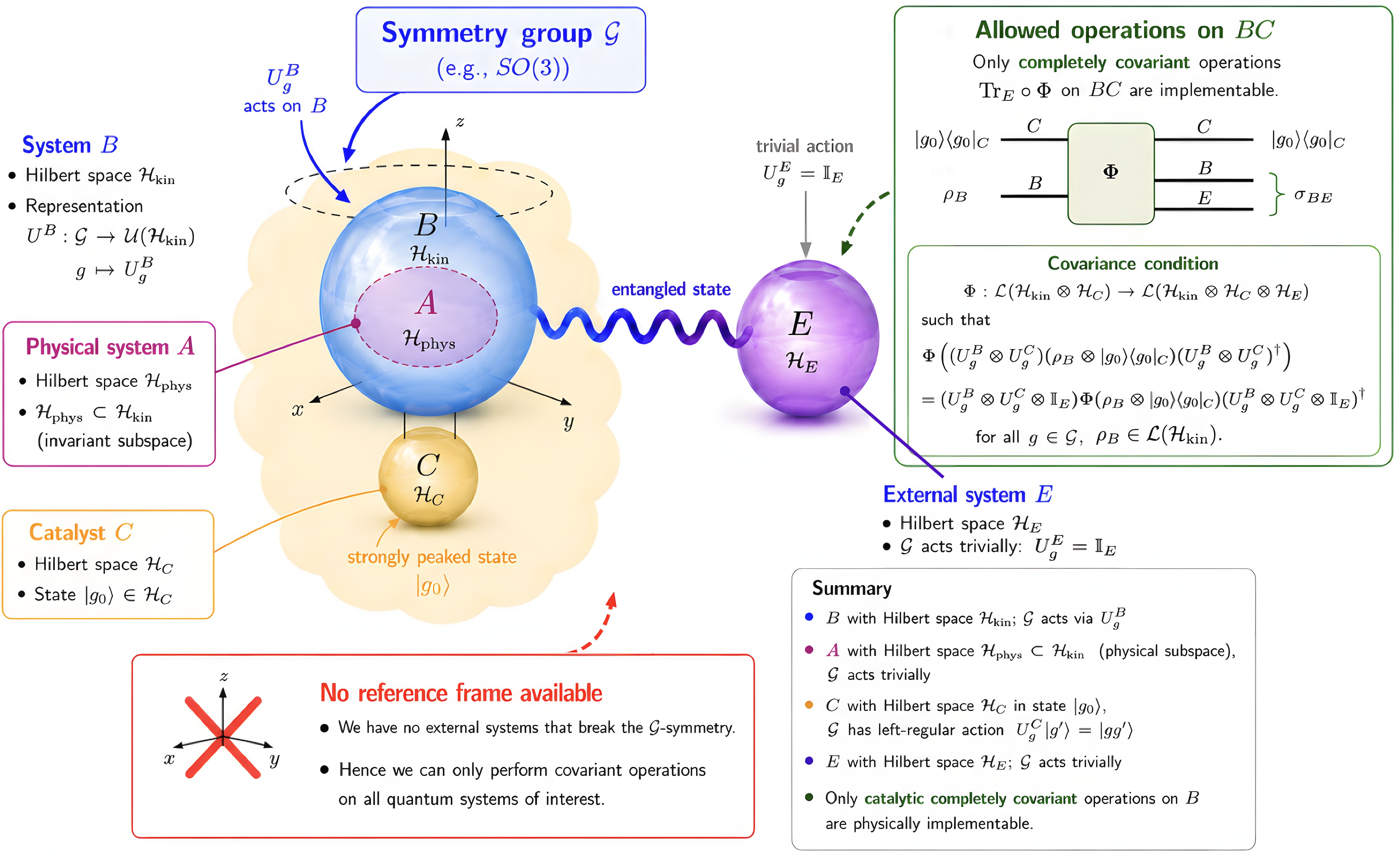}
\caption{Extension of the scenario of Figure~\ref{fig:pn1} for the reconstruction of reduction maps and QRF changes of the PN framework. The only change is that $B$ is replaced by the composite system $BC$, where $C$ is an ideal reference frame. The mere role of $C$ is to define what it means that some structure of $B$ ``points in some direction'', i.e.\ to break the symmetry while catalytically preserving its own local pure state. This allows for the implementation of arbitrary isometries $M_g:\ch_{\rm phys}\to\ch_{\rm kin}$ which play the role of reduction maps, and complete covariance constrains the transformation behavior across perspectives, in particular that of $M_g$ over the different ``directions'' $g\in\cg$. \aidisclosure}
\label{fig:pnextended}
\end{figure}

\begin{tcolorbox}[
  colback=gray!10,
  colframe=black
]
\textbf{Postulate 2.}
{To obtain redescriptions of physical states of $A$ and maps between them, only completely covariant operations from $B'C$ to $BC$ are allowed, where $B'$ is some subspace of $B$. Moreover, they must preserve the reference frame $C$ locally in its pure state $|g_0\rangle$.
}
\end{tcolorbox}

This postulate does not completely determine the mathematical formulation of our colloquial prescription above, but it constrains it, as we will see. We would like to implement an isometry $M_{g_0}:A\to B$ via some completely covariant operation from $AC$ to $BC$. Complete covariance requires that our operation has a covariant purification on some external system; the simplest possibility to ensure this is to have an operation that is itself an isometry. Let us hence attempt to construct an implementation of the form
\begin{equation}
   U\left(|\psi\rangle_A\otimes |g_0\rangle_C\right)=\left(M_{g_0} |\psi\rangle\right)_B\otimes |g_0\rangle_C,
   \label{EqCatalysis}
\end{equation}
where $U:AC\to BC$ is an isometry, i.e.\ $U^\dagger U=\mathds{1}_{AC}$. For $U$ to be covariant, we must have
\[
   U(U_h^A\otimes U_h^C)=\chi(h)(U_h^B\otimes U_h^C)U
\]
for some one-dimensional representation $\chi$ and all $h\in\mathcal{G}$ (note that $U_h^A=\mathds{1}$ for all $h\in\cg$). It is easy to see that the last two equations imply that, for all $h\in\cg$,
\begin{equation}
   U(|\psi\rangle_A\otimes |h g_0\rangle_C)=(M_{h g_0}|\psi\rangle)_B\otimes |h g_0\rangle_C,\quad\mbox{where }M_{hg_0}:=\chi(h)U_h^B M_{g_0}.
   \label{eqFromg0tohg0}
\end{equation}
Therefore, we have
\begin{equation}
   U=\sum_{g\in\mathcal{G}} M_g\otimes |g\rangle\langle g|,
   \label{DefU}
\end{equation}
where every $M_g:A\to B$ is an isometry since $M_{g_0}$ is. Moreover, by the definition of the $M_g$, we have
\begin{equation}
   \chi(h)U_h^B M_g = \chi(h)U_h^B M_{g g_0^{-1}g_0}=\chi(h)U_h^B\chi(g g_0^{-1})U_{g g_0^{-1}}^B M_{g_0}=\chi(hg g_0^{-1})U_{hg g_0^{-1}}^B M_{g_0}=M_{hg} \mbox{ for all }g,h\in\cg.
   \label{eqPhasesMatter}
\end{equation}
This transformation behavior admits a consistent interpretation in the following way. If $C$ is in state $|g_0\rangle$, it enforces the application of $M_{g_0}$ to $|\psi\rangle_A$, yielding a description $|\psi(g_0)\rangle_B:=M_{g_0}|\psi\rangle_A$ of that state such that some internal degree of freedom of $B$ will point in direction $g_0$ in this state description. If $C$ were instead prepared in the state $|hg_0\rangle$, it would enforce the application of $M_{h g_0}$ to $|\psi\rangle_A$. Up to a phase, this is the same as $U_h^B|\psi(g_0)\rangle_B$, meaning that the internal degree of freedom of interest (whatever it was -- it need not correspond to a subsystem) will now (up to some phase factor) point in direction $h g_0$. In other words, the state $|g\rangle$ of $C$ controls the direction into which this internal degree of freedom will point, after the completely covariant isometry $U$ has been applied.

Hence, every $M_g$ can be interpreted as an isometry that yields a state description such that ``some internal structure is pointing in direction $g$'', implementing our colloquial prescription of above. The requirement of complete covariance, which has led us to assume that $U$ is itself a covariant isometry, relates the $M_g$ over different $g\in\cg$, and furthermore tells us that the phases ultimately \emph{do} matter: we cannot redefine $M_g\mapsto e^{i\theta_g} M_g$ arbitrarily without in general breaking the transformation law of Eq.~(\ref{eqPhasesMatter}). Again, requiring complete covariance rather than covariance implies restrictions on the level of state vectors, not only of density matrices.

Indeed, it is easy to see that covariance alone is insufficient to conclude anything about the phases of the $M_g$. For example, we can implement an analogous construction with the covariant channel $T:AC\to BC$ which measures $C$ in its $|g\rangle$-eigenbasis, and depending on the outcome $g$ performs the channel $M_g\bullet M_g^\dagger$. This channel is
\[
   T(\rho_{AC}):=\sum_{g\in\cg} \left(M_g\otimes |g\rangle\langle g|_C\right)\rho_{AC}\left(M_g^\dagger\otimes |g\rangle\langle g|_C\right).
\]
It is easy to check that $T$ is covariant, i.e.\ $T\circ\left(\mathcal{U}_h^A\otimes \mathcal{U}_h^C\right)=\left(\mathcal{U}_h^B\otimes\mathcal{U}_h^C\right)\circ T$, but $T$ is not completely covariant because it violates condition (v) of Lemma~\ref{LemCompleteCovariance}. Now, $T$ is also a covariant channel that yields a redescription of states on $\ch_{\rm phys}$ ``depending on some internal structure pointing in direction $g$'', if $C$ is prepared in the state $|g\rangle\langle g|$. However, this channel is insensitive to the choice of phases: redefining $\tilde M_g:=e^{i\theta_g} M_h$ with arbitrary phases $\theta_g$ preserves $T$. Complete covariance is necessary to obtain restrictions on the phases.

One might be worried that the construction above, together with Postulate 2, allows us only to implement $M_{g_0}$ on $A$, because $C$ is never allowed to be prepared in a state $|g\rangle\neq |g_0\rangle$. However, this can easily be resolved by setting $g_0:=e$ (the unit element of the group), and implementing the completely covariant isometries $U(h):=\sum_{g\in\cg} M_{gh}\otimes |g\rangle\langle g|_C$ instead of $U$. Those maps implement $M_h$ on $A$, for any $h\in\cg$.

So far, we have interpreted the $M_g$ as isometries from $\mathcal{H}_{\rm phys}$ to $\mathcal{H}_{\rm kin}$. However, we can also interpret them as operators that define a \textit{covariant instrument}~\cite{Carmeli_2009}: define the completely positive maps $\mathcal{I}_g:A\to B$ by
\[
   \mathcal{I}_g(\rho):=\frac 1 {|\mathcal{G}|} M_g\rho M_g^\dagger.
\]
then the associated map $\mathcal{I}(\rho):=\sum_{g\in\mathcal{G}} \mathcal{I}_g(\rho)$ is a channel. Hence this defines an instrument: we have a measurement on $A$ that yields some outcome $g\in\mathcal{G}$, and produces the corresponding post-measurement state on $B$,
\[
   \rho'(g):=\frac{\mathcal{I}_g(\rho)}{{\rm Tr}(\mathcal{I}_g(\rho))}=|\mathcal{G}|\,\mathcal{I}_g(\rho)=M_g\rho M_g^\dagger.
\]
All outcomes $g$ have the same probability ${\rm Tr}(\mathcal{I}_g(\rho))=\frac 1 {|\mathcal{G}|}$, and this special property makes the assignment that maps $\rho$ to the post-measurement state $\rho'(g)$, including its renormalization, linear. Moreover, the instrument is covariant, because for all $g,h$, we have
\[
   \mathcal{I}_{gh}(\rho)=\frac 1 {|\mathcal{G}|} M_{gh} \rho M_{gh}^\dagger =\frac 1 {|\mathcal{G}|} U_g^B M_h \rho M_h^\dagger {U_g^B}^\dagger=U_g^B \mathcal{I}_h({U_g^A}^\dagger \rho U_g^A){U_g^B}^\dagger.
\]
This holds because $U_g^A=\mathds{1}$ for all $g$. We can hence interpret this instrument as \textit{asking whether some internal degree of freedom of $A$ ``points in direction $g$''}, and -- depending on the answer -- yielding the post-measurement state $M_g\rho M_g^\dagger$, i.e.\ implementing the isometry $M_g$. Hence, $M_g$ can also be interpreted as ``jumping into a particular internal perspective on $A$''.

For example, if we have several particles and if $\mathcal{G}$ is the translation group (or a finite discrete version thereof, as in~\cite{KrummHoehnMueller,HoehnKrummMueller}), then we can ask: ``at which position $g$ is the first particle located?'' If we find outcome $g=0$, then the resulting post-measurement state describes the positions of all remaining particles relative to the first particle -- we have ``jumped into the perspective which describes everything in relation to the first particle''.

If we interpret the isometries $M_g$ as \textit{jumping into a particular perspective}, and if $M'_{g'}$ describes the \textit{jump into another perspective} of this kind, then we can compose these maps to obtain  a corresponding \textit{QRF transformation} from one perspective to the other, via $M'_{g'}M^\dagger_g$, restricted to the image of $M_g$. This is depicted in Fig.~\ref{fig:qrf-change} below, see also Definition~\ref{DefQRFTrafo}. Since compositions of completely covariant operations are completely covariant, this means that we can understand these QRF transformations as \textit{catalytic completely covariant operations}, which is again the reason for obtaining maps on the level of state vectors rather than density matrices.

The fact that all $g$ have the same outcome probability may seem puzzling at first: it seems as if nothing is really measured, and a completely random outcome is produced. This is the case because we have here restricted our attention to $\ch_{\rm phys}$ where all states are invariant. It is sometimes more natural to define the instrument on all of the kinematical Hilbert space where the interpretation as a meaningful quantum measurement is more obvious, and to obtain $\mathcal{I}$ as its restriction to the physical Hilbert space. Indeed, in Definition~\ref{DefRel} and the subsequent lemmas, we will consider more general covariant instruments on all of $\ch_{\rm kin}$ (as other approaches to QRFs would), and show how they can be interpreted as QRFs yielding a corresponding relationalization map with all the properties that we would naturally expect it to have.

Calling $C$ a ``catalyst'' is true in a technical sense, but should not be misinterpreted as a deep resource-theoretic statement. Indeed, $C$ is an ideal reference frame, and it is well-known that ideal frames in perfect eigenstates $|g\rangle$ allow us to implement \textit{arbitrary} local quantum channels in a globally covariant way. Intuitively, $C$ contains an infinite amount of asymmetry, and any finite amount that is used up in a process does not degrade this reference frame. In the case of connected (rather than finite) Lie groups $\cg$, well-known no-go theorems~\cite{janzing_quasi-order_2003,lostaglio_coherence_2019,marvian_no-broadcasting_2019} show that the sort of catalysis that we have constructed above is impossible at least if $C$ is finite-dimensional. However, it can be achieved approximately~\cite{aberg_catalytic_2014}, and it is this kind of approximate implementation that we have in mind when generalizing the construction of this subsection to continuous groups, such as compact Lie groups.

Having shown that essential elements of the PN approach such as the physical Hilbert space, reduction maps, and their transformation behavior can be obtained from an operational scenario that is resource-theoretic in spirit, we now turn to the formal mathematical construction of the framework. None of the above has assumed that QRFs are associated to subsystems, and our definitions and results will hence not rely on this assumption either, generalizing the PN framework to QRFs which are covariant instruments.

\subsection{Mathematical construction of the framework}
\label{SubsecPNMathematical}
From now on, we will assume that the underlying symmetry group $\cg$ is a compact Lie group unless stated otherwise. Furthermore, we will assume that $g\mapsto U_g^B$ is a continuous unitary representation of $\cg$ on $\ch_\kin$ (recall that we use the notation $B=\mathcal{L}(\mathcal{H}_{\rm kin})$) and that all Hilbert space are finite-dimensional. We define the physical Hilbert space as motivated and introduced for finite groups in the previous subsection:
\begin{definition}
The physical Hilbert space $\mathcal{H}_{\rm phys}$ is the subspace of $\mathcal{H}_{\rm kin}$ defined by
\[
   \mathcal{H}_{\rm phys}=\{|\psi\rangle\in\mathcal{H}_{\rm kin}\,\,|\,\, U_g^B|\psi\rangle=|\psi\rangle\mbox{ for all }g\in\mathcal{G}\}.
\]
The set of linear operators on this space will be denoted $A:=\mathcal{L}(\mathcal{H}_{\rm phys})$ (in contrast to $B=\mathcal{L}(\mathcal{H}_{\rm kin})$).
\end{definition}
Before we introduce and define our generalized notion of QRFs, we recall the definition of a covariant quantum instrument~\cite{carmeli2009}. To this end, let $\Omega$ be a compact topological space which is Hausdorff and satisfies the second countability axiom. Moreover, we assume that $\cg$ acts continuously and transitively on $\Omega$. We will denote the Borel $\sigma$-algebra of $\Omega$ by $\cb(\Omega)$. Furthermore, we will fix the notation $gX:=\{gx|x\in X\}$ for $X\in\cb(\Omega)$ and $g\in\cg$.

 We can think of $\Omega$ as the classical outcome space of the quantum instrument, and every set $X\in\cb(\Omega)$ is a classical output of the instrument. In addition to the classical output $X$, the instrument maps the input state $\rho$ to an (in general subnormalized) post-measurement state $\mathcal{I}_X(\rho)$. This is in more detail formalized in the following definition~\cite{chiribella2009}, which denotes the set of completely positive linear maps from $\mathcal{H}_{\rm in}$ to $\mathcal{H}_{\rm out}$ by ${\rm CP}(\mathcal{H}_{\rm in},\mathcal{H}_{\rm out})$ (for an alternative equivalent definition of quantum instruments, see e.g.~\cite{carmeli2009}):
\begin{definition}
A map $\mathcal{I}:\mathcal{B}(\Omega)\to {\rm CP}(\mathcal{H}_{\rm in},\mathcal{H}_{\rm out})$ is a \emph{quantum instrument} if it satisfies the properties
\begin{enumerate}
    \item ${\rm Tr}\left( \mathcal{I}_\Omega(\rho)\right)={\rm Tr}(\rho)$ for all $\rho\in\mathcal{L}(\mathcal{H}_{\rm in})$,
    \item $\mathcal{I}_{\bigcup_{i=1}^\infty B_i}=\sum_{i=1}^\infty \mathcal{I}_{B_i}$ for all sequences of pairwise disjoint $B_i\in\mathcal{B}(\Omega)$.
\end{enumerate}
\end{definition}
We will refer to $\ch_{\rm in}$ and $\ch_{\rm out}$ as the input and output Hilbert spaces, and we consider two continuous unitary representations $U_g^{\rm in}$ and $U_g^{\rm out}$ acting on the input and output Hilbert spaces, respectively. Let us recall the definition of a covariant instrument~\cite{carmeli2009}:
\begin{definition}
An instrument $\ci$ is called covariant if 
\begin{equation}
\mathcal{I}_{gX}=\mathcal{U}_g^{\rm out}\circ \mathcal{I}_X\circ \mathcal{U}_{g^{-1}}^{\rm in}
\end{equation}
for every $g\in\cg$ and every $X\in\cb(\Omega)$.
\end{definition}
Every covariant instrument comes with a covariant positive operator-valued measure (POVM).  We will denote the set of effects on a Hilbert space, i.e.\ the set of operators $E$ with $0\leq E \leq\mathds{1}$, by $\cale(\ch)$, and recall the definition of general POVMs~\cite{ziman2012}:
\begin{definition}
    The map $E:\cb(\Omega)\rightarrow\cale(\ch)$ is called a positive operator-valued measure (POVM) if it satisfies the following properties:
    \begin{enumerate}
        \item $E(\varnothing)=0$;
        \item $E(\Omega)=\mathds{1}$;
        \item $E(\bigcup_i X_i)=\sum_i E(X_i)$ for any sequence of pairwise disjoint $X_i\in\cb(\Omega)$.
    \end{enumerate}
\end{definition}
Covariance for POVMs is defined in a similar fashion as for instruments:
\begin{definition}
    Let $E$ be a POVM with values in $\Omega$ and let $g\mapsto U_g$ be a continuous unitary representation of $\cg$. The POVM $E$ is covariant if 
    \begin{equation}
     E(gX)=\mathcal{U}_g\circ E(X)\equiv U_gE(X)U_g^\dagger
    \end{equation}
    for every $g\in\cg$ and every $X\in\cb(\Omega)$.
\end{definition}

To see that every covariant instrument $\ci$ comes with a covariant POVM acting on $\ch_{\rm in}$, we define \begin{equation}
    E(X):=\ci_X^\dagger(\mathds{1})
\end{equation}
for every $X\in\cb(\Omega)$. It is easy to see that $E$ is a POVM. To check that it is covariant, we write
\begin{align}
   (E(gX),\rho)_{\ch\cs}=& (\ci^\dagger_{gX}(\mathds{1}),\rho)_{\ch\cs}= (\mathds{1},\ci_{gX}(\rho))_{\ch\cs}=(\mathds{1},U_g^{\rm out}\ci_{X}((U_g^{\rm in})^\dagger \rho U_g^{\rm in})(U_g^{\rm out})^\dagger)_{\ch\cs}\nonumber\\
   =& ((U^{\rm out}_g)^\dagger\mathds{1}U^{\rm out}_g,\ci_{X}((U_g^{\rm in})^\dagger \rho U_g^{\rm in}))_{\ch\cs}=(\ci_X^\dagger(\mathds{1}),(U_g^{\rm in})^\dagger \rho U_g^{\rm in})_{\ch\cs}\nonumber\\
   =&(U_g^{\rm in}E(X)(U_g^{\rm in})^\dagger,\rho)_{\ch\cs},
   \label{covofpovmpr}
\end{align}
where $(X,Y)_{\ch\cs}={\rm Tr}(X^\dagger Y)$ denotes the Hilbert-Schmidt inner product. Eq.~(\ref{covofpovmpr}) holds for every $g\in \cg$, every $X\in\cb(\Omega)$ and every $\rho\in\cl(\ch_{\mbox{\tiny{in}}})$, and thus $E$ is covariant.

Covariant instruments as well as covariant POVMs can be generated from so-called seeds. Recall that we have assumed that the group action of $\cg$ on $\Omega$ acts transitively. Hence, fixing some $x_0\in\Omega$, the continuous map $x$: $g\mapsto gx_0$ is surjective, and we will write $x_g=x(g)=gx_0$. The stabilizer subgroup $\cg_0$ is defined as\begin{equation}
    \cg_0=\{g\in\cg\,\,|\,\,gx_0=x_0\}.
\end{equation}  
As an abstract group, $\mathcal{G}_0$ is independent of the choice of $x_0$: concretely, fixing $x_0'=:h x_0$ rather than $x_0$ yields the isomorphic group $\mathcal{G}'_0=h\mathcal{G}_0 h^{-1}$ as the stabilizer group. Under our assumptions, $\cg_0$ is always closed and hence compact.

Let $\mu$ be the normalized Haar measure of $\cg$, i.e.\ $\mu$ is invariant and $\mu(\cg)=1$. In the case that $\cg_0$ is compact (which for us is always true), it is possible to construct a $\cg$-invariant measure $\nu$ on $\Omega$ from the Haar measure $\mu$~\cite{folland_course_2015}, where $\cg$-invariance means that the measure satisfies
\begin{equation}
    \nu(gX)=\nu(X)
\end{equation}
for every $X\in\cb(\Omega)$ and $g\in\cg$.  
Explicitly,
\begin{equation}
    \nu(X):=\mu(x^{-1}(X)),
\end{equation}
where $x^{-1}(X):=\{g\,\,|\,\,gx_0\in X\}$ is the pre-image of $X\in\cb(\Omega)$. Written as integrals, this means that\begin{equation}
    \int_\cg f(gx_0)d\mu(g)=\int_{\Omega}f(x)d\nu(x)
\end{equation}
holds for every integrable function $f$ on $\Omega$.

Every covariant instrument can be written in terms of a density in the following way:
\begin{lemma}[Theorem 2 of~\cite{chiribella2009}]
Let $\ci$ be a covariant instrument. Then there exists a completely positive map $I_{x_0}$ (the seed), which generates a density 
\begin{equation}
   I_{g x_0}=\mathcal{U}_g^{\rm out}\circ I_{x_0}\circ \mathcal{U}_g^{\rm in}
   \label{DefFromSeed}
\end{equation}
with respect to $\nu$, such that
\begin{equation}
    \ci_X=\int_X I_x d\nu(x)
    \label{EqInstrumentDensity}
\end{equation}
for every $X\in\cb(\Omega)$.
\end{lemma}
To obtain $I_x$, we can write $x=g x_0$ and use Eq.~(\ref{DefFromSeed}). It is part of the claim above that $I_x$ does not depend on the choice of $g$, i.e.\ if $x=g' x_0$ for some $g'\neq g$, then $\mathcal{U}_g^{\rm out}\circ I_{x_0}\circ \mathcal{U}_g^{\rm in}=\mathcal{U}_{g'}^{\rm out}\circ I_{x_0}\circ \mathcal{U}_{g'}^{\rm in}$. In particular, $\mathcal{U}_h^{\rm out}\circ I_{x_0}\circ \mathcal{U}_h^{\rm in}=I_{x_0}$ for all $h\in\mathcal{G}_0$, which is also shown in~\cite{chiribella2009}.

A similar result is known to hold for POVMs:
\begin{lemma}[Theorem 4.2.3 of~\cite{holevo2011}]
\label{LemPOVM}
    Let $E_{x_0}\in \cl(\ch_{\rm in})$ be a positive operator such that $[E_{x_0},U_g^{\rm in}]=0$ for every $g\in\cg_0$, satisfying
    \begin{equation}
        \int_{\mathcal{G}} \mathcal{U}_g^{\rm in}(E_{x_0})\, d\mu(g)=\mathds{1}.
    \end{equation} 
    By setting $E_{g x_0}:=\mathcal{U}_g^{\rm in}(E_{x_0})$, we obtain a covariant POVM $E$ defined by
    \begin{equation}
        E(X):=\int_{X}E_x d\nu(x).
    \end{equation}
    Conversely, for every covariant POVM $E$, there exists a unique operator $E_{x_0}$ satisfying the above conditions, generating a POVM $E$ in the fashion just described. 
\end{lemma}
Covariant instruments were shown to be closely related to the notion of a quantum reference frame (QRF) in Subsection~\ref{SubsecOpMot}: essentially, $I_x$ asks whether an internal degree of freedom is ``oriented in direction $x$'', and the corresponding state update yields a description of the state relative to this convention. To obtain a state update that yields a QRF perspective in this sense, we need a covariant instrument of the following form:
\begin{definition}
\label{DefQRF}
    A quantum reference frame (QRF) is given by
\begin{itemize}
    \item a covariant instrument $\mathcal{I}:\mathcal{B}(\Omega)\to {\rm CP}(\mathcal{H}_{\rm phys},\mathcal{H}_{\rm kin})$ such that the density seed of the instrument is of the form\begin{equation}
        I_{x_0}(\rho)=M_{x_0}\rho M_{x_0}^\dagger,\label{krausrank1}
    \end{equation}
    that is, $I_{x_0}$ has Kraus rank 1; and
    \item a choice of Kraus operators $M_x:A\to B$ such that $I_x(\rho)=M_x\rho M_x^\dagger$ and
\begin{equation}
   U_g^B M_x=\chi(g)M_{gx} \qquad \mbox{for all }x\in\Omega,g\in\mathcal{G},
   \label{EqTransformM}
\end{equation}
where $\chi$ is some one-dimensional representation of $\mathcal{G}$.
\end{itemize}
\end{definition}
Note that Eq.~(\ref{EqTransformM}) implies that the map of Eq.~(\ref{krausrank1}) is a covariant instrument if only one further condition is satisfied: the normalization condition, $\mathcal{I}_\Omega^\dagger(\mathds{1})=\mathds{1}_{\ch_{\rm phys}}$. This is equivalent to $\int_\Omega M_x^\dagger M_x\, d\nu(x)=\mathds{1}_{\ch_{\rm phys}}$.

For simplicity, we will often denote the QRF by $\mathcal{I}$ (i.e.\ the corresponding covariant instrument), but recall that it always comes with an implicit choice of (the phases of) the associated isometries $\{M_x\}_{x\in\Omega}$.

Let us first see that the standard approach of ``QRFs as subsystems'' is covered by this definition as a special case. In this standard scenario, one has $\ch_{\kin} \cong \ch_R \otimes \ch_S$, transforming under a product unitary representation $g \mapsto U_g^{R} \otimes U_g^S$ of $\cg$. The QRF is then associated with the subsystem $R$, and it is assumed that there is a distinguished covariant POVM with density $E_x^R$ for the representation $g \mapsto U_R(g)$ on $\ch_R$ (also known as a system of imprimitivity) which defines how the ``orientation'' of the QRF should be understood.
\begin{lemma}[Subsystem QRFs as a special case]
\label{LemSubsystemQRFs}
In the scenario just described, define the Kraus operators as $M_x^R:=\sqrt{E_x^R}$ and $M_x := M_x^R \otimes \mathds{1}_S\upharpoonright \ch_\phys$. Then this defines a QRF in the sense of Definition~\ref{DefQRF}.
\end{lemma}
\begin{proof}
By the definition of a covariant POVM, and by Lemma~\ref{LemPOVM}, we have $E_{gx}^R=U_g^R E_x^R {U_g^R}^\dagger$, and taking the square root implies $M_{gx}^R=U_g^R M_x^R {U_g^R}^\dagger$. Hence
\begin{eqnarray*}
M_{gx}&=& M_{gx}^R\otimes\mathds{1}\upharpoonright\ch_{\rm phys}=(U_g^R\otimes U_g^S)(M_x^R\otimes\mathds{1}_S)(U_g^R\otimes Ug^S)\upharpoonright\ch_{\rm phys}=(U_g^R\otimes U_g^S)(M_x^R\otimes\mathds{1}_S)\upharpoonright\ch_{\rm phys}\\
&=&(U_g^R\otimes U_g^S)M_x,
\end{eqnarray*}
since $U_g^B:=U_g^R\otimes U_g^S$ acts as the identity on $\ch_{\rm phys}$. This proves Eq.~(\ref{EqTransformM}) of Definition~\ref{DefQRF}, for $\chi=\mathbf{1}$ the trivial one-dimensional representation. To also show Eq.~(\ref{krausrank1}), we only have to check normalization, since covariance of the associated instrument follows from the condition we have just proved. This is shown as follows:
\begin{eqnarray*}
\mathcal{I}_\Omega^\dagger(\mathds{1})&=&\int_\Omega M_x^\dagger M_x d\nu(x)=\int_{\Omega} E_x^R\otimes\mathds{1}_S\, d\nu(x) \upharpoonright \ch_{\rm phys} = E^R(\Omega)\otimes\mathds{1}_S\upharpoonright\ch_{\rm phys}=\mathds{1}_{\ch_{\rm phys}}.
\end{eqnarray*}
This completes the proof that the special case of subsystem QRFs is covered by Definition~\ref{DefQRF}. What to do with it, for example how to define reduction maps, will be discussed in the remainder of this section.

\end{proof}

This constitutes already a substantial generalization of the earlier PN framework~\cite{Hamette2021}, where the measurements on $R$ were restricted to being proportional to coherent state systems. Now they may correspond to general POVMs on $R$. This lifts the PN framework to the same level of generality as other QRF frameworks, and Definition~\ref{DefQRF} generalizes this even further.

For any $\rho$ that is supported on the physical Hilbert space, i.e.\ $\rho\in A=\mathcal{L}(\mathcal{H}_{\rm phys})$, we interpret $\rho$ as the ``perspective-neutral'' description of this state, and $I_{x}(\rho)$ as its description within the particular perspective that the internal degree of freedom of $A$, defined by $\mathcal{I}$, ``points in direction $x$''. Hence, a QRF is here not defined to be a physical subsystem that appears in a tensor product factorization of $\mathcal{H}_{\rm kin}$, but a more general structure associated to a covariant instrument. We will soon show that the $M_x$ are isometries from $\mathcal{H}_{\rm phys}$ into $\mathcal{H}_{\rm kin}$, and hence, ``jumping into a perspective'' is a reversible map, which ultimately also admits the definition of reversible maps \textit{between perspectives}, i.e.\ QRF changes.

The condition~(\ref{krausrank1}) is a purity condition, in the sense that it assures that $I_{x_0}$ maps pure states to pure states. We will see later that this leads to the fact that we can define reversible QRF transformations, and hence, giving up this property would make changes of perspective genuinely irreversible.

This pragmatic justification notwithstanding, the main motivation for the form~(\ref{krausrank1}) comes from Subsection~\ref{SubsecOpMot} above: in the case where $\mathcal{G}$ is a finite group and $\mathcal{G}_0$ is trivial, i.e.\ $\mathcal{G}_0\simeq\{\mathds{1}\}$, the instrument density agrees up to a factor with the corresponding element of the discrete instrument of Subsection~\ref{SubsecOpMot}, i.e.\ $I_g=|\mathcal{G}|\cdot\mathcal{I}_g$. This discrete instrument was constructed as a reinterpretation of a completely covariant isometry from $AC$ to $BC$, i.e.\ as a redescription of the state relative to an ideal reference frame, and complete covariance was enforcing the transformation behavior of Eq.~(\ref{EqTransformM}).

If we only demand that $I_x(\rho)=M_x \rho M_x^\dagger$, then the covariance of the instrument implies that $U_g^B M_x=\omega(g,x) M_{gx}$ for some complex phases $\omega(g,x)\in\mathbb{C}$, $|\omega(g,x)|=1$. These phases are arbitrary in the following sense: for any given choice of Kraus operators $\{M_x\}_{x\in\Omega}$ that generate the instrument $\mathcal{I}$, multiplying them by arbitrary phases, $M'_x:=\theta(x)M_x$ with $|\theta(x)|=1$, gives another set of Kraus operators that also generates $\mathcal{I}$. But if we demand that the $M_x$ come directly from the completely covariant transformation~(\ref{DefU}), then these phases must also transform under $\mathcal{G}$, i.e.\ they must be one-dimensional representations.

It is easy to see that for a given QRF $\ci$, the corresponding covariant POVM is given by \begin{equation}
    E(X):=\int_X M_x^\dagger M_xd\nu(x)
\end{equation}
for every $X\in\cb(\Omega)$. We will write \begin{equation}
    E_x=M^\dagger_xM_x
\end{equation}
and for $x=gx_0$, we have
\begin{equation}
    E_x=E_{gx_0}=M_{g x_0}^\dagger M_{g x_0}=\left(M_{x_0}^\dagger {U_g^B}^\dagger\right)\left(U_g^B M_{x_0}\right)=M_{x_0}^\dagger M_{x_0}=E_{x_0}.
\end{equation}
Hence, the POVM density is constant, and equal to the unit element of $A$, i.e.\ the projector onto $\ch_{\rm phys}$:
\begin{equation}
\mathds{1}_A=E(\Omega)=\int_\Omega E_x\, dx=E_{x_0}.
\end{equation}
Hence, in some sense, the POVM that is associated to the covariant instrument is completely uninformative: it yields a uniformly random outcome. However, recall from Subsection~\ref{SubsecOpMot} that the Kraus operators have an alternative interpretation: they are isometries that yield asymmetric redescriptions of invariant states in the physical Hilbert space. This follows directly from $M_x^\dagger M_x=E_x=\mathds{1}_A$:
\begin{lemma}
\label{LemKrausIsometries}
    Let $\ci$ be a QRF. Then every associated Kraus operator $M_{x}:\mathcal{H}_{\rm phys}\to\mathcal{H}_{\rm kin}$ is an isometry.
\end{lemma}
We can hence use every $M_x$ to ``jump into the perspective where the QRF $\mathcal{I}$ points in direction $x$'', mapping every perspective-neutral state vector $|\psi\rangle\in\ch_{\rm phys}$ to a normalized state vector $M_x|\psi\rangle\in\ch_{\rm kin}$. The result corresponds to the unique description of the state relative to the choice of convention that $\mathcal{I}$ is considered to point in direction $x$.

Clearly, we can use the maps $I_x:=M_x\bullet M_x^\dagger$ to map density matrices $\rho$ and observables $Y$ from $A=\mathcal{L}(\ch_{\rm phys})$ into $B=\mathcal{L}(\ch_{\rm kin})$, preserving all expectation values via ${\rm Tr}(\rho Y)={\rm Tr}(I_x(\rho)I_x(Y))$. This preserves the algebraic structure:
\begin{corollary}
For a QRF $\mathcal{I}$, every associated $I_x$ is a $^*$-homomorphism from $A$ to $B$.
\end{corollary}
That is, $I_x$ preserves the adjoint, sums, scalar multiples, and products, i.e.\ $I_x(X^\dagger)=I_x(X)^\dagger$, $I_x(a X+bY)=a I_x(X)+b I_x(Y)$ for $a,b\in\mathbb{C}$, and $I_x(XY)=I_x(X)I_x(Y)$. This follows directly from $I_x=M_x\bullet M_x^\dagger$.

This property is analogous to that of the ``Schr\"odinger reduction map'' of~\cite{Hamette2021}; in Section~\ref{SecRelPrevious}, we will explain in more detail in what sense the framework of~\cite{Hamette2021} is a special case of ours.

There is also a converse of Lemma~\ref{LemKrausIsometries}:
\begin{lemma}
Let $\{M_x\}_{x\in\Omega}$ be any set of operators that satisfies Eq.~(\ref{EqTransformM}), and suppose that one of the $M_x$ (and thus all of them) is an isometry. Then $I_x:=M_x\bullet M_x^\dagger$ is a valid density for a covariant instrument.
\end{lemma}
\begin{proof}
We only have to check normalization. Using that $M_x^\dagger M_x=\mathds{1}$, we get
\[
   {\rm Tr}(\mathcal{I}_\Omega(\rho))={\rm Tr}\int_\Omega M_x \rho M_x^\dagger d\nu(x)=\int_\Omega {\rm Tr}(\rho M_x^\dagger M_x)\,d\nu(x)={\rm Tr}(\rho).
\]
It is a straightforward calculation that the transformation behavior is the required one, too.
\end{proof}

Let us now analyze in some more detail the images of the isometries $M_x:\mathcal{H}_{\rm phys}\to\mathcal{H}_{\rm kin}$.
\begin{definition}
\label{DefRepSpaces}
For every QRF $\mathcal{I}$ and every $x\in\Omega$, we define the representation space
\[
   \mathcal{H}_{\rm phys}^{\mathcal{I}}(x):={\rm Im}(M_x)\subseteq\mathcal{H}_{\rm kin},
\]
and the orthogonal projector $\Pi_{\rm phys}^{\mathcal{I}}(x):\mathcal{H}_{\rm kin}\to \mathcal{H}_{\rm kin}$ as
\[
   \Pi_{\rm phys}^{\mathcal{I}}(x):=M_x M_x^\dagger.
\]
\end{definition}
From simple linear algebra, if $V$ is an isometry, then $V V^\dagger$ is the orthogonal projector into the image of $V$. Hence, since every $M_x$ is an isometry, $\Pi_{\rm phys}^{\mathcal{I}}(x)$ is the orthogonal projector onto $\mathcal{H}_{\rm phys}^{\mathcal{I}}(x)$. This Hilbert space is the collection of all state descriptions that we can obtain if we ``jump into the perspective defined by QRF $\mathcal{I}$ and outcome $x$''. Its dimension equals that of $\mathcal{H}_{\rm phys}$.

Since ``jumping into a QRF perspective'' is a reversible transformation (mathematically, an isometry), we can continue in analogy with the earlier PN framework~\cite{Hamette2021,Vanrietvelde2020} and define reversible \textit{QRF transformations} as changes between different perspectives:
\begin{definition}
\label{DefQRFTrafo}
    Let $x,x'\in\Omega$, and let $\mathcal{I}$ and $\mathcal{I}'$ be two QRFs, associated to sets of Kraus operators $\{M_x\}_{x\in\Omega}$ and $\{M'_x\}_{x\in\Omega}$, respectively. The \emph{QRF transformation} $U^{\mathcal{I},\mathcal{I}'}_{x,x'}:\ch_{\rm phys}^{\mathcal{I}}(x)\to \ch_{\rm phys}^{\mathcal{I}'}(x')$ is the unitary map
    \begin{equation}
    \label{eq:QuantCoordChanges}
       U^{\mathcal{I},\mathcal{I}'}_{x,x'}:=M'_{x'} M_x^\dagger\upharpoonright{\mathcal{H}_{\rm phys}^{\mathcal{I}}(x)}.
    \end{equation}
\end{definition}

\begin{figure}[hbt]
\includegraphics[width=.7\columnwidth]{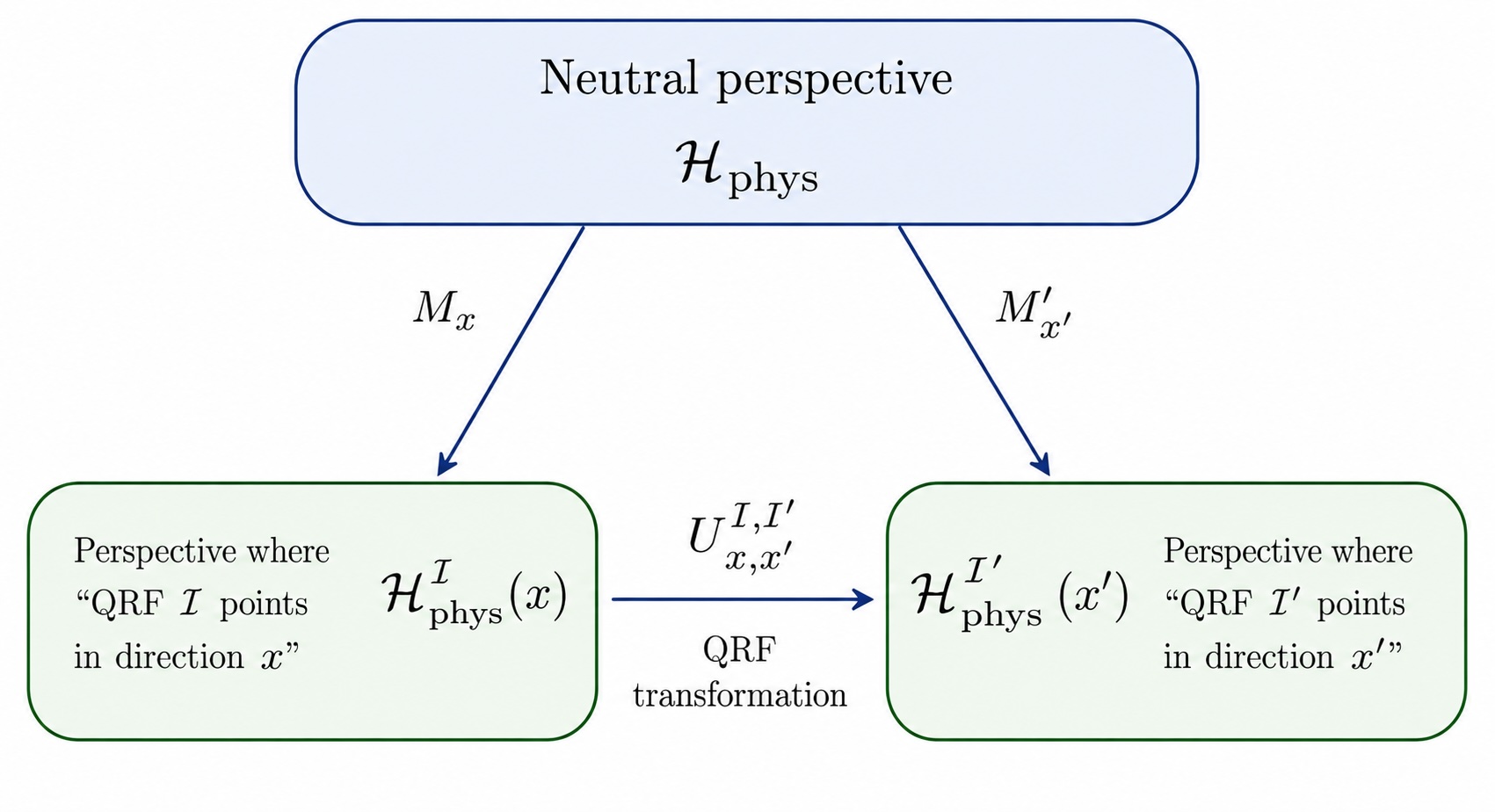}
\caption{QRF transformations as changes of perspective, similarly as in earlier formulations of the PN framework~\cite{Hamette2021,Vanrietvelde2020}. Since $M_x$ and $M'_{x'}$ are isometries, the QRF transformations $U_{x,x'}^{\mathcal{I},\mathcal{I}'}$ are unitaries that can be obtained by composing and restricting them to the relevant domain, as in Definition~\ref{DefQRFTrafo}. \aidisclosure}
\label{fig:qrf-change}
\end{figure}
The gauge transformations $U_g^B$ restricted to the physical Hilbert space $\ch_\phys^\ci(x)$ encoding the perspective of a QRF $\ci$ are special cases of QRF transformations, changing between two perspectives associated to the \textit{same} QRF:
\begin{lemma}
\label{LemQRFExGauge}
Let $\mathcal{I}$ be a QRF, with associated Kraus operators $\{M_x\}_{x\in\Omega}$ transforming via $\chi$. Moreover, let $x,x'\in\Omega$. Then, for every $g\in\mathcal{G}$ with $x'=gx$ (if $\mathcal{G}_0\not\simeq\{\mathds{1}\}$ there is more than one), we have
\begin{equation}
   U^{\mathcal{I},\mathcal{I}}_{x,x'}=\overline{\chi(g)} \,U_g^B\upharpoonright {\mathcal{H}_{\rm phys}^{\mathcal{I}}(x)}.
   \label{eqTrafoSame}
\end{equation}
That is, up to a phase and to a restriction of the domain, this QRF transformation is the gauge transformation $U_g^B$.
\end{lemma}
\begin{proof}
A simple calculation yields
\[
   U^{\mathcal{I},\mathcal{I}'}_{x,x'}=M_{x'}M_x^\dagger\upharpoonright {\mathcal{H}_{\rm phys}^{\mathcal{I}}(x)} = \overline{\chi(g)}\, U_g^B M_x M_x^\dagger \upharpoonright {\mathcal{H}_{\rm phys}^{\mathcal{I}}(x)}=\overline{\chi(g)}\, U_g^B \Pi_{\rm phys}^{\mathcal{I}}(x)\upharpoonright {\mathcal{H}_{\rm phys}^{\mathcal{I}}(x)},
\]
and the orthogonal projection $\Pi_{\rm phys}^{\mathcal{I}}(x)$ acts on its own image as the identity.
\end{proof}

Before generalizing this, let us note a reformulation of Lemma~\ref{LemQRFExGauge}. Rather than considering this as a transformation of a \textit{single} QRF that is regarded as pointing in two \textit{different} directions $x$ and $x'$, we can also view it as transforming between two \textit{different} QRFs that point in the \textit{same} direction -- at least in some special case.

To this end, recall our definition of the stabilizer group $\mathcal{G}_0=\{g\in\mathcal{G}\,\,|\,\, gx_0=x_0\}$. As an \textit{abstract} group, this does not depend on the choice of $x_0$, but as a concrete subgroup of $\mathcal{G}$, it does. To emphasize this, let us denote the dependence on $x_0$ by $\mathcal{G}_0(x_0)$. Then, if $x_0'=h x_0$, we have $\mathcal{G}_0(x'_0)=h\mathcal{G}_0(x_0) h^{-1}$. Thus, we have the following equivalence:
\[
   \mathcal{G}_0\mbox{ is independent of the choice of }x_0\quad\Leftrightarrow \quad\mathcal{G}_0\trianglelefteq \mathcal{G},\mbox{ i.e. }\mathcal{G}_0\mbox{ is a normal subgroup of }\mathcal{G}.
\]
This is automatically the case, for example, if $\mathcal{G}_0=\{\mathds{1}\}$ or if $\mathcal{G}$ is Abelian. It implies that $\Omega=\mathcal{G}/\mathcal{G}_0$ is itself a group, the quotient group.

Now, similarly as in Lemma~\ref{LemQRFExGauge}, pick some fixed elements $x,x'\in\Omega$, and let us construct a QRF $\mathcal{I}'$ that is supposed to satisfy the following:
\[
   \mbox{The new QRF }\mathcal{I}'\mbox{ points in direction }x\mbox{ if and only if the old QRF }\mathcal{I}\mbox{ points in direction }x'.
\]
That this is possible, and that it uniquely defines a new QRF, is the content of the following lemma:
\begin{lemma}
\label{LemReinterpretQRF}
Suppose that $\mathcal{G}_0\trianglelefteq \mathcal{G}$, and let $\mathcal{I}$ be a QRF with associated Kraus operators $\{M_y\}_{y\in\Omega}$ transforming via one-dimensional representation $\chi$. Let $x,x'\in\Omega$ be fixed. Then there is a unique QRF $\mathcal{I}'$ with associated Kraus operators $\{M'_y\}_{y\in\Omega}$ transforming via $\chi$ such that $M'_x=M_{x'}$.
\end{lemma}
\begin{proof}
Set $M'_x:=M_{x'}$. For every $y\in\Omega$, we have to define $M'_y$ via $\overline{\chi(g)}\,U_g^B M'_x$, where $g\in\mathcal{G}$ is some group element such that $y=gx$. However, we have to check that this is independent of the choice of $g$. So suppose that $g'$ is a potentially different group element with $y=g' x$ too, then
\[
   gx=g'x\Leftrightarrow g^{-1} g'\in\mathcal{G}_0(x)\Leftrightarrow g^{-1}g'\in \mathcal{G}_0(x')\Leftrightarrow gx'=g'x'.
\]
Therefore, we have
\[
   \overline{\chi(g')}\, U_{g'}^B M'_x = \overline{\chi(g')}\, U_{g'}^B M_{x'}=M_{g'x'}=M_{gx'}=\overline{\chi(g)}\, U_g^B M_{x'}=\overline{\chi(g)}\, U_g^B M'_{x}.
\]
So the $\{M'_y\}_{y\in\Omega}$ are well-defined, and they are clearly isometries because $M_{x'}$ is. We finally have to check that they satisfy Eq.~(\ref{EqTransformM}), i.e.\ the condition to generate a QRF. Let $g\in\mathcal{G}$ and $y\in\Omega$ be arbitrary elements, and fix some $h\in\mathcal{G}$ such that $hx=y$. Then
\begin{equation}
   \overline{\chi(g)}\, U_g^B M'_y =\overline{\chi(g)}\, U_g^B \overline{\chi(h)}\, U_h^B M'_x=\overline{\chi(gh)}U_{gh}^B M'_x = M'_{ghx}=M'_{gy}.
   \label{EqTransformsNicely}
\end{equation}
Hence, we obtain a unique QRF with the required properties.
\end{proof}
What the above says is that we can obtain a new QRF $\mathcal{I}'$ by ``readjusting'' a given QRF $\mathcal{I}$, in a similar way as we can get a new clock by setting a given clock to a new time: \textit{``The Lisbon clock is the clock which shows noon exactly when the Viennese clock shows 1pm''.}

Now consider the idea of adjusting a clock ``in superposition'':
\begin{lemma}[Coherently controlled adjustment of a QRF]
\label{LemCoherentlyControlled}
Suppose that $\mathcal{G}_0\trianglelefteq\mathcal{G}$. Let $\{\Pi_j\}_j\subset B$ be a projective measurement on $\mathcal{H}_{\rm kin}$ which is $\mathcal{G}$-symmetric, i.e.\ $[\Pi_j,U_g^B]=0$ for all $j$ and all $g\in\mathcal{G}$. Let $x\in\Omega$, and let $\mathcal{I}$ be a QRF with associated Kraus operators $\{M_y\}_{y\in\Omega}$ transforming via one-dimensional representation $\chi$. Let $\{g_j\}_{j\in J}\subset\mathcal{G}$ be arbitrary group elements, and $\{\theta_j\}_{j\in J}\subset\mathbb{R}$ arbitrary phases. Then there is a unique QRF $\mathcal{I}'$ with Kraus operators $\{M'_y\}_{y\in\Omega}$, also transforming via $\chi$, with the property
\[
   M'_x=\sum_{j\in J} e^{i\theta_j} \Pi_j M_{g_j x}.
\]
\end{lemma}
The projectors $\Pi_j$ can be thought of pertaining to some observable on $B$ which is itself $\mathcal{G}$-invariant, i.e.\ commutes with all $U_g^B$. Depending on the value of this observable, the QRF is readjusted, as in the following example: \textit{Consider an auxiliary spin-$1/2$ particle. If its spin is up, then the clock is the Vienna clock; if its spin is down, then the clock is the Lisbon clock.} This defines a new clock that measures time in a world which also contains the auxiliary particle. See Figure~\ref{fig:cc} for an illustration.
\begin{proof}
Define the linear map $V(\{(g_j,\theta_j)\}_{j\in J}):\mathcal{H}_{\rm kin}\to\mathcal{H}_{\rm kin}$ as
\[
   V(\{(g_j,\theta_j)\}_{j\in J}):=\sum_{j\in J} e^{i\theta_j}\Pi_j\overline{\chi(g_j)}\, U_{g_j}^B.
\]
A straightforward calculation shows that this map is unitary, and hence $M'_x=V(\{(g_j,\theta_j)\}_{j\in J}) M_x$ is an isometry. If $y\neq x$, find some $g\in\mathcal{G}$ such that $y=gx$ and set $M'_y:=\overline{\chi(g)}\, U_g^B M'_x$. We first have to show that this definition makes sense: choosing some other $g'$ with $y=g' x$ should yield the same result. As in Lemma~\ref{LemReinterpretQRF}, $gx=g'x$ implies that $gz=g'z$ for all $z\in\Omega$. This implies that we can replace $g'$ by $g$ in the final expression of the following equation:
\begin{eqnarray*}
\overline{\chi(g')} U_{g'}^B M'_x &=& \sum_{j\in J} e^{i\theta_j}\overline{\chi(g')} U_{g'}^B \Pi_j \overline{\chi(g_j)} U_{g_j}^B M_x=\sum_{j\in J}e^{i\theta_j} \overline{\chi(g' g_j)} \Pi_j U_{g' g_j}^B M_x=\sum_{j\in J}e^{i\theta_j} \Pi_j M_{g'g_j x}.
\end{eqnarray*}
Hence the result is the same for $g$ and $g'$, and the $M'_y$ are well-defined operators. To check that they satisfy Eq.~(\ref{EqTransformM}), i.e.\ the condition to generate a QRF, we simply repeat the calculation of Eq.~(\ref{EqTransformsNicely}) of Lemma~\ref{LemReinterpretQRF}.
\end{proof}

\begin{figure}[hbt]
\includegraphics[width=\columnwidth]{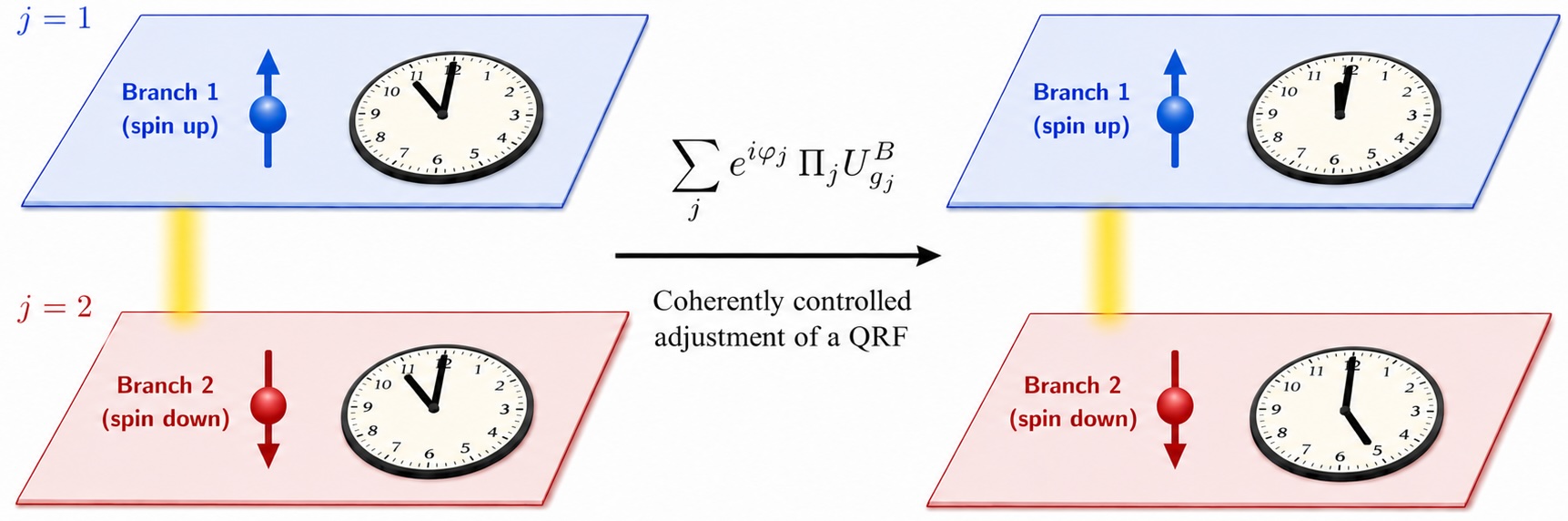}
\caption{As shown in Lemma~\ref{LemCC}, coherently controlled adjustments of a QRF are valid QRF transformations. A $\mathcal{G}$-invariant observable $X\in B$ with eigenprojectors $\Pi_j$ (here, the spin of a spin-$1/2$ particle) defines the branches; depending on the value of $X$, i.e.\ on the branch $j$, the QRF is adjusted differently, described by group elements $g_j$. At the same time, the relative phases between the branches can be altered. The special case that all $g_j$ (and all $\varphi_j$) are equal covers the case of the original gauge transformations $U_g^B$, treated in Lemma~\ref{LemQRFExGauge}. \aidisclosure}
\label{fig:cc}
\end{figure}

\begin{lemma}[Coherently controlled gauge transformations as QRF transformations]
\label{LemCC}
Suppose that $\mathcal{G}_0\trianglelefteq \mathcal{G}$, and let $x\in\Omega$. Let $\{\Pi_j\}_{j\in J}$ be a projective measurement such that $[\Pi_j,U_g^B]=0$ for all $j$ and all $g\in\mathcal{G}$, and let $\{(g_j,\varphi_j)\}_{j\in J}$ be arbitrary group elements and real numbers. Then, for every QRF $\mathcal{I}$, the coherently controlled application of $g_j$ with phases $\varphi_j$ defines a valid QRF transformation from $\mathcal{I}$ to a new ``coherently adjusted'' QRF $\mathcal{I}'$, i.e.
\[
   U^{\mathcal{I},\mathcal{I}'}_{x,x}= \sum_{j\in J} e^{i\varphi_j} \Pi_j U_{g_j}^B \upharpoonright {\ch_{\rm phys}^{\mathcal{I}}(x)}.
\]
\end{lemma}
Colloquially, this is a ``branchwise adjustment'' of the QRF: in branch $j$, we adjust by $g_j$, and doing so, we can also adjust the relative phases arbitrarily.
\begin{proof}
Define the phases $\theta_j$ such that $e^{i\theta_j}=e^{i\varphi_j}\chi(g_j)$, and consider the QRF $\mathcal{I}'$ of Lemma~\ref{LemCoherentlyControlled}. According to Definition~\ref{DefQRFTrafo}, the associated QRF transformation is
\[
   U^{\mathcal{I},\mathcal{I}'}_{x,x}=M'_x M_x^\dagger\upharpoonright {\ch_{\rm phys}^{\mathcal{I}}(x)}=V(\{g_j,\theta_j\}_{j\in J})\underbrace{M_x M_x^\dagger}_{\Pi_{\rm phys}^{\mathcal{I}}(x)} \upharpoonright {\ch_{\rm phys}^{\mathcal{I}}(x)}=\sum_{j\in J} e^{i\varphi_j} \Pi_j U_{g_j}^B \upharpoonright {\ch_{\rm phys}^{\mathcal{I}}}.
\]
This completes the proof.
\end{proof}
QRF transformations of the form of Lemma~\ref{LemCC}, i.e.\ ``coherent superpositions of symmetry transformations'', have previously been studied e.g.\ in~\cite{Giacomini,HametteGalley} and follow-up works, where it was suggested that covariance under some of those maps (``quantum covariance principles'') may allow us to obtain predictions in some regimes of quantum gravity~\cite{QRFIndefiniteMetric,GiacominiBrukner2022,GiacominiBrukner2020,hardy_implementation_2020}. In~\cite{mekonnen_invariance_2026}, some of us have shown that invariance under a group of QRFs of this form allows us to understand the absence of parastatistics in nature, while pointing out that the correct choice of phases $\varphi_j$ (typically set to zero in earlier works) is crucial for doing so.

To see an example of how further elements of the earlier PN framework extend to our setup, let us generalize the notion of relational Dirac observables as defined in~\cite{Hamette2021} to our framework. The observables will be the results of applying a \textit{relativization map} of an observable $O$ relative to some QRF $\mathcal{I}$ ``pointing in direction $x$''. However, we will begin by considering more general observables $O\in B=\mathcal{L}(\mathcal{H}_{\rm kin})$, rather than on $A=\mathcal{L}(\mathcal{H}_\mathrm{phys})$ alone; and we will consider instruments $\mathcal{I}:\mathcal{B}(\Omega)\to{\rm CP}(\mathcal{H}_{\rm kin},\mathcal{H}_{\rm kin})$ that describe measurements on general input states $\rho\in B$ (rather than $A$).
While we have defined instruments and their instrument densities $I_x$ to act on states by convention, the associated action on observables is the Hilbert-Schmidt adjoint CP map $I_x^\dagger$, which is employed in the following definition.

\begin{definition}[Relationalization of observables]
\label{DefRel}
Let $\{I_x\}_{x\in\Omega}\subset B$ be an instrument density that generates a covariant instrument $\mathcal{I}:\mathcal{B}(\Omega)\to{\rm CP}(\mathcal{H}_{\rm kin},\mathcal{H}_{\rm kin})$.
We define the relationalization map $\text{\euro}_x^{\mathcal{I}}:B\to B$ as
\[
   \text{\euro}_x^{\mathcal{I}}(O):=\int_\mathcal{G} d\mu(g)\, U_g^B I_x^\dagger(O) {U_g^B}^\dagger.
\]
\end{definition}

As shown at the end of this section, this generalizes the construction of the $\$$ map~\cite{Bartlett2007} and the $\yen$ map~\cite{Loveridge2018,Loveridge2017,Carette2025,glowacki_quantum_2024,fewster_quantum_2024,glowacki_w-algebraic_2026} (a different generalization of the $\yen$ map to instruments in the operational framework of~\cite{Carette2025} has been proposed in~\cite{JorqueraRiera2026,JorqueraRiera2026b}) and has the following properties:
\begin{lemma}
\label{LemRel}
The relationalization map is unital, i.e.\ $\text{\euro}_x^{\mathcal{I}}(\mathds{1})=\mathds{1}$, and completely positive, i.e.\ a valid quantum channel acting on observables. Furthermore, the relational observable $\text{\euro}_x^{\mathcal{I}}(O)$ is a ``Dirac observable'', i.e.\ $[\text{\euro}_x^{\mathcal{I}}(O),U_g^B]=0$ for all $g\in\mathcal{G}$ and $O\in B$, and
\begin{equation}
   \text{\euro}_{gx}^{\mathcal{I}}=\text{\euro}_x^{\mathcal{I}}\circ\mathcal{U}_{g^{-1}}^B.
   \label{EqContra}
\end{equation}
\end{lemma}
\begin{proof}
Complete positivity is immediate from the form of the map; unitality follows from
\begin{eqnarray*}
   \text{\euro}_x^{\mathcal{I}}(\mathds{1})&=& \int_{\mathcal{G}}d\mu(g)\, U_g^B I_x^\dagger(\mathds{1}) U_g^B=\int_{\mathcal{G}}d\mu(g)\, I_{gx}^\dagger(\mathds{1})=\int_\Omega I_y^\dagger(\mathds{1}) d\nu(y)=\mathcal{I}_\Omega^\dagger(\mathds{1})=\mathds{1}.
\end{eqnarray*}
Furthermore, we have
\begin{eqnarray*}
U_h^B \text{\euro}_x^{\mathcal{I}}(O) {U_h^B}^\dagger &=&\int_\mathcal{G} d\mu(g)\, U_{hg}^B I_x^\dagger(O){U_{hg}^B}^\dagger=\text{\euro}_x^{\mathcal{I}}(O)
\end{eqnarray*}
by the invariance of the Haar measure under $g\mapsto hg$. Similarly,
\[
\text{\euro}_{hx}^{\mathcal{I}}(O)=\int_\mathcal{G} d\mu(g)\, U_g^B I_{hx}^\dagger(O) {U_g^B}^\dagger=\int_\mathcal{G} d\mu(g)\, U_{gh}^B I_x^\dagger \left({U_h^B}^\dagger O U_h^B\right){U_{gh}^B}^\dagger=\text{\euro}_x^{\mathcal{I}}\left({U_h^B}^\dagger O U_h^B\right)
\]
by the invariance of the Haar measure under $g\mapsto gh$.
\end{proof}
Eq.~(\ref{EqContra}) says that the relationalization of an observable, assuming that the QRF points in direction $gx$, is the same as the relationalization of the $(g^{-1})$-transformed observable, assuming that the QRF points in direction $x$.

This notion of relational Dirac observables is not specific to the PN framework -- the above simply demonstrates that the move from subsystems to covariant instruments as reference frames can generally be applied to the relationalization map. The physical Hilbert space or other PN-specific structure does not appear in the above, and hence, the result also applies to the quantum information approach to QRFs, for example. The next result, however, will be PN-specific.

\begin{lemma}
\label{LemOurTheorem1}
Let $\{I_x\}_{x\in\Omega}\subset B$ be an instrument density that generates a covariant instrument $\mathcal{I}:\mathcal{B}(\Omega)\to{\rm CP}(\mathcal{H}_{\rm kin},\mathcal{H}_{\rm kin})$. Moreover, suppose that $I_x$ has Kraus rank $1$, i.e.\ $I_x=M_x\bullet M_x^\dagger$ for matrices $M_x\in B$; by convention, we can choose phases for these matrices such that $M_{gx}=\overline{\chi(g)}U_g^B M_x {U_g^B}^\dagger$, where $\chi$ is some one-dimensional representation of $\mathcal{G}$. Then the operators $M'_x:= M_x\upharpoonright\mathcal{H}_{\rm phys}$ define a QRF in the sense of Definition~\ref{DefQRF} (which we will call $\mathcal{I}'$), and it has representation spaces $\mathcal{H}_{\rm phys}^{\mathcal{I}'}(x)$ in the sense of Definition~\ref{DefRepSpaces}. We have
\[
   \text{\euro}_x^{\mathcal{I}}(O)\in\mathcal{L}(\ch_{\rm phys})\mbox{ and }
   I'_x(\text{\euro}_x^{\mathcal{I}}(O))=O \mbox{ for all }O\in \mathcal{L}(\mathcal{H}_{\rm phys}^{\mathcal{I}'}(x)).
\]
In particular, the relationalization map is physically invertible on this subalgebra, i.e.\ $\text{\euro}_x^{\mathcal{I}}\upharpoonright \mathcal{L}(\mathcal{H}_{\rm phys}^{\mathcal{I}'}(x))={I'_x}^\dagger={I'_x}^{-1}$.
\end{lemma}
That is, on this representation subspace, the map $\text{\euro}_x^ {\mathcal{I}}$ ``jumps from a perspective back to the neutral description'', and this is a physically invertible map. This invertibility even in cases where the QRF is not ideal is a speciality of the PN framework as compared to other approaches to QRFs. This generalizes Theorem 1 of~\cite{Hamette2021} to our framework.
\begin{proof}
We use the notation $M'_x:=M_x\upharpoonright\ch_{\rm phys}$, then $M'_x:\ch_{\rm phys}\to\ch_{\rm kin}$, and, by simple linear algebra, ${M'_x}^\dagger=\Pi'_{\rm phys}M_x^\dagger$, where $\Pi'_{\rm phys}$ is the same as $\Pi_{\rm phys}$, but defined as a map from $\ch_{\rm kin}$ to $\ch_{\rm phys}$. We have
\[
   M'_{gx}=M_{gx}\upharpoonright \ch_{\rm phys}=\overline{\chi(g)} U_g^B M_x U_g^B \upharpoonright \ch_{\rm phys}=\overline{\chi(g)} U_g^B M_x \upharpoonright \ch_{\rm phys} =\overline{\chi(g)} U_g^B M'_x,
\]
and so the $\{M'_x\}_{x\in\Omega}$ transform as required from a QRF by Definition~\ref{DefQRF}. Defining $\mathcal{I}_X':=\int_X I'_x\,d\nu(x)$ via the density $I'_x:= M'_x\bullet {M'_x}^\dagger$, gives us $\mathcal{I}'_X(\rho)=\mathcal{I}_X(\rho)$ for all $\rho\in A=\mathcal{L}(\ch_{\rm phys})$ and all $X\in\mathcal{B}(\Omega)$. In particular, ${\rm Tr}(\mathcal{I}'_\Omega(\rho))={\rm Tr}(\rho)$, and so $\mathcal{I}'$ is a valid covariant instrument. Since its Kraus operators transform as required, $\mathcal{I}'$ is a QRF according to Definition~\ref{DefQRF}; in particular, due to Lemma~\ref{LemKrausIsometries}, every $M'_x:\mathcal{H}_{\rm phys}\to\mathcal{H}_{\rm kin}$ is an isometry.

Now let $O\in\mathcal{L}(\ch_{\rm phys}^{\mathcal{I}'}(x))$. By definition, there is some $O'\in\mathcal{L}(\ch_{\rm phys})$ such that $O=M'_x O' {M'_x}^\dagger$. Hence
\[
   \text{\euro}_x^{\mathcal{I}}(O)=\int_{\mathcal{G}}d\mu(g)\, U_g^B {M'_x}^\dagger O M'_x U_g^B=\int_{\mathcal{G}}d\mu(g)\, U_g^B ({M'_x}^\dagger M'_x) O' ({M'_x}^\dagger M'_x) U_g^B=\int_{\mathcal{G}}d\mu(g)\, U_g^B O' {U_g^B}^\dagger=O'.
\]
Moreover, $I'_x(O')=M'_x O' {M'_x}^\dagger=O$, which completes the proof.
\end{proof}
It is well-known that relationalization maps cannot be physically invertible if the QRF is not ideal, even if they are only applied to observables on subsystems (as in Lemma~\ref{LemEuroYen} below). The $\text{\euro}$ map, however, \textit{is} invertible on the image of the physical Hilbert space, even for non-ideal QRFs. This is a distinctive feature of the PN approach.

As a special case, we consider relational Dirac observables between two subsystems, where one subsystem serves as the reference system. We associate commuting subalgebras $S,R\subset\cl(\ch_\kin)$ with the subsystems, with $R$ corresponding to the reference system. We assume that for every $M_x\in R$, its adjoint also satisfies $M_x^\dagger\in R$ (i.e.\ $R$ is a $^*$-subalgebra), and that $O\in S$. It then follows immediately that the associated effect density $E_x=M_x^\dagger M_x$ also belongs to $R$. Intuitively, the measurement of this POVM on the subsystem $R$ gives us an immediate notion of conditional state on $S$, regardless of how the measurement on $R$ is implemented. Formally, it allows us to rewrite the definition of $\text{\euro}_x^{\mathcal{I}}$ in terms of the effect density $E_x$:
\begin{align}\label{eq:rel_S_effect}
    \text{\euro}_x^{\mathcal{I}}(O)=&\int_\mathcal{G} d\mu(g)\, U_g^B I_x^\dagger(O) {U_g^B}^\dagger=\int_\mathcal{G} d\mu(g)\, U_g^B M_x^\dagger O M_x {U_g^B}^\dagger=\int_\mathcal{G} d\mu(g)\, U_g^B E_x O {U_g^B}^\dagger \nonumber\\
    =& \int_\mathcal{G} d\mu(g)\, U_g^B E_x {U^B_g}^\dagger U^B_g O {U_g^B}^\dagger=\int_\mathcal{G} d\mu(g)\,  E_{gx}  U^B_g O {U_g^B}^\dagger.
\end{align}

Let us return to the standard setting of QRFs where $\ch_\kin \cong \ch_R \otimes \ch_S$, $\Omega=\cg$ with $x_0=e$, and $U_g^B \cong U_R(g) \otimes U_S(g)$, and see that the $\text{\euro}$ map recovers the $\yen$ map. 

Given a covariant POVM $E_R$ acting on $\ch_R$ and an observable $O_\mathcal{S} \in \cl(\ch_S)$, the $\yen$ map is defined as~\cite{Loveridge2018}
\begin{align}
    \yen^R(O_\mathcal{S}) := \int_\cg U_S(g) O_S U_S(g)^\dagger \otimes E_R(dg).
\end{align}
\begin{lemma}[The \euro\enspace map generalizes the \yen\enspace map]\label{LemEuroYen}
In the standard setting just defined, set
$d\mu(g)\, E_g=E_R(dg)\otimes\mathds{1}_S$,
$O = \mathds{1}_R \otimes O_S$ and $x =e$ (the identity element of the group). Then ${\text{\euro}^\ci_e(O)} = \yen^R(O_S)$.
\end{lemma}
\begin{proof}
Simply substitute the expressions into Eq.~(\ref{eq:rel_S_effect}).
\end{proof}
Consider the simple special case where $\cg$ is a finite group and $\ch_R \cong \Cl^{|\cg|}$ carries the left regular representation, i.e.\ the QRF on $R$ is ideal in this sense. Let $M_g := \sqrt{|\mathcal{G}|}\ketbra{g}{g}_R \otimes \mathds{1}_S$, and the observable of interest to be a system observable, $O = \mathds{1}_R \otimes O_S$. Then
    \begin{align}
        \text{\euro}^\ci_e(O) =& \frac 1 {|\mathcal{G}|}\sum_{g \in \cg}  U_R(g) \otimes U_S(g)\sqrt{|\mathcal{G}|} (\ketbra{e}{e}_R \otimes \mathds{1}_S )( \mathds{1}_R \otimes O_S ) \sqrt{|\mathcal{G}|}(\ketbra{e}{e}_R \otimes \mathds{1}_S)(U_R(g)^\dagger \otimes U_S(g)^\dagger) \\
        =&  \sum_{g \in \cg} \ketbra{g}{g} \otimes U_S(g) O_S U_S(g)^\dagger = \yen^R(O_S).
    \end{align}
The $\yen$ map generalizes the $\cg$-twirl $O\mapsto \int_{\mathcal{G}} d\mu(g)\, U_g^B O {U_g^B}^\dagger$, and so does the $\text{\euro}$ map. This can be seen by choosing a trivial instrument $x = e$,  $M_e = \mathds{1}$ in the definition of $\text{\euro}_x^{\mathcal{I}}(O)$, and allowing arbitrary observables $O$ on $\ch_{\rm kin}$.

\section{Relation to the previous perspective-neutral framework}
\label{SecRelPrevious}

The perspective-neutral approach for general symmetry groups~\cite{Hamette2021} starts from a kinematical Hilbert space
\begin{align*}
\mathcal{H}_{\mathrm{kin}}
=
\mathcal{H}_R \otimes \mathcal{H}_S
\end{align*}
carrying a tensor product representation of a unimodular Lie group $\mathcal{G}$, $U_{RS}(g)=U_R(g)\otimes U_S(g)$. The PN framework of~\cite{Hamette2021} hence assumes that the symmetries act independently on associated subsystems, which is the main assumption that is dropped in our approach. We interpret $R$ as the reference system and $S$ as the system described relative to it. On $\mathcal{H}_R$, we have a distinguished coherent state system $\{|\phi(g)\rangle\}_{g\in\mathcal{G}}$, where all $|\phi(g)\rangle$ are normalized state vectors, such that
\begin{align*}
|\phi(g)\rangle_R
=
U_R(g)|\phi(e)\rangle_R,
\end{align*}
yielding a resolution of the identity
\begin{align}
\int_{\mathcal{G}} d'g\,|\phi(g)\rangle\langle\phi(g)|_R
=
\mathds{1}_R.
\label{EqResolutionIdentity}
\end{align}
In this notation, the normalization of the Haar measure is deliberately chosen such that Eq.~(\ref{EqResolutionIdentity}) is true; consequently, we may have $\int_{\mathcal{G}}d'g\neq 1$, and the notation indicates that $d'g$ of~\cite{Hamette2021} and our normalized Haar measure $dg$ will typically differ by a constant factor. In general, the states $|\phi(g)\rangle_R$ are not orthogonal, so the frame is not in general ideal. Here and in all of the examples to follow, the topological space $\Omega$ representing classical outcomes is given by the group $\mathcal{G}$ itself. We will now see that this formalism reduces to a special case of ours if the Lie group $\mathcal{G}$ is compact.

In~\cite{Hamette2021}, the physical Hilbert space is defined exactly as ours, i.e.\ $\ch_{\rm phys}=\{|\psi\rangle\in\ch_{\rm kin}\,\,|\,\,U_{RS}|\psi\rangle=|\psi\rangle\}$. However, for general unimodular Lie groups, this is not in general a proper subspace of $\ch_{\rm kin}$. Among other facts, this has the consequence that their projector $\Pi'_{\rm phys}$ into $\ch_{\rm phys}$, defined as $\Pi'_{\rm phys}:=\int_{\mathcal{G}}d'g U_{RS}(g)$, agrees in the compact case with ours, $\Pi_{\rm phys}=\int_{\mathcal{G}}dg\,U_{RS}(g)$, only up to a constant factor. We can interpret this as a convention of~\cite{Hamette2021} to choose a different inner product on $\ch_{\rm phys}$ rather than the one that is induced from $\ch_{\rm kin}$. Let us denote the physical Hilbert space with this modified inner product $(\cdot,\cdot)$ by $\ch'_{\rm phys}$. Internal frame descriptions are obtained from $\mathcal{H}'_{\mathrm{phys}}$ by reduction maps
\begin{align}
\label{eq:redMapsOld}
    R^{(g)}: \mathcal{H}'_\mathrm{phys} \rightarrow \mathcal{H}_{S,g}^\mathrm{phys}, \quad  R^{(g)}=\langle\phi(g)|_R\otimes\mathds{1}_S,
\end{align}
namely by conditioning on a frame orientation $g$ and eliminating the subsystem carrying the reference frame. Changes of frame are implemented as unitary quantum coordinate transformations between such reduced descriptions:
\begin{align*}
    V_{R_i \rightarrow R_j}(g_i,g_j)
=
R_{R_j}^{(g_j)}\,\cdot\,\left( R_{R_i}^{(g_i)}\right)^{-1}.
\end{align*}
We now have two reference frames as subsystem factors $\mathcal{H}_{R_i}$, with $i\in \{1,2\}$, i.e.\ the total Hilbert space is
\begin{equation}
\ch_{\rm kin}=\ch_{R_1}\otimes\ch_{R_2}\otimes\ch_S.
\label{EqTwoQRFs}
\end{equation}
The inverse of the reduction maps can be obtained as
\begin{align}
    \left( R_i^{(g)}\right)^{-1}: \mathcal{H}_{R_j,S,g}^\mathrm{phys}\rightarrow \mathcal{H}_\mathrm{phys}, \quad \left(R_i^{(g)}\right)^{-1}=\Pi'_\mathrm{phys}(|\phi(g)\rangle_{R_i}\otimes\mathds{1}_{R_j,S}).
    \label{EqOldQRFTrafo}
\end{align}
In \cite{Hamette2021}, the perspectival data are tied to a distinguished tensor-factor subsystem $\mathcal{H}_R$ carrying a coherent state system. By contrast, in our generalized setting, neither do we require the frame to be encoded in a subsystem of the form $\mathcal{H}_R$, nor that the relevant perspectives must be encoded by a coherent state system on such a factor. Instead, we set out with a Kraus seed on the full kinematical space, together with its orbit generated by the group action. This preserves the central theme of the perspective-neutral approach: a subspace encoding physical states together with unitary maps into internal perspectives, while dropping the assumption that the frame must arise from a distinguished subsystem carrying a coherent state POVM. In that sense, the construction in \cite{Hamette2021} is recovered as the special case in which the seed-generated orbit is realized by a coherent state system on $\mathcal{H}_R$.

Formally, the relation between the two formalisms is as follows. In our formalism, the reduction maps $R^{(g)}$ are replaced by the isometries
\begin{align}
\label{eq:ConnectionOldNewPN}
    M_g = \sqrt{n}|\phi(g)\rangle\langle\phi(g)|_R\otimes \mathds{1}_S\upharpoonright\ch_{\rm phys}=\sqrt{n}(|\phi(g)\rangle_R \otimes \mathds{1}_S)R^{(g)},
\end{align}
where $n>0$ is the constant such that $\int_{\mathcal{G}}|\phi(g)\rangle\langle\phi(g)|_Rdg=(1/n)\cdot\mathds{1}$, or equivalently, $n=d'g/dg$. Similarly, $M_g^\dagger=\Pi_{\rm phys}M_g$, where $\Pi_{\rm phys}$ is interpreted as a map from $\ch_{\rm kin}$ to $\ch_{\rm phys}$. That is, when passing to a perspective with the reduction maps, rather than eliminating the subsystem carrying the reference frame, we keep the associated factor intact in the perspectival description. Eq.~(\ref{EqResolutionIdentity}) then ensures that the coherent state system defines a valid QRF $\mathcal{I}$ in the sense of our Definition~\ref{DefQRF}, with trivial one-dimensional representation $\chi=1$. The representation spaces in both approaches are hence related by
\[
   \ch_{\rm phys}^{\mathcal{I}}(g)=|\phi(g)\rangle_R\otimes \ch_{S,g}^{\rm phys}.
\]
When comparing the reduction maps, care has to be taken that the inner product on $\ch'_{\rm phys}$ is different in~\cite{Hamette2021}, namely $(\varphi,\psi)=\langle \varphi|\psi\rangle/n$ for $\varphi,\psi\in\ch'_{\rm phys}$. Therefore, if we pick a normalized vector $\psi\in\ch_{\rm phys}$, i.e.\ $\langle\psi|\psi\rangle=1$, then the associated normalized vector in $\ch'_{\rm phys}$ is $\psi':=\sqrt{n}\psi$, and we obtain
\begin{equation}
   M_g|\psi\rangle=|\phi(g)\rangle_R\otimes (R^{(g)}\psi')_S.
   \label{eqFactorLeftIn}
\end{equation}
In this sense, our reduction maps fully recover those of~\cite{Hamette2021} in the case of compact Lie groups. Moreover, our QRF transformations of Definition~\ref{DefQRFTrafo} are equivalent to those of Eq.~(\ref{EqOldQRFTrafo}). However, they are strictly more general even in the special case of tensor product subsystem QRFs: for example, they allow us to switch between perspectival descriptions associated to two different choices of coherent state system $\{|\phi(g)\rangle\}_{g\in\mathcal{G}}$ and $\{|\phi'(g)\rangle\}_{g\in\mathcal{G}}$ for a single subsystem QRF $R$, and they admit Kraus rank 1 covariant instruments where the associated operators on $R$ are not projectors onto coherent states, but more general matrices allowing the associated POVMs on the reference to be of rank greater than 1. A concrete instance of this is shown in Example~\ref{ExLossUnitarity} below.

Furthermore, our Lemma~\ref{LemOurTheorem1} generalizes Theorem 1 of~\cite{Hamette2021}. We expect that several further results of~\cite{Hamette2021} have precise analogues in our more general framework, but leave an elaboration of this to future work.

\section{Examples of the generalized framework}
\label{SecExamples}

\subsection{Generalized aspects of the perspective-neutral formalism for finite Abelian groups}
\label{SubsecAbelian}

We will begin to illustrate examples of the generalized approach within the perspective-neutral setting for finite Abelian groups \cite{KrummHoehnMueller, HoehnKrummMueller}. While this is in principle covered by the formalism of \cite{Hamette2021}, it already exhibits instances of the main aspects that are generalized in the present framework.

We consider the group $\mathbb{Z}_n=\{0,1,\ldots,n-1\}$ with addition modulo $n$ as the group operation. We interpret the numbers as the elements of a discretized position space for a particle along a circle, and addition as translation. For a single particle, the associated Hilbert space is $\mathcal{H}_1=\mathbb{C}^{|\mathcal{G}|}$, i.e.\ an $n$-dimensional Hilbert space with orthonormal basis $\{|g\rangle\}_{g\in\mathcal{G}}$. For $N$ particles, the total (kinematical) Hilbert space is $\ch_{\rm kin}=\ch_1^{\otimes N}$.

\begin{figure}[hbt]
\includegraphics[width=.35\columnwidth]{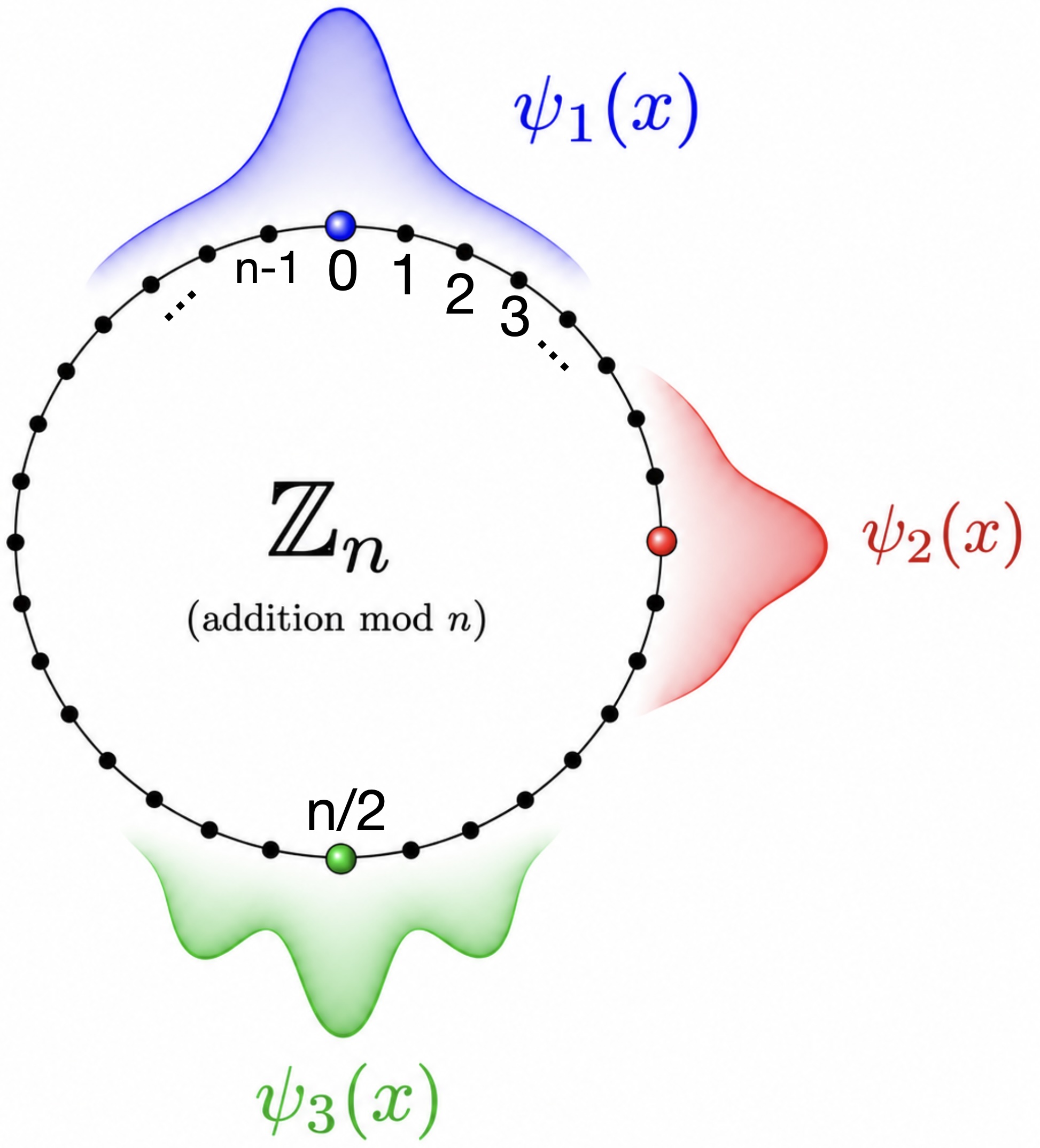}
\caption{The finite Abelian group $\mathbb{Z}_n$ describes a discretized position space for $N$ particles (here $N=3$), with addition modulo $n$ as translation. There is also an associated notion of momentum, and a discrete Fourier transform that relates the two. Hence, this provides a mathematically simple discrete version of the usual ``particle on a line'', such that our results are directly applicable in a mathematically rigorous way. \aidisclosure}
\label{fig:abelian}
\end{figure}
The representation of the translation group $\mathcal{G}$ is $U_g^B=U_g^{\otimes N}$, such that $U_g^B|x_1,x_2,\ldots,x_N\rangle=|x_1+g,x_2+g,\ldots,x_N+g\rangle$, where $+$ denotes addition modulo $n$. An important role is played by the state vectors $|\mathbf{d};p\rangle$, labelled by relative distances $\mathbf{d}=(d_1,\ldots,d_{N-1})$ with $d_i=x_{i+1}-x_1\mod n$, and a generalized notion of momentum $p\in\{0,1,\ldots,n-1\}$. Formally, the expressions
\[
   \chi_p(g)=\exp\left(\frac{2\pi i} n p g\right)\qquad (p\in\{0,1,\ldots,n-1\})
\]
are characters of $\mathcal{G}$, and they appear in the discrete Fourier transform definition
\begin{align*}
|\mathbf{d};p\rangle
:=
\frac{1}{\sqrt{n}}
\sum_{g=0}^{n-1}
\overline{\chi_p(g)}\, |g,\mathbf{d}+g\rangle,
\end{align*}
where $|g,\mathbf{d}+g\rangle=|g,d_1+g,d_2+g,\ldots,d_{N-1}+g\rangle$. Note that $\chi_{p=0}=\mathbf{1}$, the constant-one function, and the physical Hilbert space is the momentum-zero subspace
\begin{align*}
\mathcal{H}_{\mathrm{phys}}
=
\mathrm{span}\bigl\{\, |\mathbf{d};0\rangle \;\big|\; \mathbf{d} \in \mathcal{G}^{N-1} \,\bigr\}.
\end{align*}

The first aspect of relevance for the present paper is that already in the finite Abelian setting -- where particles are reference systems associated to subsystems -- one does \emph{not} need to implement the passage to the perspective of one particle as eliminating the corresponding tensor factor. Earlier works~\cite{KrummHoehnMueller,HoehnKrummMueller} have defined Schrödinger reduction maps
\begin{align*}
R_{S,i}^{(g)}
:
\mathcal{H}_{\mathrm{phys}}
\to
\mathcal{H}_{\bar{i}},
\qquad
R_{S,i}^{(g)}
=
\sqrt{n}\,\langle g|_i \otimes \mathds{1}_{\bar{i}},
\end{align*}
where $\ch_{\bar i}\simeq\ch_1^{\otimes(N-1)}$ is the Hilbert space of all particles but $i$. This has been interpreted as conditioning on the gauge-fixing condition of particle $i$ being in orientation $g$. However, if one wishes to keep the frame subsystem explicit, then the isometries of a QRF according to our Definition~\ref{DefQRF} are instead
\begin{align*}
M^{(i)}_g
:=
\sqrt{n}|g\rangle\langle g|_i \otimes \mathds{1}_{\bar{i}} \upharpoonright \ch_{\rm phys} \qquad
(g \in \mathcal{G}).
\end{align*}

\begin{example}[Frame of the first particle]
Take $\mathcal{G} = \mathbb{Z}_n$ and choose particle $i$ as the frame. Then the aligned seed is $M^{(i)}_e=\sqrt{n}|0\rangle\langle 0|_i \otimes \mathds{1}_{\bar{i}}\upharpoonright \ch_{\rm phys}$, while the translated family is $M^{(i)}_g=\sqrt{n}|g\rangle\langle g|_i \otimes \mathds{1}_{\bar{i}}\upharpoonright\ch_{\rm phys}$ for $g\in\mathbb{Z}_n$. For $i=1$, we can also write 
\begin{align*}
M_g^{(1)}=&\sqrt{n}|g\rangle\langle g|_1 \otimes \mathds{1}_{\bar{1}}\upharpoonright \mathcal{H}_\mathrm{phys}
=\sqrt{n}|g\rangle\langle g|_1 \otimes \sum_\mathbf{d}|\mathbf{d}+g\rangle\langle \mathbf{d}+g|_{\bar{1}} \Pi_\mathrm{phys} \upharpoonright\ch_{\rm phys}=\sum_{\mathbf{d}}|g,\mathbf{d}+g\rangle \langle \mathbf{d};0|
\end{align*}
for the form of the isometry corresponding to the frame of the first particle at position $g$. Setting $g=0$ as well, we see that $M_e^{(1)}$ acts on the basis elements of the physical Hilbert space as
\begin{align*}
|\mathbf{d};0\rangle \longmapsto |0\rangle_1 \otimes |\mathbf{d}\rangle_{\bar 1}=|0,d_1,...,d_{N-1}\rangle,
\end{align*}
so the physical relational state is mapped to the description of the remaining particles relative to particle $1$ being at the origin. As in our generalized construction, the initial orientation is selected by $M^{(i)}_e$, and the family of internal perspectives is generated by its orbit under the group action.

The finite Abelian framework also makes the linking structure between different internal perspectives explicit. The quantum coordinate change from the perspective of particle $i$ at position $g_i$ to that of particle $j$ at position $g_j$ is
\begin{align}
\label{eq:TransBetwParticles}
V_{i \to j}(g_i,g_j)
=
\sum_{g \in \mathcal{G}}
|g_i-g\rangle\langle g_i|_i
\otimes
|g_j\rangle\langle g+g_j|_j
\otimes
U_{-g}^{\otimes (N-2)}
\qquad
(i \neq j),
\end{align}
where all additions and subtractions are, as always here, modulo $n$. For example, applying $V_{1 \rightarrow 2}(e,e)$ to a state describing the first particle's perspective at $0$ will map
\begin{align*}
    |0\rangle_1 \otimes |\mathbf{d}\rangle_{\bar 1} \longmapsto |0\rangle_2 \otimes |\mathbf{d}-d_1\rangle_{\bar 2},
\end{align*}
and switch to the second particle's perspective at $0$. It is straightforward to check that the quantum coordinate changes of Eq.~(\ref{eq:TransBetwParticles}) are equivalent to our QRF transformations of Definition~\ref{DefQRFTrafo} for the special case of the QRFs that are defined by the $M_g^{(i)}$ above.
\end{example}

In the following example, we consider a scenario which goes beyond the formalism of~\cite{Hamette2021}. Instead of a single coherent state system on the reference system, we introduce a linear combination of two coherent state systems -- essentially constructing a smeared version of the initially sharp frame leading to superposed perspectival states, which yields an isometry defining a QRF in the sense of Definition~\ref{DefQRF}. While still fulfilling our definition of being a valid seed for a QRF, this does not represent a valid coherent state system any more. An implication of this is that reduction maps which eliminate the reference system, as in (\ref{eq:redMapsOld}), may lead to a loss of unitarity. While the following example illustrates the phenomenon purely mathematically, we give a physically motivated case of a similar construction (for relativistic clocks) in Example~\ref{ExPWdouble} below.

\begin{example}[Frame smearing via superposition of perspectives]
\label{ExLossUnitarity}
Consider the above scenario with $N=2$ particles on $n\geq 3$ discrete positions, where the first particle is interpreted as the reference $R$, and the second particle as the system $S$. With the $M_g^{(1)}$ of above, describing a state $\rho\in\mathcal{L}(\ch_{\rm phys})$ relative to ``$R$ being in position $g$'' yields, according to the PN framework of~\cite{Hamette2021,HoehnKrummMueller,KrummHoehnMueller},
\[
   R_{S,1}^{(g)}\rho {R_{S,1}^{(g)}}^\dagger=n\langle g|_R\otimes\mathds{1}_S\rho |g\rangle_R\otimes\mathds{1}_S=n\,{\rm Tr}\left(|g\rangle\langle g|_R\otimes\mathds{1}_S \rho |g\rangle\langle g|_R\otimes\mathds{1}_S\right)={\rm Tr}_R\left(M_g^{(1)}\rho{M_g^{(1)}}^\dagger\right).
\]
The associated map $R_{S,1}^{(g)}$ is an isometry despite the partial trace since the result of conjugating with $M_g^{(1)}$ -- a coherent state system -- disentangles $R$ from $S$ by projecting it into a local pure state. In our framework, and as shown in Eq.~(\ref{eqFactorLeftIn}) above, we leave the factor $|g\rangle_R$ in the perspectival state rather than tracing it out.

Omitting the partial trace is indeed necessary for our framework to admit unitary QRF changes, as we will now see. Consider the linear combination of the two seeds at $g=0$ and $g=2$,
\begin{align*}
    M'_0:=&\frac{1}{\sqrt{2}}(M_0^{(1)}+M_2^{(1)}) \upharpoonright \ch_\phys = \frac{1}{\sqrt{2}}\left(\sum_{\mathbf d}\ketbra{0, {\mathbf d}}{\mathbf{d};0} + \sum_{\mathbf d}\ketbra{2, {\mathbf d + 2}}{\mathbf{d};0}\right) = \frac{1}{\sqrt{2}}\sum_{\mathbf d}(\ket{0,{\mathbf d}} + \ket{2, {\mathbf d + 2}})\bra{\mathbf{d};0}.
\end{align*}
Since $M_0^{(1)}$ and $M_2^{(1)}$ are orthogonal projectors, $M'_0$ is still a valid Kraus seed for a covariant instrument, even though it is not a projector on a coherent state. Furthermore, the cross terms vanish after tracing out $R$, and so
\begin{equation}
   {\rm Tr}_R\left(M'_0\rho{M'_0}^\dagger\right)=\frac n 2 {\rm Tr}_R\left(|0\rangle\langle 0|_R\otimes\mathds{1}_S\rho|0\rangle\langle 0|_R\otimes\mathds{1}_S+|2\rangle\langle 2|_R\otimes\mathds{1}_S\rho|2\rangle\langle 2|_R\otimes\mathds{1}_S\right)=\frac 1 2 R_{S,1}^{(0)}\rho R_{S,1}^{(0)} + \frac 1 2 R_{S,1}^{(2)}\rho R_{S,1}^{(2)}.
   \label{EqNotUnitary}
\end{equation}
Thus, the generalized reduction associated with $M'_0$ and subsequent elimination of the reference factor induces the uniform mixture of the two projecting reduction maps associated to each coherent state system, which indicates that $M'_0|\psi\rangle$ for $|\psi\rangle\in\ch_{\rm phys}$ is typically entangled between $R$ and $S$. It is clear that the map in Eq.~(\ref{EqNotUnitary}) hence does not in general correspond to a unitary map. In our framework, this loss of unitarity is avoided since the reference system $R$ is not traced out. 

Essentially, we can view $M'_0$ as a smearing of the frame generated by $M^{(1)}_0$. While for the former we always have orthogonal, and thus sharply resolved perspectival states $\langle \psi |(M^{(1)}_g)^\dagger M^{(1)}_{g'}|\psi \rangle = \delta_{g,g'}$, states such as $M'_0|\psi\rangle$ and $M'_2|\psi\rangle$ are overlapping. For $\mathcal{G}=\mathbb{Z}_4$, for example, these would coincide entirely, and there are only two distinct perspectival states available to resolve the four possible frame orientations in that case.
\end{example}

In the finite Abelian setting, the authors of~\cite{KrummHoehnMueller,HoehnKrummMueller} have identified a group of symmetry transformations which is much larger than the ordinary translation group $\mathcal{G}$, which the authors have interpreted as a group of QRF transformations. With the orthogonal projectors $\Pi_{\mathbf{d}}$ onto the subspaces of constant particle distances, $\ch_{\mathbf{d}}={\rm span}\{|\mathbf{d};p\rangle\,\,|\,\,0\leq p \leq n-1\}$, this group is
\begin{align*}
\mathcal{U}_{\mathrm{sym}}
=
\left\{
\sum_{\mathbf{d}\in\mathcal{G}^{N-1}} U_{g(\mathbf{d})}^{\otimes N}\Pi_{\mathbf{d}}
\,\,\middle|\,\,
g(\mathbf{d})\in\mathcal{G}
\right\}.
\end{align*}
These are relation-conditional translations: on each sector of fixed particle distances, one applies a global translation that may depend on the distances themselves. Hence, $\mathcal{U}_{\mathrm{sym}}$ is strictly larger than the original translation group $\mathcal{G}$, and it contains it when we set $g(\mathbf{d})=g$ for all $\mathbf{d}$. At the same time, this larger group already contains all transformations of the form (\ref{eq:TransBetwParticles}) between the internal particle perspectives. 

Furthermore, we can see immediately that $\mathcal{U}_\mathrm{sym}$ can be recovered from our definition of quantum coordinate changes, given by (\ref{eq:QuantCoordChanges}). Pick any seed $M_0$ such as the one associated to passing to the frame of the first particle. Now, for any $U \in \mathcal{U}_\mathrm{sym}$, observe that $M_0':=UM_0$ (or $M'_0:=U M_0 U^\dagger$ which is the same since $U^\dagger$ acts as the identity on $\ch_{\rm phys}$) is still a valid seed according to our Definition~\ref{DefQRF}, because $[U_g^B,U]=0$ for all $U\in\mathcal{U}_{\rm sym}$ and all $g\in\mathcal{G}$. It is then clear that we recover every such $U$ as a QRF transformation in the sense of our Definition~\ref{DefQRFTrafo}, if we denote the QRFs associated to the $M_g$ and the $M'_g$ by $\mathcal{I}$ and $\mathcal{I}'$, respectively:
\[
   U\upharpoonright \ch_{\rm phys}^{\mathcal{I}}(0)=U_{0,0}^{\mathcal{I},\mathcal{I}'} = M'_0 M_0^\dagger\upharpoonright \ch_{\rm phys}^{\mathcal{I}}(0).
\]
Now we give an example of a QRF in our sense which cannot be associated with any of the canonical subsystems in this setting, even though it is physically very well-motivated: the description relative to the center of mass of the particles.

\begin{example}[Center-of-mass frame]
\label{ex:centOfMass}
In the scenario described above, let $m_1,\dots,m_N$ be non-negative real numbers with $m := m_1 + \cdots + m_N > 0$, which we interpret as the particle masses. For relative positions $\mathbf{d}=(d_1,\dots,d_{N-1})=(x_2-x_1,x_3-x_1,\ldots,x_N-x_1)$, define
\begin{align*}
g(\mathbf{d})
:=
-
\left\lfloor
\frac{1}{m}
\left(
m_2 d_1 + \cdots + m_N d_{N-1}
\right)
\right\rfloor,
\end{align*}
which we can interpret as the center of mass, up to rounding the numbers with the floor function which is necessary because we have discrete integer positions. Set
\begin{align*}
U:=\sum_{\mathbf{d}\in\mathcal{G}^{N-1}} U_{g(\mathbf{d})}^{\otimes N}\Pi_{\mathbf{d}},
\end{align*}
and note that this is a special case of the QRF transformations that we have described in Lemma~\ref{LemCC}. This conditional translation $U$ describes a change of quantum coordinates from the perspective of particle $1$ to the perspective of the center of mass.

Now, we start from the usual seed for the frame of particle $1$,
\begin{align*}
M^{(1)}_0 = \sqrt{n}|0\rangle\langle 0|_1 \otimes \mathds{1}_{\bar 1}\upharpoonright\ch_{\rm phys} =\sqrt{n}\sum_\mathbf{d}|0,\mathbf{d}\rangle\langle0,\mathbf{d}|\upharpoonright\ch_{\rm phys}, 
\end{align*}
and define
\begin{align*}
   M_0^{(\mathrm{cm})}:=UM_0^{(1)}=U\sqrt{n}|0\rangle\langle 0|_1 \otimes \mathds{1}_{\bar 1}U^\dagger\upharpoonright \ch_{\rm phys}=\sqrt{n}\sum_{d}|g(\mathbf{d}),\mathbf{d}+g(\mathbf{d})\rangle\langle g(\mathbf{d}),\mathbf{d}+g(\mathbf{d})|\upharpoonright\ch_{\rm phys}.
\end{align*}
Since $U$ acts as a translation $U_{g(\mathbf{d})}^{\otimes N}$ on each fixed particle distance subspace $\ch_{\mathbf{d}}$, the projector $|0\rangle\langle 0|_1$ is shifted sector by sector to the projector onto the $\mathbf{d}$-dependent position $g(\mathbf{d})$ of particle $1$ that centers the coordinates at the center of mass. For the simplest case of two particles with equal masses, where $\mathbf{d}=d$ is a single distance $d$, we obtain
\begin{align*}
    M^{(\mathrm{cm})}_0 = \sqrt{n}
    \sum_{d}|-\lfloor d/2\rfloor,d-\lfloor d/2\rfloor\rangle\langle-\lfloor d/2\rfloor,d-\lfloor d/2\rfloor|\,\upharpoonright\ch_{\rm phys}.
\end{align*}
Thus, rather than conditioning on particle $1$ being at a fixed orientation $0$, one conditions on particles $0$ and $1$ having mutually inverse positions (only for even $d$, due to the discreteness). This already shows, within the finite Abelian setting, that the relevant seed need not be associated with a fixed subsystem factor alone.

It is true that the total kinematical Hilbert space $\ch_{\rm kin}$ can be refactored into the center-of-mass and the relational degrees of freedom, and the QRF associated to the center of mass can hence be interpreted as pertaining to a tensor product subsystem in this more general sense (this is shown in more detail in Example~\ref{ex:PageWoottersContinuum} below). However, this refactorization will not be possible in all cases where our Definition~\ref{DefQRF} applies. And even in the center of mass example where it \emph{is} possible, the two QRFs associated to the seeds $M_0^{(1)}$ and $M_0^{\rm (cm)}$ are associated to \emph{distinct} factorizations of the Hilbert space $\ch_{\rm kin}$. Hence, the associated QRF transformations $U$ map between two QRFs that are \emph{not} associated to two tensor factors of a \emph{single} factorization as in Eq.~(\ref{EqTwoQRFs}) of the earlier PN framework of~\cite{Hamette2021}. This shows yet another way in which our framework is more general, even in the simple case of finite Abelian groups.
\end{example}

\subsection{Describing bosons and fermions within the generalized PN framework}
\label{SubsecBosonsFermions}

We will now show that our generalized PN framework is a natural formalism to describe bosons and fermions; and that, indeed, the symmetrization postulate can be seen as a particular instance of this.

For simplicity, let us work in the formalism of first quantization, and consider a total Hilbert space for $N$ (initially thought of being distinguishable) particles, $\ch_{\rm kin}=\ch^{\otimes N}$, where $\ch$ is a single-particle Hilbert space.
The symmetry group is the symmetric group of $N$ elements, i.e.\ $\mathcal{G}=S_N$, where $S_N$ contains all permutations of $\{1,2,\ldots,N\}$. Now we go beyond the earlier PN frameworks by considering a representation of $\cg$ that does not act independently on each factor, but rather permutes the tensor factors:
\[
   U_\pi^B|x_1,\ldots,x_n\rangle=|x_{\pi^{-1}(1)},\ldots, x_{\pi^{-1}(N)}\rangle \qquad (\pi\in S_N).
\]
It now turns out that the physical Hilbert space becomes
\[
   \ch_{\rm phys}=\{|\psi\rangle\in\ch_{\rm kin}\,\,|\,\, U_\pi^B|\psi\rangle=|\psi\rangle\}=\vee^N\ch,
\]
i.e.\ the symmetric subspace. Alternatively, we could have started with the projectively equivalent representation
\[
   {U_\pi'}^B|x_1,\ldots,x_n\rangle={\rm sgn}(\pi)|x_{\pi^{-1}(1)},\ldots, x_{\pi^{-1}(N)}\rangle={\rm sgn}(\pi)U_\pi^B \qquad (\pi\in S_N),
\]
and the associated physical Hilbert space is
\[
   \ch'_{\rm phys}=\{|\psi\rangle\in\ch_{\rm kin}\,\,|\,\, U_\pi^B|\psi\rangle={\rm sgn}(\pi)|\psi\rangle\}=\wedge^N\ch,
\]
i.e.\ the antisymmetric subspace. Hence, whatever leads one to the postulate that $\ch_{\rm phys}$ describes the physically relevant states (for example, our motivation of Subsection~\ref{SubsecOpMot}) will also lead one to arrive at the symmetrization postulate for bosons resp.\ fermions. Note that the PN formalism does \textit{not} apply to the case of parastatistics, i.e.\ higher-dimensional irreducible representations of $S_N$ that appear in the decomposition of $\ch_{\rm kin}$. This inapplicability can be lifted to a general physical argument for the absence of paraparticles in nature, as some of us have shown in~\cite{mekonnen_invariance_2026}. In this sense, our generalized PN framework allows us to understand important aspects of indistinguishable particles in physics, which can be seen as a part of a larger research program to formulate empirically adequate and predictively powerful \textit{quantum covariance principles}~\cite{GiacominiBrukner2020,GiacominiBrukner2022,QRFIndefiniteMetric,Kabel2024,hardy_implementation_2020}. 
If the symmetric and antisymmetric subspaces describe the perspective-neutral states of bosons and fermions, then what does it mean to ``jump into a perspective''? Let us first consider a simple example, and then discuss what one could do with it.

\begin{example}[Labeling frames for two bosons]
\label{ExLabeling}
    We consider a system of two bosons in first quantization whose positional degrees of freedom are described by a Hilbert space $\mathcal{H}_\mathrm{kin}\cong\mathcal{H}_1 \otimes \mathcal{H}_1 \cong \mathds{C}^2 \otimes \mathds{C}^2$, with position basis $\{\ket{0},\ket{1}\}$ for each factor. The symmetry group in this example is $\mathcal{S}_2=\{e,\mathrm{swap}\}$, and its unitary action exchanges tensor factors as $U(\mathrm{swap})\ket{\phi}\otimes \ket{\psi}=\ket{\psi}\otimes \ket{\phi}$. The perspective-neutral states of the physical Hilbert space satisfy $U(\mathrm{swap})\ket{b}=\ket{b}$; they correspond exactly to the symmetrized states of bosons. In this scenario, the physical Hilbert space $\mathcal{H}_\mathrm{phys}$ is
    \begin{align*}
        \ch_{\rm phys}=\ch_1\vee\ch_1=\mathrm{span}\left\{\ket{0,0},\ket{1,1},\frac{1}{\sqrt{2}}\left(\ket{0,1}+\ket{1,0}\right)\right\}.
    \end{align*}
We associate a QRF for indistinguishable particles to a choice of a labeling convention for the particles. To this end, we can start with a set of isometries $\{M_e,M_\mathrm{swap}\}:\ch_{\rm phys}\to\ch_{\rm kin}$ given by
    \begin{align*}
        M_e &= |0,0\rangle\langle0,0|+|1,1\rangle\langle1,1|+\sqrt{2}|0,1\rangle\langle0,1|\upharpoonright\ch_{\rm phys} \\
        &= |0,0\rangle\langle 0,0|+|1,1\rangle\langle1,1|+ |0,1\rangle\langle\Phi^+|,\\
        M_\mathrm{swap} &= |0,0\rangle\langle0,0|+|1,1\rangle\langle1,1|+\sqrt{2}|1,0\rangle\langle1,0| \upharpoonright\ch_{\rm phys}=U_\mathrm{swap}M_eU_\mathrm{swap}^\dagger\upharpoonright \ch_{\rm phys}\\
        &=|0,0\rangle\langle0,0|+|1,1\rangle\langle1,1|+|1,0\rangle\langle\Phi^+|=U_{\rm swap} M_e,
    \end{align*}
where $\ket{\Phi^+} = \frac{1}{\sqrt{2}}(\ket{01}+\ket{10})$.
These define a QRF $\mathcal{I}$ according to Definition~\ref{DefQRF}. Intuitively, this QRF breaks permutation symmetry: depending on the associated instrument's outcome, it yields a description that labels the particles either in ascending or descending order, as we will now see. For instance, the ascending frame description
    \begin{align*}
    \ket{\psi}_\uparrow:=\ket{0,1}=M_e\frac{1}{\sqrt{2}}\left(|0,1\rangle +|1,0\rangle\right)    
    \end{align*}
would be interpreted as labeling the particle at position 0, say Alice's position in a lab, as the ``first'' particle, while the one at position 1 (say, Bob's) is what we refer to as the ``second'' particle. For the descending frame description
\begin{align*}
    \ket{\psi}_\downarrow:=\ket{1,0}=M_\mathrm{swap}\frac{1}{\sqrt{2}}\left(|0,1\rangle +|1,0\rangle\right)    
    \end{align*}
    it is exactly the other way around. Note that this is also an instance of a non-ideal quantum reference frame. As a consequence of bosons occupying multiple positions, we would have $\langle\psi|M_\mathrm{swap}^\dagger M_e|\psi\rangle\neq 0$ for some $\ket{\psi}\in \mathcal{H}_\mathrm{phys}$ (e.g.\ for $|\psi\rangle=|00\rangle$), since doubly occupied states imply that different frame descriptions are overlapping. In contrast, the fermionic analog corresponds to an ideal QRF. For fermions, $\ch_{\rm phys}$ is spanned by $\ket{\Phi^-}$, so we could have $M_e = \ketbra{01}{\Phi^{-}}$  and $M_{\rm swap} = \ketbra{10}{\Phi^{-}}$ which would correspond to an ideal QRF in the sense that $M_g \ket{\psi} \perp M_{g'} \ket{\psi}$ for $g \neq g'$ and all $|\psi\rangle\in\ch_{\rm phys}$. For the general definition of an ``ideal QRF'' in this context, we refer the reader to~\cite{upcoming}.
\end{example}
To see what this can be used for, we have to go slightly beyond this simple example. Suppose, for example, that the single-particle Hilbert space is spanned by orthonormal states $|0\rangle,|1\rangle,\ldots,|9\rangle$ rather than only $|0\rangle,|1\rangle$. Consider the two symmetric states
\[
   |\psi\rangle=\frac 1 {\sqrt{2}}\left(|0,8\rangle+|8,9\rangle\right),\qquad |\psi'\rangle=\frac 1 2 \left(|0,8\rangle+|8,0\rangle+|1,9\rangle+|9,1\rangle\right)
\]
which are both in the physical Hilbert space. Treated as elements of $\ch_{\rm kin}$ with respect to its two-particle tensor product factorization, both are entangled states, but this statement is not operationally meaningful. Suppose that we interpret the numbers $0,1,\ldots, 9$ as spatial positions, and consider two agents, Alice and Bob, at opposite ends of these possible positions. That is, we assume that the places $0$ and $1$ are located in Alice's lab, and the places $8$ and $9$ in Bob's. This gives us \emph{operationally well-defined notions of subsystems}, and the question of entanglement now becomes physically meaningful. To answer this question, we can use the ``labeling QRF'' from Example~\ref{ExLabeling} above and compute
\begin{eqnarray*}
    |\psi\rangle_{\uparrow}=M_e|\psi\rangle=|0,8\rangle,\qquad |\psi'\rangle_{\uparrow}=\frac 1 {\sqrt{2}}\left(|0,8\rangle+|1,9\rangle\right).
\end{eqnarray*}
Relative to our ``choice of perspective'', resembling the internal breaking of permutation symmetry induced by Alice's and Bob's labs, we obtain state descriptions that allow us to conclude that we should regard $\psi$ as separable and $\psi'$ as entangled. In upcoming work~\cite{upcoming}, we explain in more detail how labeling QRFs in the PN framework can be used to describe the entanglement of bosons and fermions, and how this relates to the standard description in terms of subalgebras~\cite{benatti_entanglement_2020}.

\subsection{Page-Wootters mechanism with frames for non-local clocks and constraints with degenerate spectra}
\label{SubsecPW}

We will now consider an example involving two particles of the same mass which are acted on by one-dimensional translations, together with a harmonic-oscillator interaction term between them. If the interaction vanishes the usual Page-Wootters mechanism \cite{PageWootters} emerges: conditioned on the position of the first particle taking the role of a clock, the evolution of the second one is generated by a time-independent Hamiltonian. With the non-vanishing interaction term, on the other hand, the effective generating Hamiltonian becomes time-dependent, which does not result in a decoupling between the dynamics of the clock and the system. However, we will see that by using the center-of-mass frame as a non-local clock, we recover the Page-Wootters mechanism again, and a relational description in terms of a time-\emph{independent} Hamiltonian, despite the presence of the interaction term. We refer the reader to Section \ref{sec:PageWoottersReview} in the appendix for a review of the Page-Wootters formalism and its extension in the interaction case \cite{Smith_2019}.

Technically, the present setting lies outside of our framework, since we will be assuming the locally compact group $(\mathbb{R},+)$ represented regularly on the infinite-dimensional space $L^2(\mathbb{R})$, in order to underline the importance of this non-local clock in a setting where the time parameter is continuous and aperiodic. As a consequence, we will be ignoring the correct normalization of states and operators in this setting, while remaining precise otherwise. In addition to the example below, however, Example \ref{ex:PageWoottersFiniteAbelian} in the appendix gives the finite Abelian analogue of this example, showing that our generalized formalism can still accommodate such a scenario rigorously.

\begin{example}[Non-local Page-Wootters clock]
\label{ex:PageWoottersContinuum}
    The kinematical Hilbert space is given by $\mathcal{H}_{\mathrm{kin}} = \mathcal{H}_1 \otimes \mathcal{H}_2=L^2(\mathbb{R}) \otimes L^2(\mathbb{R})$, with position operators $X_i$ acting by multiplication, momentum operators $P_i = -i\partial_{x_i}$ in the position basis, satisfying $[X_i, P_j] = i\delta_{ij}$, and the translation operator acts as
\begin{align*}
    e^{-iP_i a}|x\rangle_i = |x + a\rangle_i.
\end{align*}
Note that in an expression such as $e^{-iP_1a}|x\rangle_1$ we only consider the action of $P_1$ on the first tensor factor, even though $P_1$ as an operator acts on both factors, albeit trivially on the second. Similarly simplified notation will be used in all of the following.

The constraint Hamiltonian is the continuous version of the finite Abelian generator of global translations together with a harmonic oscillator potential term:
\begin{align}
\label{eq:HContinuum}
    H = P_1 + P_2 + K(X_1 - X_2)^2, \qquad K > 0.
\end{align}

We note that the relative-distance coordinate is invariant under global translations, so $[P_1+P_2,(X_1-X_2)^2] = 0$ and $e^{-iHt}$ factorizes as a product of these parts. However, the evolution generated by $H$ does not factorize in the $1\otimes 2$ tensor decomposition, because the interaction $K(X_1-X_2)^2$ mixes the two factors.

In order to determine the physical states, we will solve the analog of the constraint equation $H|\psi\rangle = 0$. Introducing center-of-mass and relative-distance coordinates $X = (x_1 + x_2)/2$ and $r = x_1 - x_2$, with $x_1 = X+r/2$ and $x_2 = X-r/2$, the chain rule gives $\partial_{x_1} + \partial_{x_2} = \partial_X$, so in these coordinates we obtain
\begin{align*}
    H = -i\partial_X + Kr^2.
\end{align*}
The constraint $H\psi = 0$ becomes a first-order ordinary differential equation in $X$ at each fixed $r$,
\begin{align*}
    \partial_X\psi(X, r) = -iKr^2\psi(X, r),
\end{align*}
with solution
\begin{align}
\label{eq:PhysStateContinuum}
    \psi(X, r) = e^{-iKr^2 X}f(r)
\end{align}
The physical Hilbert space is thus parametrized by functions $f(r)$, depending only on the relative-distance coordinate. In the original coordinates the same physical state translates back to
\begin{align}
\label{eq:PhysStateContinuumPosCoord}
    \psi(x_1, x_2) = e^{-iK(x_1-x_2)^2(x_1+x_2)/2}f(x_1-x_2).
\end{align}
The gauge action given by $U_t = e^{-i(P_1+P_2)t}\cdot e^{-iK(X_1-X_2)^2 t}$, acts on a position basis state as
\begin{align}
\label{eq:UContinuum}
    U_t|x_1, x_2\rangle = e^{-iK(x_1-x_2)^2 t}|x_1+t,x_2+t\rangle,
\end{align}
and we find $U_t|\psi\rangle = |\psi\rangle$ for every $|\psi\rangle \in \mathcal{H}_{\mathrm{phys}}$, as expected.

Now, with the local  frame of the first particle given by
\begin{align*}
    M_0^{(1)} &= |x_1=0\rangle\langle x_1=0|_1 \otimes \mathds{1}_2 \upharpoonright\mathcal{H}_\mathrm{phys},\\
    M^{(1)}_t&:=U_tM_0^{(1)}=U_t |x_1=0\rangle\langle x_1=0|_1 \otimes \mathds{1}_2 U_t^\dagger \upharpoonright\mathcal{H}_\mathrm{phys}=|x_1=t\rangle\langle x_1=t|_1 \otimes \mathds{1}_2 \upharpoonright\mathcal{H}_\mathrm{phys},
\end{align*}
the reduced state at $t = 0$ follows from restricting the physical wavefunction \eqref{eq:PhysStateContinuumPosCoord}
to $x_1 = 0$. This gives $\psi(0, x_2) = e^{-iKx_2^3/2}f(-x_2)$, so the wavefunction of the reduced state in the local frame at $t = 0$ is $(M^{(1)}_0\psi)(x_1,x_2)=\delta(x_1)e^{-iKx_2^3/2}f(-x_2)$, i.e.\ ``particle 1 at the origin'' multiplied with the system wavefunction
\begin{align}
\label{eq:PhiLocalContinuumZero}
    \psi_S^{(1)}(x_2; t=0) := e^{-iKx_2^3/2}f(-x_2)
\end{align}
on $\mathcal{H}_2$. 

We evolve the reduced state with $U_t$ by first acting with  the translation $e^{-i(P_1+P_2)t}$
\begin{align*}
    e^{-i(P_1+P_2)t}\left(\delta(x_1)e^{-iKx_2^3/2}f(-x_2)\right)
    = \delta(x_1-t)e^{-iK(x_2-t)^3/2}f(t-x_2),
\end{align*}
and  then the interaction term $e^{-iK(X_1-X_2)^2 t}$:
\begin{align*}
    e^{-iK(X_1-X_2)^2 t} \delta(x_1-t)e^{-iK(x_2-t)^3/2}f(t-x_2) &= \delta(x_1-t) e^{-iK(X_1-X_2)^2 t} e^{-iK(x_2-t)^3/2}f(t-x_2) \\ &= \delta(x_1-t) e^{-K(t-x_2)^2(t+x_2)/2}f(t-x_2),
\end{align*}
where the multiplicative phase comes from $e^{-iK(X_1-X_2)^2 t}$ evaluated at the support of the delta, $x_1 = t$, yielding the additional factor $e^{-iK(t-x_2)^2 t}$.
The two phases combine as $-K(t-x_2)^2 t - K(x_2-t)^3/2 = -K(t-x_2)^2(t+x_2)/2$, and the perspectival state in the first particle's local frame at time $t$ amounts to
\begin{align}
\label{eq:PhiLocalContinuum}
    \psi^{(1)}(x_1, x_2;t)
    =
    \delta(x_1 - t)\cdot\psi_S^{(1)}(x_2;t),
    \qquad
    \psi_S^{(1)}(x_2;t) = e^{-iK(t-x_2)^2(t+x_2)/2}f(t-x_2).
\end{align}
Thus, as the first particle (viewed as our clock) advances from $x_1 = 0$ to $x_1 = t$, the relational evolution of the system's wavefunction on $\mathcal{H}_2$ can be tracked.

To find the form of the Schrödinger equation that $\psi^{(1)}_S$ satisfies, we define $A(x_2, t) := (t-x_2)^2(t+x_2)/2$ so that $\psi_S^{(1)} = e^{-iKA}f(t-x_2)$, and compute the partial derivatives
\begin{align*}
    \partial_t A &= \frac{1}{2}\left(2(t-x_2)(t+x_2) + (t-x_2)^2\right) = \frac{(t-x_2)(3t+x_2)}{2}, \\
    \partial_{x_2}A &= \frac{1}{2}\left(-2(t-x_2)(t+x_2) + (t-x_2)^2\right) = -\frac{(t-x_2)(t+3x_2)}{2}.
\end{align*}
For the time derivative of the system wavefunction, we then obtain
\begin{align}
\label{eq:idtPsiSLocal}
    i\partial_t\psi_S^{(1)}
    =
    i\partial_t\left(e^{-iKA}f(t-x_2)\right)
    =
    K\frac{(t-x_2)(3t+x_2)}{2}\psi_S^{(1)} + ie^{-iKA}f'(t-x_2),
\end{align}
while the derivative with respect to $x_2$ reads
\begin{align}
\label{eq:dx2PsiSLocal}
    i\partial_{x_2}\psi_S^{(1)}
    =
    -K\frac{(t-x_2)(t+3x_2)}{2}\psi_S^{(1)} - ie^{-iKA}f'(t-x_2).
\end{align}
Adding \eqref{eq:idtPsiSLocal} and \eqref{eq:dx2PsiSLocal} and rearranging yields
\begin{align*}
    i\partial_t\psi_S^{(1)}
    &=
    K\frac{(t-x_2)}{2}\left((3t+x_2) - (t+3x_2)\right)\psi_S^{(1)} - i\partial_{x_2}\psi_S^{(1)} \\
    &= K(t-x_2)^2\psi_S^{(1)} -i\partial_{x_2}\psi_S^{(1)}.
\end{align*}
Viewed as acted on by an operator, the system wavefunction in the local frame therefore obeys the Schrödinger equation
\begin{align}
\label{eq:TDSELocalContinuum}
    i\partial_t \psi_S^{(1)}(x_2;t) = \left(P_2 + K(t-X_2)^2\right)\psi_S^{(1)}(x_2;t)=:H^{(1)}_{\mathrm{eff}, S}(t)\psi_S^{(1)}(x_2;t),
\end{align}
 with an effective Hamiltonian
\begin{align}
    H^{(1)}_{\mathrm{eff}, S}(t):= P_2 + K(t-X_2)^2
\end{align}
that depends explicitly on the time parameter.

We also note that the unitary operator evolving the system state from time $0$ to $t$ can be read off from the action of $U_t$ on \eqref{eq:PhiLocalContinuumZero} leading to \eqref{eq:PhiLocalContinuum}. We find
\begin{align*}
    \psi^{(1)}_S(x_2;t)=e^{-iK(t-X_2)^2t}e^{-iP_2t}\psi^{(1)}_S(x_2;0)=:U^{(1)}_{\mathrm{eff},S}(t,0)\psi^{(1)}_S(x_2;0)
\end{align*}
and thus $U^{(1)}_{\mathrm{eff},S}(t_2,t_1):=U^{(1)}_{\mathrm{eff},S}(t_2,0)\left(U^{(1)}_{\mathrm{eff},S}(t_1,0)\right)^\dagger$.

We now turn to the center-of-mass  frame, the continuum analog of the non-local projectors that Example \ref{ex:centOfMass} introduced for the finite Abelian setting:
\begin{align}
\label{eq:COMSeedContinuum}
M_0^{(\mathrm{cm})} &= \int_{\mathbb{R}}
    dx|x,-x\rangle\langle x,-x|_{12}
    \upharpoonright \mathcal{H}_{\rm{phys}}, \\ M^{(\mathrm{cm})}_t&:=U_tM_0^{(\mathrm{cm})}=\int_{\mathbb{R}}
    dxU_t|x,-x\rangle\langle x,-x|_{12}U_t^\dagger
    \upharpoonright \mathcal{H}_{\rm{phys}}=\int_{\mathbb{R}}
    dx|t+x,t-x\rangle\langle t+x,t-x|_{12}
    \upharpoonright \mathcal{H}_{\rm{phys}}. \nonumber
\end{align}
It is the restriction of the projector onto the subspace $\{(x_1, x_2)\in\mathbb{R}^2 : x_1 + x_2 = 0\}$ of position states, which is non-local in the $1\otimes 2$ tensor decomposition, in the sense that it cannot be written as $A_1\otimes B_2$ for any operators $A_1, B_2$ acting on $\mathcal{H}_1, \mathcal{H}_2$ separately. Applying it to the physical state $\psi(x_1, x_2) = e^{-iK(x_1-x_2)^2(x_1+x_2)/2}f(x_1-x_2)$ restricts the wavefunction such that the phase $e^{-iK(x_1-x_2)^2(x_1+x_2)/2}$ vanishes. The reduced state at $t = 0$ in the center-of-mass frame is therefore
\begin{align*}
    \left(M_0^{(\mathrm{cm})}\psi\right)(x_1, x_2) = \delta(x_1 + x_2)f(x_1 - x_2).
\end{align*}
 
Evolving this state with \eqref{eq:UContinuum}, the translation $e^{-i(P_1+P_2)t}$ shifts both coordinates by $+t$ and sends $\delta(x_1+x_2)$ to $\delta(x_1 + x_2 - 2t)$ while leaving the relative argument $x_1-x_2$ of $f$ unchanged, and the operator $e^{-iK(X_1-X_2)^2 t}$ acts as $e^{-iK(x_1-x_2)^2 t}$ on the wavefunction. The perspectival state in the center-of-mass frame at time $t$ is therefore
\begin{align}
\label{eq:PhiCMContinuum}
    \psi^{(\mathrm{cm})}(x_1, x_2;t) = \delta(x_1 + x_2 - 2t)e^{-iK(x_1-x_2)^2 t}f(x_1-x_2).
\end{align}
As the non-local center-of-mass clock advances from $x_1+x_2 = 0$ to $x_1+x_2 = 2t$, we can now track the relational evolution of the system, i.e. the part of the wavefunction which is parameterized by the relative-distance coordinate. Expressing \eqref{eq:PhiCMContinuum} in center-of-mass/relative-distance coordinates, we obtain
\begin{align}
\label{eq:PhiCMContinuumCMrel}
    \psi^{(\mathrm{cm})}(X, r;t) = \frac{1}{2}\delta(X - t)e^{-iKr^2 t}f(r)
\end{align}
with the system wavefunction
\begin{align*}
    \psi_S^{(\mathrm{cm})}(r;t) := \frac{1}{2}e^{-iKr^2 t}f(r),
\end{align*}
which depends on $t$ only linearly in the phase. The equation of motion is then simply evaluated as
\begin{align*}
    i\partial_t\psi_S^{(\mathrm{cm})}(r;t)
    =
    i\partial_t\left(\frac{1}{2}e^{-iKr^2 t}f(r)\right)
    =
    i(-iKr^2)\frac{1}{2}e^{-iKr^2 t}f(r)
    =
    Kr^2\psi_S^{(\mathrm{cm})}(r;t),
\end{align*}
i.e.\ the Schrödinger equation
\begin{align}
\label{eq:TISECMContinuum}
    i\partial_t\psi_S^{(\mathrm{cm})}(r;t)
    =
    K(X_1-X_2)^2\psi_S^{(\mathrm{cm})}(r;t)
\end{align}
generated by the time-independent Hamiltonian $H_{\mathrm{eff},S}^{(\rm cm)} = K(X_1-X_2)^2$, which coincides with the interaction term of the constraint in \eqref{eq:HContinuum} and generates the time evolution operator $U^{(\mathrm{cm})}_{\mathrm{eff},S}(t_2,t_1):=e^{-iH^{(\mathrm{cm})}_{\mathrm{eff},S}(t_2-t_1)}$.

Furthermore, one could also express the center-of-mass seed \eqref{eq:COMSeedContinuum} in the center-of-mass/relative-distance tensor decomposition $\mathcal{H}_\mathrm{cm}\otimes \mathcal{H}_\mathrm{rel} \cong \mathcal{H}_1\otimes \mathcal{H}_2$, in which it takes the form
\begin{align*}
    M^{\mathrm{(cm)}}_0=|X=0\rangle\langle X=0|_\mathrm{cm}\otimes \mathds{1}_\mathrm{rel}\upharpoonright\mathcal{H}_\mathrm{phys},
\end{align*}
such that \eqref{eq:PhiCMContinuumCMrel} splits naturally into clock and system factors in that basis. Writing the frame seed in this particular factorization allows us to see clearly why we recover the usual form of the Page-Wootters mechanism despite adding an interaction term. Since $H$ in \eqref{eq:HContinuum} splits into local, commuting terms $H_\mathrm{cm} \otimes \mathds{1}_\mathrm{rel}:=P_1+P_2$ and $ \mathds{1}_\mathrm{cm}\otimes H_\mathrm{rel}:=K(X_1-X_2)^2$, we obtain 
\begin{align*}
    i\partial_t|\psi^{(\mathrm{cm})}_S(r;t)\rangle&=i\partial_t (\langle X=t|_\mathrm{cm}\otimes \mathds{1}_\mathrm{rel})|\psi\rangle 
    = (\langle X=t|_\mathrm{cm}\otimes \mathds{1}_\mathrm{rel})(-H_\mathrm{cm}\otimes \mathds{1}_\mathrm{rel})|\psi\rangle \\
    &= (\langle X=t|_\mathrm{cm}\otimes \mathds{1}_\mathrm{rel})(\mathds{1}_\mathrm{cm}\otimes H_\mathrm{rel})|\psi\rangle 
    = H_\mathrm{rel}|\psi^{(\mathrm{cm})}_S(r;t)\rangle.
\end{align*}
\end{example}

In the previous example, we have seen that in the presence of the interaction term $K(X_1-X_2)^2$, taking the first particle as a clock will not result in its  decoupling from the rest of the configuration, as opposed to the non-interacting case. However, the center of mass as a non-interacting clock will still lead to a  decoupling between it and the relative-distance part of the configuration, due to the associated split of the constraint Hamiltonian. Our framework allows us to make use of this by choosing the non-local reference frame corresponding to the center of mass, which is aligned with the clean decomposition of the Hamiltonian, thereby again recovering the relational evolution of the Page-Wootters construction. While reduction maps to the center-of-mass frame could not be defined previously, we note that refactorizations of the kinematical Hilbert space within the PN framework were investigated in \cite{carrozzaCorrespondenceQuantumError2025}. In our generalization, however, the center of mass can be understood -- with respect to  the fixed initial tensor factorization -- as a non-local frame.

Now, while a time-dependence in the Hamiltonian generator of the time evolution is usually understood as coming from some external, potentially unspecified influence on the system of interest, the relational picture offers an explanation of this effect from a different angle. In the Schrödinger picture, an explicit dependence of a system's observable on the clock's parameter is the result of the interaction between that clock and the system. This kind of interaction was understood as ``kick back'' in \cite{Angelo_2012}, where it is connected to fictitious forces of non-inertial frames. Thus, such an interacting clock's perspectival state does not only track the system's evolution, but also its own correlation with the system, leading to an \textit{indirect} occurrence of self-reference. This effect, as well as the conditions for it to arise for switches between non-interacting clocks, are described in \cite{Trinity,Equivalence_Hoehn_2021}. Furthermore, we note that this is independent of the ideality properties of the clock, as our example illustrates. Both the first particle's and the center-of-mass system have clock states which are pairwise orthogonal and generated by Hamiltonian constraints $P_1$ and $P_1+P_2$, respectively, and with non-degenerate spectra given by $\mathbb{R}$. 

A more general treatment of non-ideal Page-Wootters clock constraints, such as quadratic Hamiltonians with kinetic-potential terms, can be found in \cite{hoehnHowSwitchRelational2020,Trinity, Equivalence_Hoehn_2021, hosseini2026timearrivalproblempagewootters}, and we will see in Example \ref{ExPWdouble} below that our formalism gives such clocks a clear interpretation as QRFs as well. Besides the connection to the Page-Wootters formalism, constraints with degeneracies are, for instance,  discussed in \cite{Ahmad2022} in the context of the quantum relativity of subsystems. For a treatment of the center-of-mass and relative-distance partition in connection with (non)-inertial QRFs, we refer the reader to~\cite{Angelo_2012} and \cite{Vanrietvelde2020}, the latter of which explicitly includes harmonic oscillator potentials between systems of particles.

We will now discuss how clock constraints with two-fold degenerate spectra, such as $H_C=\frac{P^2}{2m}$, can be accommodated naturally within our framework. This has been treated in \cite{Equivalence_Hoehn_2021,hosseini2026timearrivalproblempagewootters}, where separate reduction maps were defined on the positive and negative frequency sectors of the physical Hilbert space. Following closely the analysis of \cite{Equivalence_Hoehn_2021}, the example below will show that this sector-wise construction, while going beyond the standard PN formalism, defines a valid QRF in our generalized approach. Similarly as Example~\ref{ex:PageWoottersContinuum}, it is slightly beyond our mathematical framework because it involves a non-compact group acting on an infinite-dimensional Hilbert space. In particular, we will ignore all normalization factors for states in the following example, because most states appearing in the calculations are not even normalizable. Similarly as in~\cite{Equivalence_Hoehn_2021} on which this rests, we expect that the details can be filled in (e.g.\ via rigged Hilbert spaces) without fundamental difficulties.
\begin{example}[Page-Wootters clock as QRF from quadratic constraint]
\label{ExPWdouble}
    We consider a constraint equation of the form $H|\psi\rangle=\left(\frac{P^2}{2}\otimes \mathds{1}_S+\mathds{1}_C\otimes H_S\right)|\psi\rangle=0$, assuming the spectrum of $H_S$ to have non-positive elements such that $\mathcal{H}_\mathrm{phys}$ is non-empty. Kinematical states are of the form
    \begin{align*}
        |\psi_\mathrm{kin}\rangle = \sumint_E\int_\mathbb{R}dp\; \psi_\mathrm{kin}(p,E)|p\rangle_C|E\rangle_S,
    \end{align*}
    and by rewriting the constraint as
    \begin{align*}
        H=H_+\cdot H_- \quad \mathrm{with} \quad H_\sigma = \frac{P}{\sqrt{2}}+\sigma\sqrt{-H_S}, 
    \end{align*}
    where $[H_+,H_-]=0$ and $\sigma \in \{\pm\}$ labels the frequency degeneracies of $H_C=P^2/2m$, one obtains a decomposition of the physical Hilbert space as $\mathcal{H}_\mathrm{phys}=\bigoplus_{\sigma=\pm}\mathcal{H}_\sigma$ into sectors. Physical states are then given by  
    \begin{align*}
        |\psi_\mathrm{phys}\rangle = \sum_\sigma \sumint_{E \in \mathrm{\sigma_{SC}}}\frac{\psi_\sigma(E)}{(2|E|)^{1/4}}|p_\sigma(E)\rangle_C|E\rangle_S,
    \end{align*}
    where $\psi_\sigma(E):=\psi(p_\sigma(E),E)/(2|E|)^{1/4}$, $p_\sigma(E):=-\sigma\sqrt{2|E|}$ and $\sigma_{SC}:=\{E \in \mathrm{Spec}(H_S)|E\leq 0\}$. Since we only consider the non-positive part of the spectrum of $H_S$ on $\mathcal{H}_\mathrm{phys}$, $\sqrt{-H_S}\upharpoonright\ch_{\rm phys}$ above is self-adjoint.
    
    In the momentum basis the clock states are defined as
    \begin{align*}
        |t,\sigma \rangle := \int_\mathbb{R}dp\sqrt{|p|}\theta(-\sigma p)e^{-itp^2/2}|p\rangle_C, \quad \mathrm{with} \quad |t+t',\sigma\rangle = e^{-itP^2/2}|t',\sigma\rangle,
    \end{align*}
    where $\theta$ denotes the Heaviside step function. While these states are only mutually orthogonal for differing values of $\sigma$, they do form resolutions of the identity (up to a normalization) on each frequency sector of $\mathcal{H}_C$, respectively. Thus, their sum fulfills 
    \begin{align*}
        \frac{1}{2\pi}\sum_\sigma\int_\mathbb{R}dt |t,\sigma\rangle\langle t,\sigma|=\mathds{1}_C.
    \end{align*}
    With the help of sector-wise reduction maps $\mathcal{R}^\sigma_\mathrm{PW}:\mathcal{H}_\mathrm{phys}\rightarrow\mathcal{H}^\mathrm{phys}_{S,\sigma}$ defined as $\mathcal{R}^\sigma_\mathrm{PW}(\tau):= \langle \tau,\sigma| \otimes \mathds{1}_S$, as well as their inverses $(\mathcal{R}^\sigma_\mathrm{PW}(\tau))^{-1}:\mathcal{H}^\mathrm{phys}_{S,\sigma}\rightarrow\mathcal{H}_\sigma \subset \mathcal{H}_\mathrm{phys}$, the authors of \cite{Equivalence_Hoehn_2021} recover exactly the standard PN-formalism per frequency sector, and a Page-Wootters relational evolution law. With the reduced states
    \begin{align*}
    \frac{1}{\sqrt{2}}|\psi^\sigma_S(\tau)\rangle := \mathcal{R}^\sigma_\mathrm{PW}(\tau)|\psi_\mathrm{phys}\rangle 
    = \sumint_{E\in \sigma_{SC}}\psi_\sigma(E)e^{-i\tau E}|E\rangle_S,
    \end{align*}
    they obtain
    \begin{align*}
        i\frac{d}{d\tau}|\psi^\sigma_S(\tau)\rangle=H_S|\psi^\sigma_S(\tau)\rangle.
    \end{align*}
    Now, as the authors of~\cite{Equivalence_Hoehn_2021} explain (in their Footnote 16), the standard PN framework of \cite{Hamette2021} as outlined in Sec. \ref{SecRelPrevious} does not admit a notion of a frame that describes the combined reduction for both frequency sectors, as the superposition of seed states such as $|t\rangle:=1/\sqrt{2}(|t,+\rangle + |t,-\rangle)$ is not a coherent seed state any more. A reduction map of the form $\mathcal{R}_\mathrm{PW}:=\langle \tau | \otimes \mathds{1}_S$ would not lead to a resolution of the identity as it mixes contributions from both frequency sectors, and it would not be invertible. However, with the identification given by Eq. (\ref{eq:ConnectionOldNewPN}), we find that
    \begin{align*}
        M_t := &\left(|t,+\rangle\langle t,+|+|t,-\rangle\langle t,-|\right)\otimes \mathds{1}_S\upharpoonright\ch_{\rm phys} \\
        = &\left(|t,+\rangle \otimes \mathds{1}_S \right) \mathcal{R}^+_\mathrm{PW}(t)+\left(|t,-\rangle \otimes \mathds{1}_S \right) \mathcal{R}^-_\mathrm{PW}(t)
    \end{align*}
    is a Kraus seed corresponding to a covariant instrument. The QRF defined by this construction can be interpreted as the clock that conditions on the rank-2 coherent POVM coming from both frequency states. With the results for the Page-Wootters reduction maps of~\cite{Equivalence_Hoehn_2021}, we find
    \begin{align*}
        M_t|\psi_\mathrm{phys}\rangle=\left(|t,+\rangle \otimes |\psi^+_S(t)\rangle+|t,-\rangle \otimes |\psi^-_S(t)\rangle \right).
    \end{align*}
This result shows how the two reduction maps for the frequency sectors are combined into a single isometry in our generalization of the PN framework. The two solutions $|\psi_S^\sigma(t)\rangle$, obtained separately in~\cite{Equivalence_Hoehn_2021}, are now jointly encoded into the two orthogonal frequency sectors (since $\langle t,+|t,-\rangle=0$). What has been an ad hoc generalization in~\cite{Equivalence_Hoehn_2021} is now (up to the infinite-dimensionality issue) a systematic part of our framework.
\end{example}
Note that the construction of the frame given by $\{M_t\}_t$ above is similar to the finite, Abelian Example \ref{ExLossUnitarity}, where a linear combination of two coherent state seed projectors was used to construct a new QRF. While we identified a type of smearing of the frame in Example \ref{ExLossUnitarity}, this would not be an appropriate interpretation here. Due to the combination of frames $M^{(\sigma)}_t:=\left(|t,\sigma\rangle \otimes \mathds{1}_S \right) \mathcal{R}^\sigma_\mathrm{PW}(t)$ operating entirely on orthogonal subspaces, i.e.\ the frequency sectors $\mathcal{H}_\sigma$ and their images, no additional overlap of perspectival states can be created in this case. For simplicity, consider the above example in an ideal scenario (for instance for finite groups), where all different perspectival states are orthogonal: $\langle \psi_\mathrm{phys}|(M^{(\sigma)}_{\mathrm{id},g})^\dagger M^{(\sigma')}_{\mathrm{id},g'}|\psi_{\mathrm{phys}}\rangle=\delta_{\sigma,\sigma'}\delta_{g,g'}$. Then, a combination such as $M_{\mathrm{id},g}:=M^{(+)}_{\mathrm{id},g}+M^{(-)}_{\mathrm{id},g}$ would preserve the sharpness of perspectival states $M_{\mathrm{id},g}|\psi_\mathrm{phys}\rangle$ exactly, and not introduce any additional non-ideality. Furthermore, it is shown in \cite{Equivalence_Hoehn_2021} that Dirac observables are superselected across the frequency sectors; hence coherences of physical states $|\psi_\mathrm{phys}\rangle \in \mathcal{H}_\mathrm{phys}=\mathcal{H}_\mathrm{+}\oplus \mathcal{H}_\mathrm{-}$ between the sectors do not influence expectation values, and superpositions of different sector states are indistinguishable from the corresponding mixed state. Tracing out the reference system would thus not necessarily lead to a loss of physically relevant information, as opposed to Example \ref{ExLossUnitarity}.

\section{Outlook and Conclusions}
\label{SecConclusions}
Our contribution in this paper is threefold: first, we have given a resource-theoretic and operational reconstruction of the perspective-neutral (PN) framework. Second, based on this, we have generalized the PN framework beyond the case of quantum reference frames (QRFs) as subsystems on which the symmetry acts independently. Third, we have shown that this generalized framework unlocks a variety of physically relevant applications.

Let us dicsuss these three aspects separately.\\

\textbf{Resource-theoretic motivation.} It has previously been noticed that the notion of QRFs is intimately linked to operational constraints --- concretely, to the question of which (algebra of) observables is accessible to an experimenter~\cite{Bartlett2007,Doat2025, garmier_perspectives_2025,CastroRuiz2025}. For example, some agents will only have access to certain subsystems. Moreover, in the presence of a $\cg$-symmetry, one might conclude that only $\cg$-invariant observables are measurable, a postulate that has a prominent motivation via the WAY theorem~\cite{wigner_messung_1952,araki_measurement_1960,yanase_optimal_1961}. However, previous works tended to postulate the set of accessible observables, or the set of associated QRF transformations, in an ad hoc manner, based on specific examples.

We think that resource-theoretic notions are the right tools to study such restrictions in a  systematic, clear and consistent way. Resource theories have first been introduced in quantum information theory, where the resource theory of entanglement describes what observers can accomplish if they are restricted to LOCC operations (local operations and classical communication)~\cite{horodecki_quantum_2009}, they have been used to put thermodynamics on a new rigorous foundation~\cite{horodecki_fundamental_2013,brandao_second_2015}, and they have led to the study of a \textit{resource theory of asymmetry}~\cite{Bartlett2007,MarvianThesis} where the free operations are restricted to be covariant with respect to the symmetry group. Resource-theoretic notions have been used to study the asymmetry induced by an internal QRF~\cite{Ludescher2022}, and applying resource theories to quantum field theory has been described as one frontier on the way to a the successful application of quantum information theory to high-energy physics~\cite{goto_rethinking_2026}.

In Subsection~\ref{SubsecOpMot}, we have given a simple operational scenario for which resource-theoretic considerations, understood in a broad sense, lead naturally to several aspects of the PN framework. In more detail, we identified the notion of \emph{completely covariant operations} as the free operations of relevance in this context, and this naturally leads to the definition of the physical Hilbert space. Expressed in technical jargon, this gives an operational motivation for the appearance of \emph{coherent twirling}, $\Pi_{\rm phys}:=\int_\cg U_g^B\, dg$, as opposed to the \emph{incoherent twirling} that appears naturally in the resource theory of asymmetry, $\rho\mapsto\int_\cg U_g^B \rho {U_g^B}^\dagger$.

From a physical point of view, complete covariance is a more ``paranoid'' version of symmetry than plain covariance: it expresses the postulate that gauge transformations on a system $S$ are not only required to preserve all physical predictions for measurements performed on $S$, but also for correlated measurements performed on $SE$, where $E$ is any external quantum system. If the symmetry group $\mathcal{G}$ is a connected Lie group, then $|\psi\rangle\in\ch_{\rm phys}$ is the same as $C|\psi\rangle=0$ for all generators $C$ of the representation. Hence, insofar as complete covariance serves as a motivation of the PN framework and its physical Hilbert space, it has \emph{also} the potential to suggest an explanation of why one would want to write down constraint equations like the ones that appear in Dirac quantization.

An early formulation of (a special case of) complete covariance in~\cite{HoehnKrummMueller}, arguing for the invariance of purifications of states, has been criticized in~\cite{Doat2025}: \textit{``[...] it would appear more logical to ban these `bad purifications', rather than all the states to which they could potentially be applied.''} Let us briefly respond to this criticism, which will also make the analogy with \textit{complete positivity} from quantum information theory clearer from which the name derives.

On any quantum system $X$, a linear trace-preserving positive map $T:X\to X$ is called \textit{completely positive} if $T_X\otimes {\rm Id}_Y:XY\to XY$ is also positive for every finite-dimensional quantum system $Y$. The usual argumentation for why all physical maps must be completely positive is as follows. Suppose that $T$ is positive, but not completely positive, such as the transposition map on a qubit $X$. Then we could prepare an entangled two-qubit state $|\psi\rangle$ on $XY$ (where $Y$ is another qubit) such that $T_X\otimes {\rm Id}_Y$ maps this state to a non-state, i.e.\ to an operator $T_X\otimes {\rm Id}_Y(|\psi\rangle\langle\psi|)$ that has negative eigenvalues, which is absurd.

Now, similarly as the authors of~\cite{Doat2025}, one could raise the following objection: \textit{``It would appear more logical to ban these pure entangled states, rather than some of the transformations that could potentially be applied on them.''} And indeed, this is a valid objection: our world \textit{could} potentially contain qubit systems which can never become entangled with other quantum systems, and on which the transpose map can be physically implemented. However, experience tells us that most qubit systems are not of this sort. Regardless of the empirical question of which types of qubits actually exist in the world, this objection does not invalidate the significance of complete positivity: \textit{if} there exist qubit systems that may be entangled with other quantum systems, \textit{then} complete positivity is a crucial constraint.

Similarly, complete covariance is a stronger version of covariance that applies to physical systems for which it turns out to be possible to hold a purification on an external system that is consistent with the local reduced state. Our claim is \textit{not} that all systems are necessarily of this type, or that the perspective-neutral approach is the only valid formulation of quantum reference frames -- only that \textit{if} we have systems of this type, \textit{then} the perspective-neutral approach is a natural framework to apply. Our reconstruction of the PN framework adds further evidence that it is a well-motivated approach, but this is consistent with our view that all approaches to QRFs have significance in their own context of applicability.

Despite the insights that we think our motivating scenario of Subsection~\ref{SubsecOpMot} provides, we have not yet given a definition of a complete, compositional resource theory where the completely covariant operations are the free operations. It would be interesting to construct a resource theory of this kind in future work. Furthermore, our motivation of the PN framework in Subsection~\ref{SubsecOpMot} still rested on some plausibility arguments and was not based on purely axiomatic, logical deductions. Further and more rigorous insights might be gained by analyzing the scenario as part of a multipartite system or of an operational theory, which again motivates the search for a complete associated resource theory and the analysis of its compositional structure.  \\

\textbf{Generalizing QRFs beyond subsystems.} Previous approaches have assumed that QRFs are subsystems on which the relevant symmetry group acts independently. For example, in earlier formulations of the PN framework including~\cite{Hamette2021}, a tensor product factorization of the kinematical Hilbert space into $R$ and $S$ is postulated, and the symmetry group is assumed to act as $U(g)=U_R(g)\otimes U_S(g)$. Similarly, the Page-Wootters construction~\cite{PageWootters} assumes that system and clock do not interact. However, these assumptions are arguably unrealistic, and some approaches to quantum gravity suggest that independent subsystems (say, described by local commuting subalgebras) do not even fundamentally exist~\cite{donnelly_observables_2016}.

This has motivated us to go beyond QRFs as subsystems when reconstructing the PN framework. We think that the restriction to subsystems is not only unwarranted, but that it also indicates a conceptual glitch in the physical intuition that underlies some approaches to QRFs: that we should think of a QRF as a ``thing'', or as a material ``object'' (more operationally, a ``system'') such as a particle; and that we should be interested in ``the perspective taken by that thing'', describing e.g.\ ``how the particle (which might be itself in superposition) sees the world''. We believe that this is a naive framing of the internal QRF research program. A perhaps clearer and more fruitful way to think of an internal QRF is to see it as \emph{an internal structural feature of the quantum system that allows us to break the gauge symmetry, and to choose one particular representation (among many equivalent ones) of the quantum state or observable}. This can mean to look at a non-interacting subsystem and declaring it to be the clock, as in the previous approaches; but it can also mean, for example, to track the evolution of the center of mass and use it as a clock, or to sort identical particles according to their spatial location and use this as a ``labeling frame'' that breaks the permutation symmetry.

When QRFs are subsystems $R$ (as in previous approaches), then measuring them produces a well-defined conditional state on the other subsystems $S$. In the previous formulation of the PN framework~\cite{HametteGalley}, this measurement was assumed to be given by a coherent state system on $R$. Even with the subsystem assumption (say, that $\ch_{\rm kin}$ factorizes into $R$ and $S$ as tensor product subsystems), our framework already provides a substantial generalization: it admits general covariant POVMs on the QRF subsystem (Lemma~\ref{LemSubsystemQRFs}), therefore lifting the PN framework to a level of generality (for compact Lie groups) that is already present in other approaches such as the operational approach to QRFs~\cite{Loveridge2018,Loveridge2017,Carette2025,glowacki_quantum_2024,glowacki_w-algebraic_2026}. However, when dropping the subsystem assumption, one does not only have to specify the POVM that is measured (and whose result will break the gauge symmetry), but also what the associated post-measurement states are. This is what has forced us to consider \textit{covariant instruments} as QRFs. As we have shown in Lemma~\ref{LemEuroYen}, this leads to a natural generalization of the relationalization map (a $\text{\euro}$ map, generalizing previous $\$$ and $\yen$ maps) to instrument QRFs that applies not only to the PN framework, but more generally. In the PN framework, this will be a covariant instrument with a density that has Kraus rank $1$ (Definition~\ref{DefQRF}). The associated Kraus operators play a double role: they describe both the action of the measurement on the input state, and they are isometries from the physical to the kinematical Hilbert space. They lead to unitary (that is, physically invertible) QRF transformations between different perspectives (Definition~\ref{DefQRFTrafo}), and the relationalization map $\text{\euro}$ becomes invertible on the image of the physical Hilbert space (Lemma~\ref{LemOurTheorem1}), generalizing results on relational Dirac observables of the previous approach~\cite{Hamette2021}.

Furthermore, we have given examples of general QRF transformations, including gauge transformations that are coherently controlled by (i.e.\ conditioned on) gauge-invariant observables~(Lemma~\ref{LemCoherentlyControlled}), which have previously been studied as prototypical QRF transformations~\cite{QRFIndefiniteMetric,GiacominiBrukner2022,GiacominiBrukner2020,hardy_implementation_2020,Giacomini,HametteGalley,mekonnen_invariance_2026}.\\

\textbf{Applications.} We have shown that our generalization of the PN framework to instruments enables a variety of physically relevant applications. For example, it allows us to introduce ``frames of labeling'' for bosons and fermions (see Subsection~\ref{SubsecBosonsFermions}). Here, the relevant symmetry is permutation symmetry, and the associated physical Hilbert space is either the symmetric (for bosons) or antisymmetric subspace (for fermions). Hence, whatever motivates the physical Hilbert space, either in the PN framework or in Dirac quantization, also motivates the well-known symmetrization postulate for indistinguishable particles, and our operational reconstruction in Subsection~\ref{SubsecOpMot} may perhaps play this role. In another work~\cite{mekonnen_invariance_2026}, some of us have shown how this observation can be turned into arguments for ruling out fundamental parastatistics in nature, in agreement with empirical results.

In Example~\ref{ExLabeling} and following, we have argued that this formalism is useful for studying entanglement of bosons and fermions: the form of the state vector in the physical Hilbert space (symmetrized or antisymmetrized) says nothing about its entanglement; but ``jumping into the labeling perspective'' that is associated with the operational scenario of interest will lead to a representation of the state that does. We will elaborate on this in future work~\cite{upcoming}, which will also discuss its relation to the standard algebraic approaches such as~\cite{benatti_entanglement_2020}.

Our generalization of QRFs to covariant instruments is also relevant for quantum clocks and the Page-Wootters framework. In a world where all subsystems are interacting (as expected in quantum gravity~\cite{donnelly_observables_2016}), it is usually assumed that we would still pick a subsystem as a relational clock, but that this results in corrections to the Schr\"odinger equation~\cite{Smith_2019}. However, our approach points at an alternative possibility: choose a (``non-local'') instrument QRF as the clock, similarly as the global distribution of dust is used as a clock in cosmology~\cite{brown_dust_1995}, and if this instrument is covariant, it defines time in a way that reproduces the Schr\"odinger equation exactly. In Example~\ref{ex:PageWoottersContinuum}, we have given a simple toy model that does this by jumping into the center-of-mass frame. One may worry that this is still a ``tensor product factorization in disguise'', because the kinematical Hilbert space can be refactored into the center-of-mass and relational degrees of freedom. However, our approach allows us to transform between QRFs that are associated to \emph{different factorizations of the kinematical Hilbert space}, which goes beyond earlier formulations of the PN framework. Furthermore, it is clear that our approach via covariant instruments applies in principle much more generally, even in cases where no associated Hilbert space factorization exists. It would be interesting to study this in more detail, e.g.\ in the toy model case of quantum spin chains where the subsystems of all sites are interacting.

In Example~\ref{ExPWdouble}, we have shown that our approach covers naturally the case of quantum clocks in the presence of frequency superselection sectors. This case has been treated in an ad hoc way in~\cite{hohn_equivalence_2021}, where it was explicitly mentioned that it goes beyond the standard PN framweork. This indicates that our generalization is natural, given that it has been implicitly used in previous works.\\

\textbf{Concluding remarks.} While our work has focused on the PN framework, the main idea to generalize QRFs beyond subsystems to covariant instruments applies more generally. An example of this is the relationalization map $\text{\euro}$ of Definition~\ref{DefRel} and Lemma~\ref{LemRel}. It would hence be interesting to study this generalization to instruments also in the context of the operational~\cite{Loveridge2018,Loveridge2017,Carette2025,glowacki_quantum_2024,glowacki_w-algebraic_2026} or extra particle~\cite{CastroRuiz2021,garmier_perspectives_2025} approaches to QRFs. From a more technical point of view, our work has focused on the case of compact Lie groups to avoid technicalities and to be fully mathematically rigorous. A generalization to unimodular Lie groups and beyond, as in the earlier framework of~\cite{Hamette2021}, would be interesting for several physical applications, and could be pursued to varying degrees of mathematical rigor. In particular, we hope that our insights motivate further study of the relation between resource theories, internal QRFs and constraint quantization, and thus contributes to the application of quantum information theory and its operational concepts to quantum field theory and quantum gravity.

\section*{Acknowledgments}
We are grateful to Philipp A.\ H\"ohn for helpful discussions on the relevance of the charge-zero sector, properties of relational clock systems, as well as various aspects of the PN framework and the associated literature. Funded in whole or in part by the Austrian Science Fund (FWF) 10.55776/COE1 (Quantum Science Austria) and the European Union -- NextGenerationEU. For open access purposes, the authors have applied a CC BY public copyright license to any author accepted manuscript version arising from this submission. Artificial intelligence (ChatGPT 5.5) has been used for generating the illustrations, but not for any of the text, arguments, or mathematical results.

\bibliography{references}

\appendix

\section{Proof of \texorpdfstring{Lemma~\ref{LemCompleteCovariance}}{Lemma 2} and its implications}
\label{AppProofLemma1}
In this section, Lemma \ref{LemCompleteCovariance} is shown in detail below, followed by a proof of its implications for compositions, isometries and inverses.
\begin{proof}
    Clearly, $(i)\Rightarrow (ii)\Rightarrow (iii)$ holds, and we will now show that $(iii)\Rightarrow (iv)$. Let $V: X \rightarrow YE$ be an isometry for some external system $E$, such that $T': \rho \mapsto V\rho V^\dagger$ is a covariant purification of $T$. From $\mathcal{U}^{Y}_g\otimes {\rm Id}_E \circ T'=T' \circ \mathcal{U}^X_g$, we must thus have $U^Y_g\otimes \mathds{1}_E V = e^{i\phi(g)}VU^X_g$, where $e^{i\phi(e)}=1$ follows. Furthermore,
    \begin{align*}
        e^{i\phi(gh)}VU^X_{gh}&= U^Y_{gh}\otimes\mathds{1}_E V 
        = \left(U^Y_{g}\otimes\mathds{1}_E \right) \left(U^Y_{h}\otimes\mathds{1}_E\right)V 
        =e^{i\phi(h)}\left(U^Y_{g}\otimes\mathds{1}_E\right)VU^X_h
        =e^{i\phi(h)}e^{i\phi(g)}VU^X_gU^X_h \\
        &= e^{i\phi(g)}e^{i\phi(h)}VU^X_{gh}
    \end{align*}
    then shows that $\chi(g):=e^{i\phi(g)}$ is a one-dimensional representation of $\mathcal{G}$. Now, given a basis $\{|j\rangle\}_j$ of $\mathcal{H}_E$, define $K_j:=\mathds{1}_Y \otimes \langle j |_EV$ with $K_j:X \rightarrow Y$. Due to the form of
    \begin{align*}
        T(\rho) = \operatorname{Tr}_E\left(V\rho V^\dagger \right) = \sum_j \mathds{1}_Y \otimes \langle j |_EV \rho V^\dagger\mathds{1}_Y \otimes | j \rangle_E=\sum_j K_j \rho K^\dagger_j,
    \end{align*}
    we see that $\{K_j\}_j$ is a valid Kraus representation of $T$, for which
    \begin{align*}
        U^Y_gK_j = \left(\mathds{1}_Y \otimes \langle j|_E\right) \left(U^Y_g \otimes \mathds{1}_E\right) V = \mathds{1}_Y \otimes \langle j|_E \chi(g)VU^X_g = \chi(g)K_j U^X_g
    \end{align*}
    follows immediately.
    
    Let us now show that $(iv)\Rightarrow (v)$. Let $\{K_j\}_j$ be some Kraus representation that transforms as required. Given a minimal Kraus representation $\{\tilde K_i\}_i$, we can write any other Kraus representation $\{K'_l\}_l$ as $K'_l=\sum_iu'_{li}\tilde K_i$, where $u'_{li}$ with $i=1,...,d_i$ and $l=1,...,d_l \geq d_i$ are the entries of an isometry $u'$ such that $(u')^\dagger u'=\mathds{1}_{d_i}$ \cite[Theorem 8.2]{nielsen_chuang_2010}. In particular, with $K_j=\sum_iu_{ji}\tilde K_i$ for an isometry $u$, we obtain
    \begin{align*}
        U^Y_g K'_l = U^Y_g\sum_{i}u'_{li}\tilde K_i =U^Y_g\sum_{i,j}u'_{li}\bar u_{ij}K_j=\chi(g)\sum_{i,j}u'_{li}\bar u_{ij}K_jU^X_g=\chi(g)K_l'U^X_g,
    \end{align*}
    for any Kraus representation of $T$.

    Finally, let us show that $(v)\Rightarrow (i)$. Consider a minimal Kraus representation $\{K_j\}_j$ of the channel $T$. Due to $(v)$, we have $\chi(g)K_j U^X_g = U^Y_gK_j$. Given any extension $T'$ on an arbitrary external system $E$ with orthonormal basis states $\{|j\rangle_E\}_j$, let a Kraus representation of $T'$ be given by operators $\{K^E_i:X \rightarrow YE\}_i$. By defining $\bar K_{i,j}:=\mathds{1}_Y\otimes \langle j |_E K^{E}_i$,
    \begin{align*}
        T(\rho) =\operatorname{Tr}_E(T'(\rho)) 
        = \operatorname{Tr}_E\left(\sum_iK^{E}_i\rho (K^E_i)^\dagger\right) = \sum_{i,j} \mathds{1}_Y \otimes \langle j |_E K_i^{E}\rho (K^E_i)^\dagger \mathds{1}_Y \otimes |j \rangle_E = \sum_{i,j} \bar K_{i,j}\rho (\bar K_{i,j})^\dagger
    \end{align*}
    shows that $\{\bar K_{i,j}\}_{i,j}$ is a Kraus representation of $T$. As before, since there exists an isometry $u$ such that $\bar K_{i,j}=\sum_{l}u_{(i,j),l}K_l$ holds, we obtain $\bar K_{i,j}U^X_g=\chi(g)U^Y_g\bar K_{i,j}$, and thus
    \begin{align*}
        K_i^EU^X_g &= \sum_j \mathds{1}_Y \otimes |j\rangle \langle j|_EK_i^E U^X_g
        =\sum_j \left(\mathds{1}_Y \otimes |j\rangle_E \right)\bar K_{i,j}U^X_g
        = \sum_j \left(\mathds{1}_Y \otimes |j\rangle_E \right)\chi(g) U^Y_g \bar K_{i,j} \\
        &= \chi(g)U^Y_g \otimes \mathds{1}_E\sum_j\left(\mathds{1}_Y \otimes |j\rangle_E\right)\bar K_{i,j}
        = \chi(g)U^Y_g \otimes \mathds{1}_E K^E_i,
    \end{align*}
    which immediately implies that $T'$ is covariant. As both the extension $T'$ and the external system $E$ are arbitrary, we conclude that $T$ is completely covariant.
\end{proof}

\begin{lemma}
\label{app:CompInvLemma}
    Let $T:X \rightarrow Y$ and $T': Y \rightarrow Z$ be completely covariant quantum operations. Then the composition $T' \circ T$ and the adjoint $T^\dagger$ are completely covariant quantum operations. Furthermore, every isometric quantum operation $T:X\to Y$ which is covariant is also completely covariant, and its inverse $T^{-1}:{\rm im}(T)\to X$ is completely covariant too.
\end{lemma}

\begin{proof}
    Let $T:X \rightarrow Y$ and $T': Y \rightarrow Z$ be completely covariant quantum operations with associated representations $U^X_g,U^Y_g,U^Z_g$ on the input and output spaces and with Kraus representations $\{K_j\}_j$ and $\{K'_i\}_i$ respectively. Then we have $(T' \circ T)(\rho)=\sum_{ij}K'_iK_j\rho K^\dagger_j K^\dagger_i$ such that $\{M_{ij}:=K'_iK_j\}_{ij}$ defines a Kraus representation of $T'\circ T$. Then
\begin{align*}
    M_{ij}U^X_g = K'_iK_jU^X_g=\chi(g)K'_iU^Y_gK_j=\chi(g)\chi'(g)U^Z_gK'_iK_j=(\chi \circ \chi')(g)U^Z_g M_{ij},
\end{align*}
and thus item $(iv)$ of Lemma \ref{LemCompleteCovariance} holds for the composition $T' \circ T$, since $\chi \circ \chi'$ is a one-dimensional representation as well, showing that the composition is completely covariant too. The adjoint map $T^\dagger:Y\to X$ is given by $T^\dagger(\rho)=\sum_j K_j^\dagger \rho K_j$, and we have $K_j^\dagger U^Y_g=(U^Y_{g^{-1}}K_j)^\dagger=( \chi (g^{-1})K_j U^X_{g^{-1}})^\dagger=\chi(g)U^X_g K_j^\dagger$. Hence $(iv)$ holds for $T^\dagger$, showing that it is completely covariant.

To show the second part of the lemma, let $T:X \rightarrow Y$ be covariant and implemented by an isometry $K$ as $T(\rho)=K\rho K^\dagger$. Since $T$ can be viewed as its own purification on the trivial external system with $\mathrm{dim}(E)=1$, item $(iii)$ of Lemma~\ref{LemCompleteCovariance} immediately implies that $T$ is completely covariant. The left-inverse of $T$ is $T^{-1}: \mathrm{im}(T) \rightarrow X$ with $T^{-1}:=T^\dagger \upharpoonright \mathrm{im}(T)$. To say that $T^{-1}$ is (completely) covariant, its domain of definition must carry a proper representation of $\cg$. Indeed, the space ${\rm im}(T)$ consists of the self-adjoint operators fully supported on ${\rm im}(K)$, and this carries $U_g^Y\upharpoonright {\rm im}(K)$ as a representation of $\cg$, which can be seen as follows. Let $|\varphi\rangle\in{\rm im}(K)$, then there is some $|\psi\rangle\in X$ such that $|\varphi\rangle=K|\psi\rangle$, and so
\[
   U_g^Y|\varphi\rangle=U_g^Y K |\psi\rangle=\overline{\chi(g)} K U_g^X |\psi\rangle\in{\rm im}(K),
\]
where $\chi$ is some one-dimensional representation of $\cg$. Hence, as a restriction of $T^\dagger$ to the operators on an invariant subspace of $U_g^Y$, the map $T^{-1}$ is completely covariant too.
\end{proof}

\section{On the Page-Wootters mechanism}

\subsection{Original construction and extension to interacting case}
\label{sec:PageWoottersReview}

We start with a review of the original Page-Wootters mechanism \cite{PageWootters, Smith_2019}, assuming a tensor-product decomposition
\begin{align*}
\mathcal H_{\mathrm{kin}}=\mathcal H_C\otimes \mathcal H_S
\end{align*}
into a clock Hilbert space and a system Hilbert space, together with a physical state satisfying
\begin{align}
\label{eq:PageWoottersConstraint}
H |\psi\rangle:=(H_C\otimes \mathds{1}_S+\mathds{1}_C\otimes H_S)|\psi\rangle=0
\end{align}
in the non-interacting case. The clock states are generated by the clock Hamiltonian,
\begin{align*}
|t\rangle_C=e^{-itH_C}|0\rangle_C,
\end{align*}
and the system state at time $t$ is defined by conditioning on the clock,
\begin{align*}
|\psi_S(t)\rangle=(\langle t|_C\otimes \mathds{1}_S)|\psi\rangle.
\end{align*}
The central point is that the global state may be timeless or stationary while the conditioned subsystem state behaves dynamically with respect to internal clock readings. We also require the clock states to form a resolution of the identity: $\int dt\ket{t}\bra{t}_C=n\mathds{1}_C$. In our examples we may assume the group to be $\mathcal{G}=\mathrm{U}(1)$, and the constant $n$ to be absorbed into the Haar measure. 

The Page-Wootters reduction maps project onto the clock subsystem as \begin{align*}
R^{(t)}_C=\langle t|_C\otimes \mathds{1}_S,
\end{align*}
and we shall see that they are exactly the reduction maps of the perspective-neutral framework \cite{Hamette2021}.

It is also useful to note that one does not need to eliminate the clock factor, as it is done for our generalized framework that goes beyond the subsystem structure. Instead of the reduced state $|\psi_S(t)\rangle$, we may keep the full perspectival state
\begin{align}
\label{eq:DefFullPerspState}
|\Phi(t)\rangle=(|t\rangle\langle t|_C\otimes \mathds{1}_S)|\psi\rangle,
\end{align}
resulting in
\begin{align*}
|\Phi(t)\rangle=|t\rangle_C\otimes |\psi_S(t)\rangle.
\end{align*}

From here we can derive the ordinary Schrödinger evolution in two ways. First, if we eliminate the clock and use
\begin{align}
\label{eq:DiffCohState}
|\psi_S(t)\rangle=(\langle t|_C\otimes \mathds{1}_S)|\psi\rangle,
\end{align}
then differentiating and using
\begin{align*}
-i\frac{d}{dt}\langle t|_C=\langle t|_CH_C
\end{align*}
gives
\begin{align*}
i\frac{d}{dt}|\psi_S(t)\rangle
=
-(\langle t|_CH_C\otimes \mathds{1}_S)|\psi\rangle.
\end{align*}
With the constraint (\ref{eq:PageWoottersConstraint}) we obtain
\begin{align}
\label{eq:SchrEvol}
i\frac{d}{dt}|\psi_S(t)\rangle
=
H_S|\psi_S(t)\rangle.
\end{align}
Therefore, the reduced perspectival state evolves with the time-independent system Hamiltonian. This is the basic Page-Wootters result in the non-interacting case. Second, if instead one keeps the clock factor, then we can define
\begin{align*}
|\Phi(t)\rangle:=|t\rangle_C\otimes |\psi_S(t)\rangle.
\end{align*}
Differentiating yields
\begin{align*}
    i\frac{d}{dt}|\Phi(t)\rangle &= i\frac{d}{dt}(|t\rangle\langle t|_C\otimes \mathds{1}_S)|\psi\rangle \\
    &= i\frac{d}{dt}(|t\rangle_C \otimes \mathds{1}_S)(\langle t|_C\otimes \mathds{1}_S)|\psi\rangle)+i(|t\rangle_C \otimes \mathds{1}_S)\frac{d}{dt}(\langle t|_C\otimes \mathds{1}_S)|\psi\rangle) \\
    &=H_C\otimes \mathds{1}_S |\Phi(t)\rangle + \mathds{1}_C\otimes H_S|\Phi(t)\rangle \\
    &= H|\Phi(t)\rangle,
\end{align*}
by using equations (\ref{eq:DefFullPerspState}), (\ref{eq:DiffCohState}) and (\ref{eq:SchrEvol}). Note that the full perspectival state $|\Phi(t)\rangle$ evolves with respect to the full Hamiltonian $H$. This is expected since we can write $|t\rangle\langle t|_C\otimes \mathds{1}_S=e^{-itH_C}|0\rangle\langle 0|_Ce^{itH_C}\otimes \mathds{1}_S=e^{-itH}(|0\rangle\langle 0|_C \otimes \mathds{1}_S)e^{itH}$, showing that we can equally view the clock states as generated by $H$. This is necessarily true whenever we have no interaction term between the clock and the system. In that case, the clock and the physical Hilbert space are generated by the same constraint $H$, showing that the Page-Wootters reduction of \eqref{eq:DiffCohState} really coincides with the frame reduction of the perspective-neutral approach. It is exactly this feature which generally breaks down, once an interaction term $H_\mathrm{int}$ is introduced and the clock remains the same, as we shall see now.    

Following the work of Smith and Ahmadi \cite{Smith_2019}, the constraint $H$ with interaction term is given by
\begin{align*}
(H_C\otimes \mathds{1}_S+\mathds{1}_C\otimes H_S+H_{\mathrm{int}})|\psi\rangle=0.
\end{align*}
The clock states are still generated solely by $H_C$,
\begin{align*}
|t\rangle_C=e^{-itH_C}|t_0\rangle_C,
\end{align*}
and the conditional system state is still
\begin{align*}
|\psi_S(t)\rangle=(\langle t|_C\otimes \mathds{1}_S)|\psi\rangle.
\end{align*}
As stated, with this the mechanism goes beyond the usual perspective-neutral framework formalism. The physical Hilbert space is now determined by the full interacting constraint, while the clock coherent state system is still generated only by the clock Hamiltonian. These constraints now generate different structures in general, which is neither the case (for the version of the reduction maps) in \cite{Hamette2021}, nor in our generalized framework. 

However, this departure is not a problem for deriving how the resulting evolution deviates from the non-interacting, relational time evolution. In \cite{Smith_2019} it is shown that one still gets a relational evolution law. Differentiating the conditional state exactly as before gives
\begin{align*}
i\frac{d}{dt}|\psi_S(t)\rangle
=
-(\langle t|_CH_C\otimes \mathds{1}_S)|\psi\rangle.
\end{align*}
Using the full constraint, this becomes
\begin{align*}
i\frac{d}{dt}|\psi_S(t)\rangle
=
H_S|\psi_S(t)\rangle
+
(\langle t|_C\otimes \mathds{1}_S)H_{\mathrm{int}}|\psi\rangle.
\end{align*}
Inserting a resolution of the identity in terms of the clock states then yields the modified Schrödinger equation
\begin{align*}
i\frac{d}{dt}|\psi_S(t)\rangle
=
H_S|\psi_S(t)\rangle
+
\int dt'K(t,t')|\psi_S(t')\rangle,
\end{align*}
where
\begin{align*}
K(t,t')=(\langle t|_C\otimes \mathds{1}_S)H_{\mathrm{int}}(|t'\rangle_C \otimes \mathds{1}_S).
\end{align*}
We may also write down the interacting result without eliminating the clock. With
\begin{align*}
|\Phi(t)\rangle=(|t\rangle\langle t|_C\otimes \mathds{1}_S)|\psi\rangle
=
|t\rangle_C\otimes |\psi_S(t)\rangle,
\end{align*}
differentiating gives
\begin{align}
\label{eq:SmithAhmadiFinalForm}
i\frac{d}{dt}|\Phi(t)\rangle
=
(H_C\otimes \mathds{1}_S+\mathds{1}_C\otimes H_S)|\Phi(t)\rangle
+
|t\rangle_C\otimes \int dt'K(t,t')|\psi_S(t')\rangle.
\end{align}
Thus, relational time evolution is now expressed in the form of a time-non-local Schrödinger equation. From the form above we see that this modification comes directly from the interaction term and is not created by projecting out the clock. Projecting out only removes the explicit clock factor and leaves us with the system's part of the same relational evolution. Note that the evolution generated by the time-dependent Hamiltonian in our Examples \ref{ex:PageWoottersContinuum} and \ref{ex:PageWoottersFiniteAbelian} can be recovered as a special case of (\ref{eq:SmithAhmadiFinalForm}) in the $\mathcal{H}_1\otimes \mathcal{H}_2$-factorization for $H_C\otimes \mathds{1}_S=P_1$, $\mathds{1}_C \otimes H_S=P_2$ and $H_\mathrm{int}=K(X_1-X_2)^2$.

Now, our generalized formalism in the context of Examples \ref{ex:PageWoottersContinuum} and \ref{ex:PageWoottersFiniteAbelian} makes it possible to introduce a valid QRF describing a non-local clock system whose dynamics is decoupled despite the presence of certain interaction terms in the Hamiltonian. Thus, we are able to remove the time-dependence of the Hamiltonian and possibly non-local part from the relational evolution and recover the original Page-Wootters construction. While this is the primary contribution added by our formalism, we remark that it also provides the tools to generalize the form of Eq.~\eqref{eq:SmithAhmadiFinalForm} derived in \cite{Smith_2019}, by extending it to more general clock seeds $M_e$ on $\mathcal{H}_C\otimes \mathcal{H}_S$ instead of remaining with the decomposition $|t\rangle\langle t|_C \otimes \mathds{1}_S$.  

\subsection{Non-local clock example in the finite Abelian setting}

In this subsection we will give a treatment of Example \ref{ex:PageWoottersContinuum} in the context of our generalized formalism, which introduced a non-local clock in the center-of-mass frame in order to recover the standard Page-Wootters mechanism. We will give a version of that example in the finite Abelian setting of \cite{KrummHoehnMueller, HoehnKrummMueller}.

\begin{example}[Non-local Page-Wootters clock for finite Abelian groups]
\label{ex:PageWoottersFiniteAbelian}
We now extend the finite Abelian setting of the earlier examples from a pure translation constraint to a Page-Wootters-type
constraint that carries an interaction term between the particles. Two particles live on $\mathcal{H}_{\mathrm{kin}}=\left(\mathbb{C}^{|\mathbb{Z}_n|}\right)^{\otimes 2}\cong\mathbb{C}^n \otimes \mathbb{C}^n=:\mathcal{H}_1 \otimes \mathcal{H}_2$,
with position and momentum operators
\begin{align*}
    X_i = \sum_{g\in \mathbb{Z}_n}g|g\rangle\langle g|_i \otimes \mathds{1}_{\bar i}, \qquad P_i = \sum_{p\in \mathbb{Z}_n}p|p\rangle\langle p|_i\otimes \mathds{1}_{\bar i} \quad \mathrm{with} \quad \ket{p}_i=\frac{1}{\sqrt{n}}\sum_{g\in \mathbb{Z}_n}\bar \chi_p(g)\ket{g}_i,
\end{align*}
where $\chi_p(g)=e^{i2\pi pg/n}$ as defined in Subsection \ref{SubsecAbelian}, and $i \in \{1,2\}$.

We now add to the gauge generator of global translations $P_1+P_2$ a harmonic oscillator interaction-term and obtain
\begin{align*}
    H = P_1+P_2 + K(X_1-X_2)^2,
    \qquad K\in\mathbb{Z}_n
\end{align*}
as our new constraint. Since $(X_1-X_2)^2$ commutes with $P_1+P_2$, this yields a valid
$\mathbb{Z}_n$-action 
\begin{align}
\label{eq:ActionU^K}
    U_{g}^{(K)}|x_1,x_2\rangle &= e^{-i\frac{2\pi}{n}H}|x_1, x_2\rangle \nonumber \\
    &= e^{-i\frac{2\pi g}{n}K(X_1-X_2)^2}e^{-i\frac{2\pi g}{n}(P_1+P_2)}|x_1, x_2\rangle \nonumber \\
    &=
    e^{-i\frac{2\pi g}{n}K(x_1-x_2)^2}
    |x_1+g, x_2+g\rangle
\end{align}
on $\mathcal{H}_{\mathrm{kin}}$.
Setting $K=0$ recovers the action $U_g$ of the previous finite Abelian case.

The states $|d;P\rangle$
\begin{align}
    |d;P\rangle=\frac{1}{\sqrt{n}}\sum_{g\in \mathbb{Z}_n}\bar\chi_P(g)|g,g+d\rangle
\end{align}
with total momentum $P=p_1+p_2$, as defined in Subsection \ref{SubsecAbelian} and
\cite{KrummHoehnMueller,HoehnKrummMueller} ($\mathbf{d}\rightarrow d$ for two particles), are joint eigenstates of $P_1+P_2$ and
$X_1-X_2$. Using
\begin{align*}
    (X_1-X_2)|d;P\rangle = -d|d;P\rangle,
    \qquad
    e^{-i\frac{2\pi g}{n}(P_1+P_2)}|d;P\rangle
    = \chi_P(g)|d;P\rangle,
\end{align*}
one finds
\begin{align*}
    H|d;P\rangle= (-P+Kd^2)|d;P\rangle.
\end{align*}
The kernel of $H$ is therefore
\begin{align*}
    \mathcal{H}_{\mathrm{phys}}
    =
    \mathrm{span}\left\{ |d;Kd^2\rangle
    \big| d\in\mathbb{Z}_n\right\},
\end{align*}
i.e.\ for each relative-distance coordinate $d$, the corresponding $\chi$ is locked to it by $P=Kd^2$. At $K=0$ this collapses to
$\mathcal{H}_{\mathrm{phys}}=\mathrm{span}\{|d;0\rangle\}_d$
as in the non-interacting case.

Now, with the seed
\begin{align*}
    M_e^{(1)}= \sqrt{n}|0\rangle\langle 0|_1\otimes\mathds{1}_2\upharpoonright\mathcal{H}_\mathrm{phys},
\end{align*}
together with $M_g^{(1)}:= U_g^{(K)}M^{(1)}_0= \sqrt{n}|g\rangle\langle g|_1\otimes\mathds{1}_2\upharpoonright\mathcal{H}_\mathrm{phys}$ and $\frac 1 n \sum_g(M_g^{(1)})^\dagger M_g^{(1)}=\mathds{1}_\mathrm{phys}$ representing the frame of the first particle, the perspectival state at step $g$ is, for a general physical
state $|\psi\rangle=\sum_d c_d|d;Kd^2\rangle$,
\begin{align*}
    |\psi(g)\rangle^{(1)}
    = U^{(K)}_gM_0^{(1)}\,|\psi\rangle = \frac{1}{\sqrt{n}}\sum_d c_dU^{(K)}_g |0,d\rangle = \frac{1}{\sqrt{n}}\sum_d c_d\,
    e^{-i\frac{2\pi g}{n}Kd^2}
    |g,d+g\rangle = |g\rangle_1\otimes|\psi_S(g)\rangle_2,
\end{align*}
with system state
\begin{align*}
    |\psi_S(g)\rangle_2=
   \frac{1}{\sqrt{n}}\sum_y c_{y-g}\,
    e^{-i\frac{2\pi g}{n}K(y-g)^2}\,
    |y\rangle_2,
\end{align*}
parametrized by the second particle's position $y:=d+g$.

Comparing this with the expression
\begin{align*}
    |\psi_S(g+1)\rangle_2
    &=
    \frac{1}{\sqrt{n}}\sum_y c_{y-(g+1)}
    e^{-i\frac{2\pi(g+1)}{n}K(y-(g+1))^2}
    |y\rangle_2 \\
    &= e^{-i\frac{2\pi}{n}K((g+1)-X_2)^2}\frac{1}{\sqrt{n}}\sum_y c_{y-(g+1)}
    e^{-i\frac{2\pi g}{n}K(y-(g+1))^2}
    |y\rangle_2 \\
    &= e^{-i\frac{2\pi}{n}K((g+1)-X_2)^2}e^{-i\frac{2\pi}{n} P_2}\frac{1}{\sqrt{n}}\sum_{y} c_{y-(g+1)}
    e^{-i\frac{2\pi g}{n}K(y-(g+1))^2}
    |y-1\rangle_2 \\
    &= e^{-i\frac{2\pi}{n}K((g+1)-X_2)^2}e^{-i\frac{2\pi}{n} P_2}\frac{1}{\sqrt{n}}\sum_{y'} c_{y'-g}
    e^{-i\frac{2\pi g}{n}K(y'-g)^2}
    |y'\rangle_2 \\
    &=e^{-i\frac{2\pi}{n}K((g+1)-X_2)^2}e^{-i\frac{2\pi}{n} P_2}|\psi_S(g)\rangle_2 \\
    &=: U^{\mathrm{eff}}_{S,2}(g+1,g)|\psi_S(g)\rangle_2
\end{align*}

after another time step, we see that we cannot in general write $|\psi_S(g+1)\rangle_2 = U^{\mathrm{eff}}_{S,2}
|\psi_S(g)\rangle_2$ for any fixed unitary $U^{\mathrm{eff}}_{S,2}$
on $\mathcal{H}_2$, which is generated by some $g$-independent Hamiltonian $H_{S,2}^\mathrm{eff}$. The perspectival
system state does evolve, but not under a
time-independent Schrödinger generator. Thus, choosing the first particle as our clock does not result in a decoupling between it and the rest of the configuration.

We now use the same non-local projectors that Example \ref{ex:centOfMass} introduced as the center-of-mass-frame for the $K=0$ case:
\begin{align*}
    M_e^{(\mathrm{cm})}&= \sqrt{n}\sum_{x\in\mathbb{Z}_n}
    |-\lfloor x/2\rfloor,x-\lfloor x/2\rfloor\rangle\langle -\lfloor x/2\rfloor,x-\lfloor x/2\rfloor|\upharpoonright \mathcal{H}_\mathrm{phys}, \\
 M_g^{(\mathrm{cm})}&:= U_g^{(K)}M^{(\mathrm{cm})}_0=\sqrt{n}\sum_{x\in\mathbb{Z}_n}
    |-\lfloor x/2\rfloor+g,x-\lfloor x/2\rfloor+g\rangle\langle -\lfloor x/2\rfloor+g,x-\lfloor x/2\rfloor+g|\upharpoonright \mathcal{H}_\mathrm{phys}, \\
    &\frac 1 n\sum_g(M_g^{(\mathrm{cm})})^\dagger M_g^{(\mathrm{cm})}=\mathds{1}_\mathrm{phys}.
\end{align*}

Given our basis states of $\mathcal{H}_\mathrm{phys}$,
$|d;Kd^2\rangle=\tfrac{1}{\sqrt n}\sum_g e^{-i2\pi Kd^2g/n}
|g,h+g\rangle$, we obtain
\begin{align*}
    M_e^{(\mathrm{cm})}|d;Kd^2\rangle
    =
   \frac{1}{\sqrt{n}}
    e^{+i\frac{2\pi}{n}K \lfloor d/2\rfloor d^2}
    |-\lfloor d/2\rfloor,d-\lfloor d/2\rfloor\rangle,
\end{align*}
and the following for a generic physical state $|\psi\rangle=\sum_d c_d\,
|d;Kd^2\rangle$:
\begin{align*}
    M_e^{(\mathrm{cm})}|\psi\rangle
    =
    \frac{1}{\sqrt{n}}\sum_d c_d\,
    e^{+i\frac{2\pi}{n}K\lfloor d/2\rfloor d^2}|-\lfloor d/2\rfloor,d-\lfloor d/2\rfloor\rangle.
\end{align*}
Evolving as in Eq. \eqref{eq:ActionU^K} and using $(x_1-x_2)^2=d^2$, we find
\begin{align*}
    |\psi(g)\rangle^{(\mathrm{cm})}
    &= U^{(K)}_gM_e^{(\mathrm{cm})}|\psi\rangle \\
    &=  \frac{1}{\sqrt{n}}\sum_d c_d
    e^{+i\frac{2\pi}{n}K\lfloor d/2\rfloor d^2}
    e^{-i\frac{2\pi g}{n}Kd^2}|g-\lfloor d/2\rfloor,g+d-\lfloor d/2\rfloor\rangle.
\end{align*}
Each ket has center-of-mass coordinate $X:=\lfloor (x_1+x_2)/2\rfloor=g$ and
relative coordinate $r:=x_2-x_1=d$. Writing $|g-\lfloor d/2\rfloor\rangle_1 \otimes |g+d-\lfloor d/2\rfloor\rangle_2
=|X=g\rangle_\mathrm{cm}\otimes|r=d\rangle_\mathrm{rel}$ in the center-of-mass/relative-distance factorization yields
\begin{align*}
    |\psi(g)\rangle^{(\mathrm{cm})}
    =
    |X=g\rangle_\mathrm{cm}
    \otimes
    |\psi_S(g)\rangle_{\mathrm{rel}},
\end{align*}
with the system state on the relative coordinate being
\begin{align*}
    |\psi_S(g)\rangle_{\mathrm{rel}}
    =
    \frac{1}{\sqrt{n}}\sum_d c_de^{+i\frac{2\pi}{n}K\lfloor d/2\rfloor d^2}
    e^{-i\frac{2\pi g}{n}Kd^2}
    |d\rangle_{\mathrm{rel}}.
\end{align*}
The one-step difference is therefore
\begin{equation}
\label{eq:cleanSchrödinger}
    |\psi_S(g+1)\rangle_{\mathrm{rel}}
    =
    e^{-i\frac{2\pi}{n}K(X_1-X_2)^2}
    |\psi_S(g)\rangle_{\mathrm{rel}}.
\end{equation}
i.e.\ created by a fixed, $g$-independent unitary
$U_{S,\mathrm{rel}}^{\mathrm{eff}}=e^{-i2\pi H_{S,\mathrm{rel}}^{\mathrm{eff}}/n}$ generated by the
time-independent, effective Hamiltonian
\begin{align*}
    H_{S,\mathrm{rel}}^{\mathrm{eff}}=K(X_1-X_2)^2,
\end{align*}
acting non-trivially on the relative-distance factor alone. Eq.
(\ref{eq:cleanSchrödinger}) is the discrete Schrödinger equation
in the center-of-mass perspective. The center-of-mass frame
therefore constitutes a decoupled clock, while the local seed does not.
\end{example}

\end{document}